\documentclass[sigconf]{acmart}
\usepackage{balance}
\AtBeginDocument{%
	}

\copyrightyear{2023} 
\acmYear{2023} 
\setcopyright{acmlicensed}\acmConference[PODS '23]{Proceedings of the 42nd ACM SIGMOD-SIGACT-SIGAI Symposium on Principles of Database Systems}{June 18--23, 2023}{Seattle, WA, USA}
\acmBooktitle{Proceedings of the 42nd ACM SIGMOD-SIGACT-SIGAI Symposium on Principles of Database Systems (PODS '23), June 18--23, 2023, Seattle, WA, USA}
\acmPrice{15.00}
\acmDOI{10.1145/3584372.3588658}
\acmISBN{979-8-4007-0127-6/23/06}

\usepackage{enumerate}
\usepackage{amsfonts}
\usepackage{booktabs}
\usepackage{siunitx}
\usepackage{multirow}
\usepackage{listings}
\usepackage{commath}
\usepackage{epstopdf}
\usepackage{url}
\usepackage{mathtools}
\usepackage{alltt}
\usepackage{mathrsfs}
\usepackage{bold-extra}

\usepackage[normalem]{ulem}
\usepackage{etoolbox}
\usepackage{stmaryrd}
\usepackage[noend]{algorithmic}

\usepackage{mathtools}
\usepackage{etoolbox}
\usepackage{dsfont}
\usepackage{color}
\usepackage{colortbl}
\usepackage{multirow}
\usepackage[scaled]{helvet}
\usepackage{enumitem}

\usepackage[colorinlistoftodos]{todonotes}
\usepackage[noend]{algorithmic}
\usepackage{booktabs}
\usepackage{capt-of}
\usepackage{xcolor}
\usepackage{multirow}
\usepackage{wrapfig}

\usepackage{caption}
\usepackage[FIGTOPCAP]{subfigure}
\usepackage{float}
\usepackage{tcolorbox}
\tcbuselibrary{listings,breakable}

\usepackage{array}
\usepackage{arydshln}
\usepackage{todonotes}

\definecolor{mygreen}{rgb}{0,0.6,0}
\definecolor{mygray}{rgb}{0.5,0.5,0.5}
\definecolor{mymauve}{rgb}{0.58,0,0.82}

\usepackage{listings}
\lstnewenvironment{VerbatimText}[1][]{
    
    \lstset{fancyvrb=true,basicstyle=\footnotesize,captionpos=b,xleftmargin=2em,#1}
}{}

\newcommand{\E}{{\tt \mathbb{E}}}

\newcommand{\ignore}[1]{}

\newcommand{\batya}[1]{{\texttt{\color{blue} Batya: [{#1}]}}}

\newcommand{\bx}{\mathbf{X}}

\newcommand\mvd{\twoheadrightarrow}

\newcommand{\rel}{{\mathcal{R}}}
\newcommand{\real}{{\mathbb{R}}}

\newcommand{\sql}[1]{\texttt{#1}}

\newcommand{\select}{\texttt{SELECT}}

\newcommand{\groupby}{\texttt{Group By}}

\newcommand{\from}{\texttt{FROM}}

\newcommand*{\rom}[1]{\expandafter\@slowromancap\romannumeral #1@}
\newcommand{\RNum}[1]{\uppercase\expandafter{\romannumeral #1\relax}}

\newtheorem{defn}{Definition}[section]

\newtheorem{lem}[defn]{Lemma}

\newtheorem{prop}[defn]{Proposition}

\newcommand{\proj}[1]{{\Pi}_{#1}}
\newcommand{\sel}[1]{{\sigma}}

\newcommand{\cut}[1]{}
\newcommand{\eat}[1]{}

\newcommand{\defeq}{\stackrel{\text{def}}{=}}
\newcommand{\setof}[2]{\{{#1}\mid{#2}\}}        

\def\set#1{\mathord{\{#1\}}}

\def\MVD{\sigma}

\def\eqdef{\mathrel{\stackrel{\textsf{\tiny def}}{=}}}

\def\J{\mathcal{J}}
\def\V{\mathcal{V}}

\def\D{\mathcal{D}}
\def\e#1{\emph{#1}}
\newenvironment{citedtheorem}[1]
{\begin{theorem}{\it\e{(#1)}}\,\,}
	{\end{theorem}}

\newenvironment{repeatresult}[2]
{\vskip0.5em\par\textsc{#1} #2.\em}
{\vskip1em}

\newenvironment{repproposition}[1]{\begin{repeatresult}{Proposition}{#1}}{\end{repeatresult}}
\newenvironment{reptheorem}[1]{\begin{repeatresult}{Theorem}{#1}}{\end{repeatresult}}
\newenvironment{replemma}[1]{\begin{repeatresult}{Lemma}{#1}}{\end{repeatresult}}

\def\appendix{\par
	\section*{APPENDIX}
	\setcounter{section}{0}
	\setcounter{subsection}{0}
	\def\thesection{\Alph{section}} }

\def\join{\bowtie}

\def\eqdef{\mathrel{\stackrel{\textsf{\tiny def}}{=}}}

\def\J{\mathcal{J}}

\def\e#1{\emph{#1}}

\newcommand{\algname}[1]{{\sf #1}}
\def\myrulewidth{3.25in}
\def\therule{\rule{\myrulewidth}{0.2pt}}

\def\myrulewidthwide{4in}
\def\therulewide{\rule{\myrulewidthwide}{0.2pt}}

\newenvironment{insidecode}[3]
{
	\begin{tabular}{p{\myrulewidth}}
		\multicolumn{1}{c}{\rule{0mm}{3mm}{\bf #3} $\algname{#1}(\mbox{#2})$\vspace{-0.6em}}\\
		\therule\vskip-0.8em\therule
		\vspace{-1em}
		\begin{algorithmic}[1]}
		{\end{algorithmic}
		\vskip-0.4em\therule
\end{tabular}}

\newenvironment{insidecodewide}[3]
{
	\begin{tabular}{p{\myrulewidthwide}}
		\multicolumn{1}{c}{\rule{0mm}{3mm}{\bf #3} $\algname{#1}(\mbox{#2})$\vspace{-0.6em}}\\
		\therulewide\vskip-0.8em\therulewide
		\vspace{-1em}
		\begin{algorithmic}[1]}
		{\end{algorithmic}
		\vskip-0.3em\therulewide
\end{tabular}}

\newcommand{\T}{{\mathcal{T}}}

\newcommand{\JD}{\textsc{JD}}
\newcommand{\AJD}{\textsc{AJD}}

\newcommand{\schema}{\mathbf{S}}

\newcommand{\relation}{R}

\newcommand{\jointreeMapFunction}{\chi}

\def\P{\mathsf{P}}

\newcommand{\key}{\mathrm{key}}
\newcommand{\components}{\mathrm{dep}}
\newcommand{\edges}{\texttt{edges}}
\newcommand{\nodes}{\texttt{nodes}}
\newcommand{\parent}{\texttt{parent}}

\def\MVD{\mathrm{MVD}}

\def\bX{\textbf{X}}
\def\bx{\textbf{x}}
\def\bY{\textbf{Y}}
\def\by{\textbf{y}}

\usepackage{amsthm}
\newtheorem{theorem}{Theorem}[section]
\newtheorem{corollary}{Corollary}[theorem]
\newtheorem{lemma}[theorem]{Lemma}
\newtheorem{proposition}[theorem]{Proposition}

\theoremstyle{definition}
\newtheorem{definition}{Definition}[section]
\newtheorem{example}[defn]{Example}

\newcounter{texexp}

\tcbset{
	texexp/.style={colframe=blue!30, colback=white,
		coltitle=blue!50!black,
		fonttitle=\small\sffamily\bfseries, fontupper=\small, fontlower=\small},
	example/.code 2 args={\refstepcounter{texexp}\label{#2}%
		\pgfkeysalso{texexp,title={Example \thetexexp: #1}}},
}

\def\expectation{\mathbb{E}}

\DeclareMathOperator{\Ent}{\mathrm{Ent}}
\global\long\def\trre[#1,#2]{\overset{{\scriptstyle (#2)}}{#1}}
\global\long\def\d{\mathrm{d}}

\def\Rel{\mathrm{Rel}}

\def\indicator{\boldsymbol{\mathds{I}}}

\DeclareMathOperator*{\argmin}{arg\,min}

\newif\ifFlag
\Flagtrue
\newcommand{\AppRef}[2]{%
  \ifFlag
    #1 
  \else
    #2 
  \fi
}

\DeclareRobustCommand*\cal{\mathcal}
\begin{document}
	
	\title{Quantifying the Loss of Acyclic Join Dependencies}

	\author{Batya Kenig}
	\affiliation{
		\institution{Technion, Israel Institute of Technology}
		\city{Haifa}	
		\country{Israel}	
	}
	\email{batyak@technion.ac.il}

	\author{Nir Weinberger}
	\affiliation{%
		\institution{Technion, Israel Institute of Technology}
		\city{Haifa}	
		\country{Israel}		
	}
	\email{nirwein@technion.ac.il}

	\begin{abstract}
		Acyclic schemes posses known benefits for database design, speeding up queries, and reducing space requirements. An acyclic join dependency (AJD) is lossless with respect to a universal relation if joining the projections associated with the schema results in the original universal relation. An intuitive and standard measure of loss entailed by an AJD is the number of redundant tuples generated by the acyclic join. Recent work has shown that the loss of an AJD can also be characterized by an information-theoretic measure.
		Motivated by the problem of automatically fitting an acyclic schema to a universal relation, we investigate the connection between these two characterizations of loss. We first show that the loss of an AJD is captured using the notion of KL-Divergence. We then show that the KL-divergence can be used to bound the number of redundant tuples. We prove a deterministic lower bound on the percentage of redundant tuples. For an upper bound, we propose a random database model, and establish a high probability bound on the percentage of redundant tuples, which coincides with the lower bound for large databases.
	\end{abstract}
		\maketitle
		
	\section{Introduction}
\eat{
\begin{figure}
	\centering
	\includegraphics[width=0.4\textwidth]{EchocardiogramBoxPlot.pdf}
	\caption{Experiment over the Echo-cardiogram dataset with 13 attributes and 132 tuples. Every boxplot represents a set of acyclic schemas whose $\J$-measure falls within the appropriate range. The width of the boxplot represents the number of acyclic schemas within the range. It is easy to see that the loss in terms of spurious tuples varies widely even between acyclic schemas with similar $\J$-measures.}
	\label{fig:EchocardiogramBoxPlot}
\end{figure}
}
In the traditional approach to database design, the designer has a clear conceptual model of the data, and of the dependencies between the attributes. This information guides the \e{database normalization} process, which leads to a database schema consisting of multiple relation schemas that have the benefit of reduced redundancies, and more efficient querying and updating of the data. 
This approach requires that the data precisely meet the constraints of the model; various data repair techniques~\cite{DBLP:conf/icdt/AfratiK09,DBLP:series/synthesis/2011Bertossi} have been developed to address the case that the data does not meet the constraints of the schema exactly.
Current data management applications are required to handle data that is noisy, erroneous, and inconsistent. The presumption that such data meet a predefined set of constraints is not likely to hold. In many cases, such applications and are willing to tolerate the ``loss'' entailed by an imperfect database schema, and will be content with a database schema that only ``approximately fits'' the data.
Motivated by the task of
schema-discovery for a given dataset, in this work, we investigate different ways of measuring the ``loss'' of an imperfect database schema, and the relationship between these different measures.

Decomposing a relation schema $\Omega$ is the process of breaking the relation scheme into two or more relation schemes $\Omega_1,\dots,\Omega_k$ whose union is $\Omega$. The decomposition is \e{lossless} if, for any relation instance $R$ over $\Omega$, it holds that $R=\proj{\Omega_1}(R)\join \cdots \join \proj{\Omega_k}(R)$. The \e{loss} of a database scheme $\set{\Omega_1,\dots,\Omega_k}$ with respect to a relation instance $R$ over $\Omega$, is the set of tuples $\left(\join_{i=1}^k\proj{\Omega_i}(R)\right)\setminus R$ that are in the join, but not in $R$ (formal definitions in Section~\ref{sec:notations}). We say that such tuples are \e{spurious}.
Normalizing a relation scheme is the process of (losslessly) decomposing it into a database scheme where each of its resulting relational schemes have certain properties. The specific properties imposed on the resulting relational schemes define different \e{normal forms} such as 3NF~\cite{Codd1971FurtherNO}, BCNF~\cite{DBLP:conf/pacific/Codd75},  4NF~\cite{Fagin:1977:MDN:320557.320571}, and 5NF~\cite{10.1145/582095.582120,10.1145/2274576.2274589}.

A data dependency defines a relationship between sets of attributes in a database.
A \e{Join Dependency (JD)} defines a $k$-way decomposition  (where $k\geq 2$) of a relation schema $\Omega$, and is said to hold in a relation instance $R$ over $\Omega$ if the join is lossless with respect to $R$ (formal definitions in Section~\ref{sec:notations}).  Join Dependencies generalize \e{Multivalued Dependencies (MVDs)} that are effectively Join Dependencies where $k=2$, which in turn generalize \e{Functional Dependencies (FDs)}, which are perhaps the most widely studied data dependencies due to their simple and intuitive semantics.

\e{Acyclic Join Dependencies (AJDs)} is a type of JD that is specified by an \e{Acyclic Schema}~\cite{DBLP:conf/stoc/BeeriFMMUY81}. Acyclic schemes have many applications in databases and in machine learning; they enable efficient query evaluation~\cite{Yannakakis:1981:AAD:1286831.1286840}, play a key role in database normalization and design~\cite{Fagin:1977:MDN:320557.320571,DBLP:journals/is/LeveneL03}, and improve the performance of many well-known machine learning algorithms over relational data\eat{ such as ridge linear regression, classification
trees, and regression trees}~\cite{DBLP:conf/sigmod/Khamis0NOS18,DBLP:conf/sigmod/SchleichOC16,DBLP:conf/sigmod/SchleichOK0N19}.

\eat{
Functional dependencies are perhaps the most widely studied data dependencies due to their simple and intuitive semantics, as well as their applications in database design, database repairing, and data cleaning [citations].
Multivalued Dependencies (MVDs) generalize functional dependencies, and express the way by which a relation instance can be decomposed into two smaller relations such that their join is exactly the original relation. Such a decomposition is referred to as lossless. 
MVDs play a key role in database normalization and design. They characterize the fourth normal form, and were extensively studied in the database literature. 
Acyclic Join Dependencies (AJD) generalize MVDs by specifying two or more projections of the original relation, whose join restores the original relation. An AJD is specified by an Acyclic Schema that meets certain syntactical restrictions that lead efficient query evaluation [cite Yannakakis], and many other desirable properties that are exploited in database design, and for speeding up machine learning applications [citations].
}
Consider how we may measure the loss of an acyclic schema $\schema=\set{\Omega_1,\dots,\Omega_k}$ with respect to a given relation instance $R$ over $\Omega$. An intuitive approach, based on the definition of a lossless join, is to simply count the number of spurious tuples generated by the join (i.e., and are not included in the relation instance $R$).
In previous work, Kenig et al.~\cite{DBLP:conf/sigmod/KenigMPSS20} presented an algorithm that discovers ``approximate Acyclic Schemes'' for a given dataset. Building on earlier work by Lee~\cite{DBLP:journals/tse/Lee87,DBLP:journals/tse/Lee87a}, the authors proposed to measure the loss of an acyclic schema, with respect to a given dataset, using an information-theoretic
metric called the \e{$\J$-measure}, formally defined in Section~\ref{sec:notations}. Lee has shown that this information-theoretic measure characterizes lossless AJDs. That is, the $\J$-measure of an acyclic schema with respect to a dataset is $0$ if and only if the AJD defined by this schema is lossless with respect to the dataset~\cite{DBLP:journals/tse/Lee87,DBLP:journals/tse/Lee87a} (i.e., no spurious tuples).
Beyond this characterization, not much is known about the relationship between these two measures of loss.\eat{ measures: number of spurious tuples, and the $\J$-measure.} In fact, \eat{as exemplified in Figure~\ref{fig:EchocardiogramBoxPlot},} the relationship is not necessarily monotonic, and may vary widely even between two acyclic schemas with similar $\J$-measure values~\cite{DBLP:conf/sigmod/KenigMPSS20}. Nevertheless, empirical studies have shown that low values of the $\J$-measure generally lead to acyclic schemas that incur a small number of spurious tuples~\cite{DBLP:conf/sigmod/KenigMPSS20}.
\eat{
But how do we measure the accuracy of this fit ? 
The standard and intuitive approach, based on the definition of a lossless join, is to restrict the number of tuples generated by the join that are not included in the original dataset. We refer to such tuples as \e{spurious tuples}.
}

Our first result is a characterization of the $\J$-measure of an acyclic schema $\schema=\set{\Omega_1,\dots,\Omega_k}$, with respect to a relation instance $R$, as the \e{KL-divergence} between two \e{empirical distributions}; the one associated with $R$, and the one associated with $R'\eqdef \proj{\Omega_1}(R)\join\cdots \join \proj{\Omega_k}(R)$. The \e{empirical distribution} is a standard notion used to associate a multivariate probability distribution with a multiset of tuples, and a formal definition is deferred to Section~\ref{sec:notations}. The KL-divergence is a non-symmetric measure of the similarity  between two probability distributions $P(\Omega)$ and $Q(\Omega)$. It has numerous information-theoretic applications, and can be loosely thought of as a measure of the information lost when $Q(\Omega)$ is used to approximate $P(\Omega)$. Our result that the $\J$-measure is, in fact, the KL-divergence between the empirical distributions associated with the original relation instance and the one implied by the acyclic schema, explains the success of the $\J$-measure for identifying ``approximate acyclic schemas'' in~\cite{DBLP:conf/sigmod/KenigMPSS20}, and how the $\J$-measure characterizes lossless AJDs~\cite{DBLP:journals/tse/Lee87,DBLP:journals/tse/Lee87a}.

With this result at hand, we address the following problem. Given a relation schema $\Omega$, an acyclic schema $\schema=\set{\Omega_1,\dots,\Omega_k}$, and a value $\J \geq 0$, compute the minimum and maximum number of spurious tuples generated by $\schema$ with respect to \e{any} relation instance $R$ over $\Omega$, whose KL-divergence from $R'\eqdef \proj{\Omega_1}(R)\join\cdots \join \proj{\Omega_k}(R)$ is $\J$. To this end, we prove a deterministic lower bound on the number of spurious tuples that depends only on $\J$. We then consider the problem of determining an upper bound on the number of spurious tuples. As it turns out, this problem is more challenging, as a deterministic upper bound does not hold. We thus  propose a \e{random relation model}, in which a relation is drawn uniformly at random from all possible empirical distributions of a given size $N$. We then show that a bound analogous to the deterministic lower bound on the relative number of spurious tuples also holds as an upper bound with two differences: First, it holds with high probability over the random choice of relation (and not with probability $1$), and second, is holds with an additive term, though one which vanishes for asymptotically large relation instances. The proof of this result is fairly complicated, and as discussed in Section \ref{sec: relation between spurious and MI}, requires applications of multiple techniques from information theory \cite{cover2012elements} and concentration of measure \cite{boucheron2013concentration}.

Beyond its theoretical interest, understanding the relationship between the information-theoretic KL-divergence, and the tangible property of loss, as measured by the number of spurious tuples, has practical consequences for the task of discovering acyclic schemas that fit a dataset.  
Currently, the system of~\cite{DBLP:conf/sigmod/KenigMPSS20} can discover acyclic schemas that fit the data well in terms of its $\J$-measure. Understanding how the $\J$-measure relates to the loss in terms of spurious tuples will enable finding acyclic schemas that generate a bounded number of spurious tuples.
This is important for applications that apply factorization as a means of compression, while wishing to maintain the integrity of the data~\cite{DBLP:conf/icdt/OlteanuZ12}. 

To summarize, in this paper we: (1) Show that the $\J$-measure of Lee~\cite{DBLP:journals/tse/Lee87,DBLP:journals/tse/Lee87a} is the KL-divergence between the empirical distribution associated with the original relation and the one induced by the acyclic schema, (2) Prove a general lower bound on the loss (i.e., spurious tuples) in terms of the KL-divergence, and present a simple family of relation instances for which this bound is tight,
and (3) Propose a random relation model and prove an upper bound on the loss, which holds with high probability, and which converges to the lower bound for large relational instances. 
\eat{
However, besides extreme cases, this tells us very little about the quality of these schemas in terms of their loss. In this paper, we show how the $\J$-measure can be used to bound the loss entailed by an AJD. 
That is, given a relation instance $R$ over schema $\Omega$, an acyclic schema $\schema$, and a $\J$-measure, we prove lower and upper bounds on the loss entailed by $\schema$ with respect to $R$. Our bounds are general in the sense that we only assume the given $\J$-measure of $\schema$ with respect to $R$.
}

\section{Background}
\label{sec:notations}
\eat{Table~\ref{table:notations} summarizes the notations in this paper.}
For the sake of consistency, we adopt some of the notation from~\cite{DBLP:conf/sigmod/KenigMPSS20}.
We denote by $[n] = \set{1,\ldots,n}$.  Let $\Omega$ be a set of
variables, also called attributes.  If $X,Y \subseteq \Omega$, then
$XY$ denotes $X \cup Y$.
\eat{
\small{
\begin{table}	
  \centering
  \begin{tabular}{|l|l|} \hline
   $\Omega$ & set of variables (attributes) \\ \hline
   $n=|\Omega|$ & number of variables (attributes) \\ \hline
   $X,Y,A,B,\ldots$ & sets of variables $\subseteq \Omega$ \\ \hline
   $\schema$ & a schema $=\set{\Omega_1, \ldots, \Omega_m}$ \\ \hline
   $X \mvd Y|Z$ & a standard MVD \\ \hline
   $X \mvd Y_1|Y_2|\cdots |Y_m$ & an MVD~\cite{DBLP:conf/sigmod/BeeriFH77} \\ \hline
   $(\T,\chi)$ & a join tree \\ \hline
   $H(X)$ & entropy of a set of variables $X$ \\ \hline
   $H(Y|X), I(Y;Z|X)$ & entropic measures \\ \hline
   $\J(\T,\chi)$ & the entropic measure in Eq.(\ref{eq:JTScore}) \\ \hline
   $\J(\schema)$ & $\J$ of any join tree for $\schema$\\ \hline
   $\J(X \mvd Y_1|\cdots|Y_m)$ & $\J$ of the schema $\set{XY_1,\ldots,XY_m}$\\ \hline
   $\J(X \mvd Y|Z)$ & $=I(Y;Z|X)$ \\ \hline
   $R$ & a relation \\ \hline
   $N=|R|$ & number of tuples \\ \hline
   $R \models \AJD(\schema)$ & $R$ satisfies an acyclic \\
    &  join dependency \\ \hline
   $R \models_\varepsilon \AJD(\schema)$ & $R$ $\varepsilon$-satisfies an acyclic \\
    &  join dependency \\ \hline    
  \end{tabular}  
  \caption{Notations \label{table:notations}}  
\end{table}
}%
}

\normalsize{
\subsection{Data Dependencies}
\label{subsec:ci:ic}
Let $\Omega\eqdef \set{X_1,\dots,X_n}$ denote a set of attributes with domains $\D(X_1),\dots,\D(X_n)$. We denote by $\Rel(\Omega)\eqdef\set{R: R \subseteq \bigtimes_{i=1}^n\D(X_i)}$ the set of all possible relation instances over $\Omega$.
Fix a relation instance $\relation\in \Rel(\Omega)$ of size $N=|\relation|$. For $Y \subseteq \Omega$ we let $\relation[Y]$
denote the projection of $\relation$ onto the attributes $Y$.
A \e{schema} is a set $\schema=\set{\Omega_1,\dots, \Omega_m}$ such
that $\bigcup_{i=1}^m\Omega_i=\Omega$ and  $\Omega_i\not\subseteq
\Omega_j$ for $i\neq j$.

We say that the relation instance $\relation$ satisfies the \e{join dependency} $\JD(\schema)$, and write $\relation \models \JD(\schema)$, if
$\relation = \Join_{i=1}^m\relation[\Omega_i]$. If $R\not\models \JD(\schema)$, then $\schema$ incurs a \e{loss} with respect to $R$, denoted $\rho(R,\schema)$, defined:
\begin{equation}
	\label{def:rho}
	\rho(R,\schema)\eqdef \frac{\left|\Join_{i=1}^m\relation[\Omega_i]\right|-\left|\relation\right|}{\left|\relation\right|}
\end{equation}
We call the set of tuples $\left(\Join_{i=1}^m\relation[\Omega_i]\right)\setminus R$ \e{spurious tuples}.
Clearly, $\relation \models \JD(\schema)$ if and only if $\rho(R,\schema)=0$.

\eat{Let $X,Y_1,\dots,Y_m \subseteq \Omega$.} We say that
$\relation$ satisfies the \e{multivalued dependency} (MVD)
$\phi = X \mvd Y_1|Y_2|\dots|Y_m$ where $m\geq 2$, the $Y_i$s are
pairwise disjoint, and $XY_1\cdots Y_m = \Omega$, if
$\relation=\relation[XY_1]\Join\cdots \Join \relation[XY_m]$, or if the schema $\schema=\set{XY_1,\dots,XY_m}$ is lossless (i.e., $\rho(R,\schema)=0$). 
\eat{
We call
$X$ the \e{key} of the MVD and $\set{Y_1,\dots,Y_m}$ it's
\e{dependents}, denoted $\key(\phi)=X$ and
$\components(\phi)=\set{Y_1,\dots,Y_m}$. Most of the literature
considers only MVDs with $m=2$, which we call here {\em standard MVDs}. Beeri et al.~\cite{DBLP:conf/sigmod/BeeriFH77} noted that
an MVD can concisely encode multiple standard MVDs; for example
$X \mvd A|B|C$ holds iff $X \mvd AB|C$, $X \mvd A|BC$ and
$X \mvd AC|B$ hold.
}
\eat{
We say that
$\relation$ satisfies the \e{multivalued dependency} (MVD)
$\phi = X \mvd Y_1|Y_2|\dots|Y_m$ where $m\geq 2$, the $Y_i$s are
pairwise disjoint, and $XY_1\cdots Y_m = \Omega$, if
$\relation=\relation[XY_1]\Join\cdots \Join \relation[XY_m]$.  We call
$X$ the \e{key} of the MVD and $\set{Y_1,\dots,Y_m}$ it's
\e{dependents}, denoted $\key(\phi)=X$ and
$\components(\phi)=\set{Y_1,\dots,Y_m}$.} \eat{ Most of the literature
considers only MVDs with $m=2$, which we call here {\em standard
  MVDs}.  Beeri et al.~\cite{DBLP:conf/sigmod/BeeriFH77} noted that a
generalized MVD can concisely encode multiple MVDs; for example
$X \mvd A|B|C$ holds iff $X \mvd AB|C$, $X \mvd A|BC$ and
$X \mvd AC|B$ hold.
}
We review the concept of a \e{join tree} from~\cite{Beeri:1983:DAD:2402.322389}:
\begin{definition}\label{def:joinTree}
  A \e{join tree} or \e{junction tree} is a pair $\left(\T,\jointreeMapFunction\right)$
  where $\T$ is an undirected tree,
  \eat{\footnote{We will assume  w.l.o.g. that $\T$ is connected.  For example, the disconnected schema $\schema = \set{AB,CDE}$ can be represented by adding artificially the edge $AB-CDE$, thus connecting the two disconnected components.}}
  and $\jointreeMapFunction$ is a
  function that maps\eat{associates to} each $u \in \nodes(\T)$ to a
  set of variables $\jointreeMapFunction(u)$, called a \e{bag}, such
  that the \e{running intersection} property holds: for
  every variable $X$, the set
  $\setof{u \in \nodes(\T)}{X \in \chi(u)}$ is a connected component
  of $\T$.  We denote by $\chi(\T) \defeq \bigcup_u \chi(u)$, the set
  of variables of the join tree.
\end{definition}
We often denote the join tree as $\T$, dropping $\jointreeMapFunction$
when it is clear from the context. The \e{schema} defined by $\T$ is
$\schema=\set{\Omega_1,\dots,\Omega_m}$, where
$\Omega_1, \ldots, \Omega_m$ are the bags of $\T$.  We call a schema
$\schema$ \e{acyclic} if there exists a join tree whose schema is
$\schema$.  When $\Omega_i\not\subseteq \Omega_j$ for
all $i\neq j$, then the acyclic schema $\schema=\set{\Omega_1,\dots,\Omega_m}$ satisfies $m \leq |\Omega|$~\cite{Beeri:1983:DAD:2402.322389}.  We say that a relation
$\relation$ satisfies the \e{acyclic join dependency} $\schema$, and
denote $\relation \models \AJD(\schema)$, if $\schema$ is acyclic and
$\relation \models \JD(\schema)$.  An MVD $X \mvd Y_1 | \cdots | Y_m$
represents a simple acyclic schema, namely
$\schema = \set{XY_1, XY_2, \ldots, XY_m}$.

Let $\schema=\set{\Omega_1,\dots,\Omega_m}$ be an acyclic schema with
join tree $(\T,\jointreeMapFunction)$.  We associate to every
$(u,v) \in \edges(\T)$ an MVD $\phi_{u,v}$ as follows.  Let $\T_u$ and
$\T_v$ be the two subtrees obtained by removing the edge $(u,v)$.
Then, we denote by
$\phi_{u,v} \defeq \jointreeMapFunction(u) \cap
\jointreeMapFunction(v) \mvd \jointreeMapFunction(\T_u) |
\jointreeMapFunction(\T_v)$.
We call the \e{support of $\T$} the set of $m-1$ MVDs associated with
its edges, in notation
$\MVD(\T) = \setof{\phi_{u,v}}{(u,v) \in \edges(\T)}$.  Beeri et al. have shown that if $\T$
defines the acyclic schema $\schema$, then $R$ satisfies $\AJD(\schema)$ (i.e.,
$\relation \models \AJD(\schema)$) if and only if $R$ satisfies all MVDs in its
support: $\relation \models \phi_{u,v}$ for all
$\phi_{u,v} \in MVD(\T)$~\cite[Thm.  8.8]{Beeri:1983:DAD:2402.322389}. 

\eat{
\begin{example}
  \label{ex:j0} We will illustrate with the running example from
  Sec.~\ref{sec:illustrativeExample}.  The tree in Fig.~\ref{fig:TD}
  is a join tree. Its bags are the ovals labeled $AF$, $ACD$, $ABD$,
  and $BDE$, and it is custom to show the intersection of two bags on
  the connecting edge.
  $\MVD(\T) = \set{BD \mvd E| ACF, AD \mvd CF| BE, A \mvd F|BCDE}$.
\end{example}
}


}

\subsection{Information Theory}
\label{sec:it}
Lee~\cite{DBLP:journals/tse/Lee87,DBLP:journals/tse/Lee87a} gave an
equivalent formulation of functional, multivalued, and acyclic join dependencies in terms of information
measures; we review this briefly here, after a short background on
information theory.


Let $X$ be a random variable with a finite domain $\D$ and probability
mass $P$ (thus, $\sum_{x \in \D} P(x)=1$). Its entropy is:
\begin{equation}\label{eq:entropy}
H(X)\eqdef\sum_{x\in \D}p(x)\log\frac{1}{P(x)}.
\end{equation}
It holds that $H(X) \leq \log  |\D|$, and equality holds if and only if $P$ is
uniform. The definition of entropy naturally generalizes to sets of jointly distributed random variables
$\Omega=\set{X_1,\dots,X_n}$, by defining the function
$H : 2^\Omega \rightarrow \real$ as the entropy of the joint random
variables in the set. For example,
\begin{equation}
H(X_1X_2)=\sum_{x_1 \in \D_1,x_2\in \D_2}
P(x_1,x_2)\log\frac{1}{P(x_1,x_2)}.
\end{equation}
Let $A,B,C \subseteq
\Omega$. The \e{conditional mutual information} $I(A;B\mid C)$ is defined as:
\begin{equation} \label{eq:h:mutual:information}
I(A;B\mid C) \eqdef~H(BC) + H(AC) - H(ABC) - H(C).
\end{equation}
It is known that the conditional independence
$P \models A \perp B \mid C$ (i.e., $A$ is independent of $B$ given
$C$) holds if and only if $I(A;B\mid C)=0$.
When $C=\emptyset$, or if $C$ is a constant (i.e., $H(C)=0$), then the conditional mutual information is reduced to the standard mutual information, and denoted by $I(A;B)$. 

\eat{
In this paper we use only the following two properties of the mutual
information:
\begin{align}
  I(B;C|A) \geq & 0 \label{eq:shannon} \\
  I(B;CD|A) = & I(B;C|A) + I(B;D|AC) \label{eq:ChainRuleMI}
\end{align}
The first inequality follows from monotonicity and submodularity (it
is in fact equivalent to them); the second equality is called the
\e{chain rule}.  All consequences of these two (in)equalities are
called \e{Shannon inequalities}; for example, monotonicity
$H(AB) \geq H(A)$ is a Shannon inequality because it follows from
\eqref{eq:shannon} by setting $B=C$.
}
%
%

Let $\bX\eqdef\set{X_1,\dots,X_n}$ a set of discrete random variables, and let $P(\bX)$ and $Q(\bX)$ denote discrete probability distributions. The \e{KL-divergence} between $P$ and $Q$ is:
\begin{equation}
	\label{eq:KLD}
	D_{KL}(P\mid\mid Q)\eqdef  \expectation_{P}\log\left(\frac{P(\bX)}{Q(\bX)}\right)=\sum_{\bx \in \D(\bX)}P(\bx)\log\left(\frac{P(\bx)}{Q(\bx)}\right)
\end{equation}
For any pair of probability distributions $P,Q$ over the same probability space, it holds that $D_{KL}(P||Q)\geq 0$, with equality if and only if $P=Q$. It is an easy observation that  

\begin{equation}
  I(A;B\mid C)=D_{KL}(P(ABC)\mid\mid P(A|C)P(B|C)P(C))  
\end{equation}
for any probability distribution $P$.

Let $R$ be a multiset of tuples over the attribute set $\Omega=\set{X_1,\dots,X_n}$, and let $|R|=N$. The \e{empirical distribution} associated with $R$ is the multivariate probability distribution over $\D_1\times \cdots \times \D_n$ that assigns a probability of $P(t)\eqdef \frac{K}{N}$ to every tuple $t$ in $R$ with multiplicity $K$. 
When $R$ is a relation instance, and hence a set of $N$ tuples, its empirical distribution is the uniform distribution
over its tuples, i.e. $P(t) {=} 1/N$ for all $ t{\in} R$, and so its entropy is $H(\Omega)=\log N$.  
For a relation instance $R\in \Rel(\Omega)$, we let $P_R$ denote the empirical distribution over $R$. For $\alpha \subseteq [n]$, we denote
by $X_\alpha$ the set of variables $X_i, i \in \alpha$, and denote by
$\relation(X_\alpha{=}x_\alpha)$ the subset of tuples $t \in R$ where
$t[X_\alpha]{=}x_\alpha$, for fixed values $x_\alpha$.  
When $P_R$ is the empirical distribution over $R$ then 
the marginal probability is
$P_R(X_\alpha{=}x_\alpha){=}\frac{|\relation(X_\alpha{=}x_\alpha)|}{N}$.
\eat{
,
and therefore:
\begin{equation}\label{eq:jointEntropy}
H(X_\alpha)\eqdef \log N-\frac{1}{N}\sum_{x_\alpha {\in} \D_\alpha}|\relation(X_\alpha{=}x_\alpha)|\log |\relation(X_\alpha{=}x_\alpha)|
\end{equation}
}
\eat{
The sum above can be computed using a simple SQL query: 
$\select\ X_\alpha, \sql{count(*)}\times\log(\sql{count(*)})\ \from\ R\ \groupby\ X_\alpha$.
}
%


Lee~\cite{DBLP:journals/tse/Lee87,DBLP:journals/tse/Lee87a} formalized
the following connection between database constraints and entropic
measures.  Let $(\T,\chi)$ be a join tree (Definition~\ref{def:joinTree}).  The $\J$-measure is defined as:
%
\begin{equation}\label{eq:JTScore}
  \J(\T,\chi){\eqdef}\sum_{\substack{v\in\\ \nodes(\T)}}H(\jointreeMapFunction(v))-\sum_{\substack{(v_1,v_2)\in\\ \edges(\T)}}H(\jointreeMapFunction(v_1) {\cap} \jointreeMapFunction(v_2))-H(\chi(\T))
\end{equation}
\normalsize{where $H$ is the entropy (see~\eqref{eq:entropy}) taken over the empirical distribution associated with $R$.
	We abbreviate $\J(\T,\chi)$ with $\J(\T)$, or $\J$, when $\T, \chi$
  are clear from the context. Observe that $\J$ depends only on the schema
  $\schema$ defined by the join tree, and not on the tree itself.  For a simple example, consider the MVD $X \mvd U|V|W$ and
  its associated acyclic schema $\set{XU, XV, XW}$.  For
  the join tree $XU-XV-XW$, it holds that
  $\J = H(XU)+H(XV)+H(XW)-2H(X)-H(XUVW)$.  Another join tree is
  $XU-XW-XV$, and $\J$ is the same.  Therefore, if $\schema$ is
  acyclic, then we write $\J(\schema)$ to denote $\J(\T)$ for any join
  tree of $\schema$.  When $\schema=\set{XZ,XY}$, then the $\J$-measure reduces to the conditional mutual information (see~\eqref{eq:h:mutual:information}), to wit $\J(\schema)=I(Z;Y|X)$.
  \eat{
  We denote by
  $\J(X \mvd Y_1|\cdots |Y_m) \defeq
  H(XY_1)+\cdots+H(XY_m)-(m-1)H(X)-H(XY_1\cdots Y_m)$
  for any sets of variables $X, Y_1, \ldots, Y_m$ where
  $Y_1, \ldots, Y_m$ are pairwise disjoint, even when $XY_1\cdots Y_m$
  is not necessarily $\Omega$.  When $m=2$, then
  $J(X \mvd Y|Z) = I(Y;Z|X)$.} \eat{ Lee proved the following: }

\def\AcyclicCharacterizationLee{
Let $\schema$ be an acyclic schema over $\Omega$, and let $R$ be a relation instance over $\Omega$. 
 Then $R \models \AJD(\schema)$ if and only if $\J(\schema)=0$.
}
\begin{citedtheorem}{\cite{DBLP:journals/tse/Lee87a}}\label{thm:AcyclicCharacterizationLee}
\AcyclicCharacterizationLee
\end{citedtheorem}

In the particular case of a standard MVD, Lee's result implies that $\relation \models X \mvd Y|Z$ if and only if $I(Y;Z|X)=0$ in the empirical distribution associated with $R$.

\subsection{MVDs and Acyclic Join Dependencies}
Beeri et al.~\cite{Beeri:1983:DAD:2402.322389} have shown that an acyclic join dependency defined by the acyclic schema $\schema$, over $m$ relation schemas, is equivalent to $m-1$ MVDs (called its support).
In~\cite{DBLP:conf/sigmod/KenigMPSS20}, this characterization was generalized as follows. 
Let $(\T,\jointreeMapFunction)$ be the join tree corresponding to the acyclic schema $\schema$. 
Root the tree $\T$ at node $u_1$, orient the tree accordingly, and let $u_1, \ldots, u_m$ be a
depth-first enumeration of $\nodes(\T)$.  Thus, $u_1$ is the root, and
for every $i > 1$, $\parent(u_i)$ is some node $u_j$ with $j < i$. For every $i$, we define $\Omega_i \defeq \chi(u_i)$,
$\Omega_{i:j} \defeq \bigcup_{\ell=i,j} \Omega_\ell$, and
$\Delta_i \defeq \chi(\parent(u_i)) \cap \chi(u_i)$. By the running
intersection property (Definition~\ref{def:joinTree}) it holds that $\Delta_i=\Omega_{1:(i-1)}\cap\Omega_i$. }
\begin{citedtheorem}{\cite{DBLP:conf/sigmod/KenigMPSS20}}
\label{thm:BoundsOnJ}
	Let $(\T,\jointreeMapFunction)$ be a join tree over variables $\jointreeMapFunction(\T)=\Omega$, where $\nodes(\T)=\set{u_1,\dots,u_m}$ with corresponding bags $\Omega_i=\jointreeMapFunction(u_i)$.
	Then:
	\begin{equation}
		\label{eq:BeeriGen}
		\max_{i\in[2,m]}I(\Omega_{1,(i-1)};\Omega_{i:m}|\Delta_i)\leq \J(\T,\jointreeMapFunction)\leq \sum_{i=2}^mI(\Omega_{1,(i-1)};\Omega_{i:m}|\Delta_i)
	\end{equation}
\eat{	where $\Omega_{i:j} \defeq \bigcup_{\ell=i,j} \Omega_\ell$ and $\Delta_i=\Omega_{1:(i-1)}\cap\Omega_i$.}
\end{citedtheorem}
Since the support of $(\T,\jointreeMapFunction)$ are precisely the MVDs 
\begin{equation}
\set{\Delta_i\mvd \Omega_{1:(i-1)}|\Omega_{i:m}}_{i\in [2,m]},    
\end{equation}
then~\eqref{eq:BeeriGen} generalizes the result of Beeri et al.~\cite{Beeri:1983:DAD:2402.322389}.  
\begin{definition}[models $\T$]
\label{def:modelsT}
We say that a relation instance $R$ over attributes $\Omega$ \e{models} the join tree $(\T,\jointreeMapFunction)$, denoted $R\models \T$ if $I(\Omega_{1,(i-1)};\Omega_{i:m}|\Delta_i)=0$ for every $i\in [2,m]$. \eat{Equivalently, $\J(\T,\jointreeMapFunction)=0$.}
\end{definition}

\addtolength{\abovedisplayskip}{-1.5pt}
\addtolength{\belowdisplayskip}{-1.5pt}

\section{Characterizing Acyclic Schemas with KL-Divergence}
\label{sec:JMeasureKL}
Let $\bX\eqdef \set{X_1,\dots, X_n}$ be a set of discrete random variables over domains $\D(X_i)$, and let $P(X_1,\dots,X_n)$ be a joint probability distribution over $\bX$.
Let $\bx \in \D(X_1)\times \dots \times \D(X_n)$. For a subset $\bY \subset \bX$, we denote by $\bx[\bY]$ the assignment $\bx$ restricted to the variables $\bY$. We denote by $P[\bY]$ the marginal probability distribution over $\bY$. 

Let $(\T,\chi)$ be a join tree where $\nodes(\T)=\set{u_1,\dots,u_m}$. Let $\schema=\set{\Omega_1,\dots,\Omega_m}$ be the acyclic schema associated with $(\T,\chi)$ over the variables $\Omega=\chi(\T)$, where $\Omega_i=\jointreeMapFunction(u_i)$. 
\def\pTProposition{
	Let $P(X_1,\dots,X_n)$ be any joint probability distribution over $n$ variables, and let $(\T,\chi)$ be a join tree where $\chi(\T)=\set{X_1,\dots,X_n}$. Then $P\models \T$ (Definition~\ref{def:modelsT}) if and only if $P=P_\T$ where:
	\begin{equation}
		\label{eq:Pt}
		P_\T(x_1,\dots,x_n)\eqdef \frac{\prod_{i=1}^mP[ \Omega_i](\boldsymbol{\bx}[ \Omega_i])}{\prod_{i=1}^{m-1}P[ \Delta_i](\boldsymbol{\bx}[\Delta_i])}
	\end{equation}
	where $P[ \Omega_i]$ ($P[ \Delta_i]$) denote the marginal probabilities over $\Omega_i$ ($\Delta_i$).\eat{ We say that a probability distribution $P$ \e{models} the join tree $\T$ denoted $P\models \T$ if $P=P_\T$.}
}
\begin{proposition}
\label{prop:PT}
\pTProposition
\end{proposition}
\AppRef{The proof of Proposition~\ref{prop:PT} is deferred to Appendix \ref{sec:ProofsJMeasureKL}.}{The proof of Proposition~\ref{prop:PT} appears in \cite{kenig2022quantifying}.}It follows from Definition~\ref{def:modelsT}, and a simple induction on the number of nodes in $\T$.

In this section, we refine the statement of Proposition \ref{prop:PT} and prove the following variational representation.
\def\JKLThm{
	For any joint probability distribution $P(X_1,\dots,X_n)$ and any join tree $(\T,\chi)$  with $\chi(\T)=\set{X_1,\dots,X_n}$ it holds that 
	\begin{equation}
		\J(\T)=\min_{Q\models\T}D_{KL}(P||Q)=D_{KL}(P||P_\T)
	\end{equation}
\eat{
	where $P_\T\eqdef \frac{\prod_{i=1}^mP[ \Omega_i](\boldsymbol{\bx}[ \Omega_i])}{\prod_{i=1}^{m-1}P[ \Delta_i](\boldsymbol{\bx}[\Delta_i])}$ (see~\eqref{eq:Pt}). In particular, $\J(\T)=\min_{Q\models\T}D_{KL}(P||Q)$.
}
}
\begin{theorem}
	\label{prop:JMeasureKL}
	\JKLThm
\end{theorem}
In words, this theorem states that when
the join tree $(\T,\jointreeMapFunction)$ is given, then out of all probability distributions $Q$ over $\Omega=\set{X_1,\dots,X_n}$ that model $\T$ (see~\eqref{eq:Pt}), the one closest to $P$ in terms of KL-Divergence, is $P_{\T}$. Importantly, this KL-divergence is precisely $\J(\T)$ (i.e., $\J(\T)=D_{KL}(P||P_\T)$).
While the Theorem holds for all probability distributions $P$, a special case is when 
$P$ is the empirical distribution associated with relation $R$. The proof of Theorem \ref{prop:JMeasureKL} follows from the following two lemmas, interesting on their own, and is deferred to the complete version of this paper~\cite{kenig2022quantifying}.

\eat{
\begin{definition}
	Let $P(X_1,\dots,X_n)$ be any joint probability distribution over $n$ variables, and let $(\T,\chi)$ be a join tree where $\chi(\T)=\set{X_1,\dots,X_n}$. We define:
	\begin{equation}
		\label{eq:Pt}
		P_\T(x_1,\dots,x_n)\eqdef \frac{\prod_{i=1}^mP[ \Omega_i](\boldsymbol{\bx}[ \Omega_i])}{\prod_{i=1}^{m-1}P[ \Delta_i](\boldsymbol{\bx}[\Delta_i])}
	\end{equation}
where $P[ \Omega_i]$ ($P[ \Delta_i]$) denote the marginal probabilities over $\Omega_i$ ($\Delta_i$). We say that a probability distribution $P$ \e{models} the join tree $\T$ denoted $P\models \T$ if $P=P_\T$.
\end{definition}
}
\def\lemMarginals{
Let $P(X_1,\dots,X_n)$ be a joint probability distribution over $n$ random variables, and let $\T$ be a join tree over $X_1,\dots,X_n$ with bags $\Omega_1,\dots,\Omega_m$. Then $P[\Omega_i]=P_\T[\Omega_i]$ for every $i\in [1,m]$, and  $P[\Delta_i]=P_\T[\Delta_i]$ for every $i\in [1,m-1]$.
}
\begin{lemma}
	\label{lem:marginals}
\lemMarginals
\end{lemma}
The proof of Lemma~\ref{lem:marginals} follows from an easy induction on $m$, the number of nodes in the join tree $\T$, \AppRef{and is deferred to Appendix \ref{sec:ProofsJMeasureKL}.}{and can be found in \cite{kenig2022quantifying}.}
\def\lemmaPt{
	The following holds for any joint probability distribution $P(X_1,\dots,X_n)$, and any join tree $\T$ over variables $X_1,\dots,X_n$:
	\begin{equation}
		\argmin_{Q \models \T}D_{KL}(P||Q)=P_{\T}
	\end{equation}
}
\begin{lemma}
\label{lem:Pt}
\lemmaPt
\end{lemma}
\begin{proof}
	From Lemma~\ref{lem:marginals} we have that, for every $i\in \set{1,\dots,m}$,
    $P_\T[\Omega_i]=P[\Omega_i]$ where,  $\Omega_i=\jointreeMapFunction(u_i)$. Since $\Delta_i\subset \Omega_i$, then $P_\T[\Delta_i]=P[\Delta_i]$. Now, 
	\begin{align}
		\min_{Q \models \T}D_{KL}(P||Q)&=\min_{Q \models \T}\expectation_{P}\left[\log \frac{P(X_1,\dots,X_n)}{Q(X_1,\dots,X_n)}\right]\\	
		&=\min_{Q \models \T}\expectation_{P}\left[\log \frac{P(X_1,\dots,X_n)}{P_\T(X_1,\dots,X_n)}\cdot \frac{P_\T(X_1,\dots,X_n)}{Q(X_1,\dots,X_n)}\right]\\
		\eat{
		&=\min_{Q \models \T}\expectation_{P(\bX)}\left[\log \frac{P(X_1,\dots,X_n)}{P_\T(X_1,\dots,X_n)}\right]\\
		&~~~~~~~~~+\min_{Q \models \T}\expectation_{P}\left[\log \frac{P_\T(X_1,\dots,X_n)}{Q(X_1,\dots,X_n)}\right]\\}
		&=D_{KL}(P||P_\T)+\min_{Q \models \T}\expectation_{P}\left[\log \frac{P_\T(X_1,\dots,X_n)}{Q(X_1,\dots,X_n)}\right] 
	\end{align}
Since the chosen distribution $Q(\bX)$ has no consequence on the first term $D_{KL}(P|| P_\T)$, we take a closer look at the second term, $\min_{Q \models \T}\expectation_{P}\left[\log \frac{P_\T(X_1,\dots,X_n)}{Q(X_1,\dots,X_n)}\right]$. Since $Q \models \T$, then by Proposition~\ref{prop:PT} it holds that
$Q(X_1,\dots,X_n)=\frac{\prod_{i=1}^mQ[\Omega_i](\bX[\Omega_i])}{\prod_{i=1}^{m-1}Q[\Delta_i](\bX[\Delta_i])}$.
Hence, in what follows, we refer to $Q$ as $Q_{\T}$. In the remainder of the proof we show that:
\begin{align}
 \expectation_{P}\left[\log \frac{P_\T(X_1,\dots,X_n)}{Q^\T(X_1,\dots,X_n)}\right]&=\expectation_{P_{\T}}\left[\log \frac{P_\T(X_1,\dots,X_n)}{Q_\T(X_1,\dots,X_n)}\right]\label{eq:toProve} \\ &=D_{KL}(P_\T||Q_\T), 
 \end{align}
 where the last equality follows from \eqref{eq:KLD}.
 Since $D_{KL}(P_\T||Q_\T)\geq 0$, with equality if and only if $P_\T=Q_\T$, then choosing $Q_\T$ to be $P_\T$ minimizes $D_{KL}(P||Q)$, thus proving the claim. The remainder of the proof, proving~\eqref{eq:toProve}, follows from Lemma~\ref{lem:marginals} which states that $P[\Omega_i]=P_\T[\Omega_i]$ for every $i\in [1,m]$, and $P[\Delta_j]=P_\T[\Delta_j]$ for every $j\in [m-1]$. \AppRef{The proof is quite technical, and hence deferred to Appendix \ref{sec:ProofsJMeasureKL}.
}{The proof is quite technical, and can be found in \cite{kenig2022quantifying}.}
\end{proof}

\eat{
\begin{theorem}
	\label{prop:JMeasureKL}
	\JKLThm
\end{theorem}
}
\eat{
\subsection{Proof of Theorem~\ref{prop:JMeasureKL}}
	From Lemma~\ref{lem:Pt}, we have that $\min_{Q\models\T}D_{KL}(P||Q)=D_{KL}(P||P_\T)$, where $P_\T$ is defined in~\eqref{eq:Pt}. 
	Therefore, we prove that $\J(\T)=D_{KL}(P||P_\T)$.
	\begin{align*}
		&D_{KL}(P||P_\T)=E_{P(\bX)}\left[\log \frac{P(\bX)}{P_\T(\bX)} \right]\\
		&=\sum_{\bx\in \D(\bX)}P(\bx)\log \frac{P(\bx)\prod_{i=1}^{m-1}P[\Delta_i](\bx[\Delta_i])}{\prod_{i=1}^mP[\Omega_i](\bx[\Omega_i])} \\
		&=\sum_{\bx\in \D(\bX)}P(\bx)\log P(\bx)+\sum_{\bx}P(\bx)\log \frac{\prod_{i=1}^{m-1}P[\Delta_i](\bx[\Delta_i])}{\prod_{i=1}^mP[\Omega_i](\bx[\Omega_i])}\\
		&=-H(\bX)+\sum_{\bx}P(\bx)\log \prod_{i=1}^{m-1}P[\Delta_i](\bx[\Delta_i])\\
		&\qquad\qquad\qquad\qquad-\sum_{\bx}P(\bx) \log\prod_{i=1}^mP[\Omega_i](\bx[\Omega_i])\\
		&=-H(\bX)+\sum_{\bx}P(\bx)\sum_{i=1}^{m-1}\log P[\Delta_i](\bx[\Delta_i])\\
		&\qquad\qquad\qquad\qquad-\sum_{\bx}P(\bx) \sum_{i=1}^m\log P[\Omega_i](\bx[\Omega_i])\\
		&=-H(\bX)+\sum_{i=1}^{m-1}\sum_{\by\in \D(\Delta_i)}P[\Delta_i](\by)\log P[\Delta_i](\by)\\
		&\qquad\qquad\qquad\qquad-\sum_{i=1}^m\sum_{\by\in \D(\Omega_i)}P[\Omega_i](\by)\log P[\Omega_i](\by)\\
		\eat{
		&=\sum_{\bx}P(\bx)\log P(\bx)+\sum_{i=1}^{m-1} \sum_{\bx[\Delta_i]}P[\Delta_i](\bx[\Delta_i])\log P[\Delta_i](\bx[ \Delta_i])-\sum_{i=1}^m\sum_{\bx[\Omega_i]}P[\Omega_i](\bx_{[\Omega_i]}) \log P_{\Omega_i}(\bx_{\mid \Omega_i})\\}
		&=-H(\bX)-\Sigma_{i=1}^{m-1}H(\Delta_i)+\Sigma_{i=1}^mH(\Omega_i)&\\
		&=\J(\T)&
	\end{align*}
 }
\section{Spurious tuples: a lower bound based on $\J(\T)$}
Let $P_R$ be the empirical distribution over a relation instance $R$ with $N$ tuples and $n$ attributes $\Omega=\set{X_1,\dots,X_n}$. That is, the probability associated with every record in $R$ is $\frac{1}{N}$. Let $\T$ be any junction tree over the variables $\Omega$, and let $\schema\eqdef\set{\Omega_1,\dots,\Omega_m}$ where $\Omega_i=\jointreeMapFunction(u_i)$. By Theorem~\ref{thm:AcyclicCharacterizationLee} and Theorem~\ref{prop:JMeasureKL}, it holds that $R\models \AJD(\schema)$ if and only if $\J(\T)=D_{KL}(P||P_\T)=0$. \eat{In other words, $\rho(R,\schema)=0$ if and only if $\J(\T)=D_{KL}(P||P_\T)=0$.}
In what follows, given an acyclic schema $\schema$, we provide a lower bound for $\rho(R,\schema)$ that is based on its associated junction tree $\J(\T)$.
\begin{lemma}
	\label{lem:LogRho}
	Let $P$ be the empirical distribution over a relation instance $R$ with $N$ tuples, and let $\schema=\set{\Omega_1,\dots,\Omega_m}$ denote an acyclic schema with junction tree $(\T,\jointreeMapFunction)$. Then:
	\begin{equation}
	\label{eq:lowerBound}
		\J(\T) \leq \log(1+\rho(R,\schema)).
	\end{equation}
	So if $\rho(R,\schema)=0$, then $\J(\T)=0$ as well.
\end{lemma}
\begin{proof}	
 From Theorem~\ref{prop:JMeasureKL} we have that
	\begin{equation}
		\J(\T)=D_{KL}(P||P_{\T}),
	\end{equation}
	where
	\begin{equation}
		P_{\T}=\argmin_{Q \models \T}D_{KL}(P||Q)
	\end{equation}
	Let us define
	\begin{equation}
		Q_{TU}=\argmin_{\substack{Q \models \T,\\ Q \mbox { is uniform}}}D_{KL}(P||Q),
	\end{equation}
	 and verify that such a distribution $Q_{TU}$ always exists: Let $R'\eqdef \join_{i=1}^m\proj{\Omega_i}(R)$. Let $Q_{TU}$ denote the empirical distribution over $R'$. By construction, $Q_{TU}$ is a uniform distribution (i.e., over tuples $R'$), and $Q_{TU}\models \T$.
	 
	 By definition, we have that $|R'|=N(1+\rho(R,\schema))$,
	 where $|R|=N$, and $\rho(R,\schema)$ is the loss of $\schema$ with respect to $R$ (see~\eqref{def:rho}).
	 
	By limiting the minimization region to uniform distributions, the minimum can only increase. Therefore:
	\begin{equation}
		\J(\T)=D_{KL}(P||P_{\T})\leq D_{KL}(P||Q_{TU})
	\end{equation}
	Evaluating the $KL$-divergence term on the right hand side, and using the fact that $Q_{TU}$ is uniform, we get:
	\begin{align}
		D_{KL}(P||Q_{TU})&=\sum_{r\in R}{P(r)\log{\frac{P(r)}{Q_{TU}(r)}}}\\
		&=\sum_{r\in R}{\frac{1}{N}\log{\frac{1/N}{1/(N+N\cdot\rho(R,\schema))}}}\\
		&=\sum_{r\in R}\frac{1}{N}\log{(1+\rho(R,\schema))}\\
		&=\log(1+\rho(R,\schema))		
	\end{align}
Hence, $\J(\T)\leq \log(1+\rho(R,\schema))$, and $\rho(R,\schema)\geq 2^{\J(\T)}-1$
\end{proof}

The following simple example shows that the lower bound of~\eqref{eq:lowerBound} is tight. That is, there exists a family of relation instances $R$, and a schema $\schema$ where $\J(\schema)=\log(1+\rho(R,\schema))$.
\begin{example}
\label{eq:tightness}
Let $\Omega=\set{A,B}$, where $\D(A)=\set{a_1,\dots,a_N}$ and $\D(B)=\set{b_1,\dots,b_N}$ where $\D(A)\cap \D(B)=\emptyset$. 
Let
\begin{equation}
R=\set{t_1=(a_1,b_1),\dots,t_N=(a_N,b_N)}    
\end{equation}
be a relation instance over $\Omega$. By definition, $P_R(t_i)=\frac{1}{N}$. 
Noting that $H(A)=H(B)=H(AB)=\log N$ (see~\eqref{eq:entropy}), we have that $I(A;B)=\log N$ (see~\eqref{eq:h:mutual:information}).
Now, consider the schema $\schema=\set{\set{A},\set{B}}$, and let $R'=\proj{A}(R)\join \proj{B}(R)$. Clearly, $|R'|=N^2$, and $\rho(R,\schema)=N-1$ (see~\eqref{def:rho}). In particular, we have that $\J(\schema)=I(A;B)=\log N=\log(1+\rho(R,\schema))$, and that this holds for every $N\geq 2$.
\eat{
and let $R=\set{t_1=(0,1),t_2=(1,0)}$ be a relation instance over $\Omega$. By definition, $P_R(t_1)=P_R(t_2)=\frac{1}{2}$. Noting that $H(A)=H(B)=H(AB)=\log 2$ (see~\eqref{eq:entropy}), we have that $I(A;B)=\log 2$ (see~\eqref{eq:h:mutual:information}).
Now, consider the schema $\schema=\set{\set{A},\set{B}}$, and let $R'=\proj{A}(R)\join \proj{B}(R)$. Clearly, $|R'|=4$, and $\rho(R,\schema)=1$ (see~\eqref{def:rho}). In particular, we have that $\J(\schema)=I(A;B)=\log 2=\log(1+\rho(R,\schema))$.
}
\end{example}
\eat{
We recall that $\Rel(\Omega)$ is the set of all possible relation instances over $\Omega$, and that for $R\in \Rel(\Omega)$, $P_R$ denotes the empirical distribution over $R$. The following follows immediately from Lemma~\ref{lem:LogRho}.

\begin{theorem}
Let $\schema=\set{\Omega_1,\dots,\Omega_m}$ denote an acyclic schema with junction tree $(\T,\jointreeMapFunction)$, and let $\J(\T)=c\geq 0$ be a given constant. Then, for \e{any} relation instance $R\in \Rel(\Omega)$, it holds that:
\begin{equation}
    \rho_{min}(\schema,c)\eqdef \min_{\substack{R'\in \Rel(\Omega): P_{R'} \models \T,\\
    D_{KL}(P_{R'}||P_R)=c}}\frac{|R'|-|R|}{|R|}\geq e^c-1
\end{equation}
where $e$ is the natural logarithm.
\end{theorem}
}

	\section{Spurious Tuples: An Upper bound based on mutual information }
\label{sec: relation between spurious and MI}
In the previous section, we have shown that given a relation $R$
and an acyclic schema $\schema$ defined by a join tree $\T$, it
holds that $\log(1+\rho(R,\schema))\geq\J(\T)$ (Lemma~\ref{lem:LogRho}). In this section, we derive an \textit{upper}
bound on $\rho(R,\schema)$ in terms of an information-theoretic measure.
To this end, recall that Theorem~\ref{thm:BoundsOnJ} shows that if $\{\Omega_{1},\ldots,\Omega_{m}\}$
are the bags of $\T$, then $\J(\T)\leq\sum_{i=2}^{m}I(\Omega_{1:i-1};\Omega_{i:m}\mid\Delta_{i})$.
That is, $\J(\T)$ is upper bounded by the sum of the conditional
mutual information of the $m-1$ MVDs in the support $\T$, given
by $\Delta_{i}\mvd\Omega_{1:i-1}|\Omega_{i:m}$, for $i\in[2,m]$\footnote{More accurately, the MVD should be written as $\phi_{i}\eqdef\Delta_{i}\mvd(\Omega_{1:i-1}\backslash\Delta_{i})|(\Omega_{i:m}\backslash\Delta_{i})$,
so that the bags are disjoint. However, it can be easily shown, using
the chain rule of the mutual information \cite[Theorem 2.5.2]{cover2012elements},
that $I(\Omega_{1:i-1};\Omega_{i:m}\mid\Delta_{i})=I(\Omega_{1:i-1}\backslash\Delta_{i};\Omega_{i:m}\backslash\Delta_{i}\mid\Delta_{i})$, and so we adopt the simplified notation for MVD.}. In this section, we relate this sum of conditional mutual information of the
$m-1$ MVDs in the support $\T$ to an approximate upper bound on
$\rho(R,\schema)$, which holds with high probability.

To this end, we begin by relating the relative number
of spurious tuples of the relation $R$ of a schema $\schema$, that
is $\rho(R,\schema)$, with the spurious tuples of each of the MVDs
in its support. Concretely, let $\phi_{i}\eqdef\Delta_{i}\mvd\Omega_{1:i-1}|\Omega_{i:m}$
be the $i$th MVD in the support of $\T$. Then, the relative number
of spurious tuples for $\phi_{i}$ is defined as
\begin{equation}
\rho(R,\phi_{i})\eqdef\frac{|\proj{\Omega_{1:i-1}}(R)\join\proj{\Omega_{i:m}}(R)|-|R|}{|R|}.
\end{equation}

\begin{prop}
\label{prop: Spurious tuples schema to MVDs}Let a relation $R$ be
given, and let $\schema$ be an acyclic schema over the attributes
of $R$ with join tree $\T$, whose support are the MVDs $\phi_{i}=\Delta_{i}\mvd\Omega_{1:i-1}|\Omega_{i:m}$,
for $i\in[2,m]$. Then, 
\begin{equation}
\log\left[1+\rho(R,\schema)\right]\leq\sum_{i=2}^{m}\log\left[1+\rho(R,\phi_{i})\right].\label{eq: rho of schema vs the MVDs on its support}
\end{equation}
\end{prop}

\begin{proof}
We prove by induction on the number of MVDs (or nodes) in the schema. Let $m$ be the number nodes (and $m-1$ be the number of MVDs) in the schema. The base
case $m-1\leq 1$ is immediate. Assuming it holds for $m-1<k$, we prove
the claim for $m-1=k$. Let $\T$ be a join tree representing $k$ MVDs (and hence
$k+1$ nodes). Let $u_{k+1}$ be a leaf in this join tree with parent
$p\eqdef\parent(u_{k+1})$. Let $\T'$ be the join tree where nodes
$u_{k+1}$ and $p$ are merged to the node $u'$ where $\Omega(u')=\Omega(u_{k+1})\cup\Omega(p)$.
Hence, by the induction hypothesis 
\begin{equation}
1+\rho(R,\T')\leq\prod_{i=2}^{k}\left[1+\rho(R,\phi_{i})\right].\label{eq: induction hypothesis}
\end{equation}
Now, let $R'=\join_{i=1}^{k}R[\Omega_{i}]$. Consider the MVD $\phi_{k+1}=\Omega(u_{k+1})\cap\Omega(p)\mvd\Omega(u_{k+1})|\Omega_{1,k}$.
Then, $R''\eqdef \proj{\Omega_{1,k}}(R')\join\proj{\Omega_{k+1}}(R')$ and
by the induction hypothesis, $|R''|\leq|R'|\cdot[1+\rho(R,\phi_{k})]$. By (\ref{eq: induction hypothesis}),
\begin{equation}
|R'|\leq|R|\cdot\prod_{i=2}^{k}\left[1+\rho(R,\phi_{i})\right],
\end{equation}
and hence $|R''|\leq|R|\cdot\prod_{i=2}^{k+1}[1+\rho(R,\phi_{i})]$,
which proves claim. 
\end{proof}
Proposition \ref{prop: Spurious tuples schema to MVDs} reduces the
problem of upper bounding $\log[1+\rho(R,\schema)]$ to bounding each
of the terms $\log[1+\rho(R,\phi_i)]$ in \eqref{eq: rho of schema vs the MVDs on its support},
each of them corresponding to the relative number of spurious tuples of
the $m-1$ MVDs in the support of $\T$. Considering an arbitrary
MVD, which we henceforth denote for simplicity by $\phi\eqdef C\mvd A|B$, Lemma~\ref{lem:LogRho} implies the \emph{lower} bound $\log[1+\rho(R,\phi_{i})]\geq I(A;B\mid C)$,
since an MVD is a simple instance of an acyclic schema. However, obtaining
an upper bound on $\log[1+\rho(R,\phi_{i})]$ in terms of $I(A;B\mid C)$
is challenging because the mutual information $I(A;B\mid C)$ varies
wildly for an MVD $\phi$ even when $\rho(R,\phi)$ and the domains
sizes $d_{A},d_{B}$ and $d_{C}$ remain constant (where $d_{A}\eqdef|\proj{A}(R)|$,
and similarly for $d_{B}$ and $d_{C}$). Figure~\ref{fig:MI_vs_rho}
illustrates this phenomenon in the simple case in which $d_{C}=1$
and so $C$ is a degenerated random variable, and $d_{A}=d_{B}$.
In other words, the value $I(A;B\mid C)$ depends on the actual contents
of the relation instance $R$. However, while $I(A;B\mid C)$ might
not be an accurate upper bound to $\log[1+\rho(R,\phi)]$ for an
\emph{arbitrary} relation, it may hold that it is an \emph{approximate}
upper bound for \emph{most} relations. Therefore, we next propose
a \emph{random relation model}, in which the tuples of the relation
$R$ are chosen at random. We then establish an upper bound
on $\log[1+\rho(R,\phi)]$ that holds with high probability over
this randomly chosen relation. 
\begin{defn}[Random relation model]
\label{def: random relation model}Let $\Omega\eqdef\{X_{1},\ldots,X_{n}\}$
be a set of attributes with domains $\ensuremath{\D(X_{1}),\dots,\D(X_{n})}$,
and assume w.l.o.g. that $\D(X_{i})\eqdef[d_{i}]$ for $\{d_{i}\}_{i=1}^{n}\subset\mathbb{N}_{+}$.
Let $N\in\mathbb{N}_{+}$ be given such that $0<N\leq\prod_{i=1}^{n}d_{i}$.
Let $S$ be a set of $N$ tuples chosen uniformly at random from $\bigtimes_{i=1}^{n}[d_{i}]$,
without replacement. Given $S$, we let $P_{S}$ denote the empirical
distribution over $S$: 
\begin{equation}
P_{S}\left[\bigcap_{i=1}^{n}\{X_{i}=\ell_{i}\}\right]=\begin{cases}
\frac{1}{N}, & (\ell_{1},\ell_{2},\ldots,\ell_{n})\in S\\
0, & \text{otherwise}
\end{cases}\label{eq:Ps}
\end{equation}
for any $(\ell_{1},\ell_{2},\ldots,\ell_{n})\in\bigtimes_{i=1}^{n}[d_{i}]$. 
\end{defn}
In other words, in the random relational model $R$ is chosen uniformly at random from
the set of possible relations of size $N$, that is, from the set
$\rel(\Omega)\cap\{R\colon|R|=N\}$. 
The next proposition states that the existence of a high probability bound on the relative number of spurious tuples associated with an arbitrary MVD $\phi$, implies the existence of a high-probability upper bound on the relative number of spurious tuples $\log[1+\rho(R,\schema)]$ associated with 
an acyclic schema $\schema$.
\begin{prop}
\label{prop:general upper bound on relative number of spurious tuples}Let $\epsilon(\phi,N,\delta)\geq0$ where $\phi\eqdef C\mvd A|B$
is an MVD, $\delta\in(0,1)$, and $R$ is a random relation over attributes $ABC$, where $|R|=N$. Let $\schema=\{\Omega_{1},\ldots,\Omega_{m}\}$ be an acyclic schema over the attributes of $R$ with join tree $\T$. If the random relation $R$ satisfies $\log[1+\rho(R,\phi)]\leq I(A;B\mid C)+\epsilon(\phi,N,\delta)$,
with probability larger than $1-\frac{\delta}{m-1}$, for all MVDs $\phi_i\eqdef \Delta_{i}\mvd\Omega_{1:i-1}|\Omega_{i:m}$ in the support of~~~~$\schema$. Then:
\begin{align}
 \log\left[1+\rho(R,\schema)\right] 
& \leq\sum_{i=2}^{m}I(\Omega_{1:i-1};\Omega_{i:m}\mid\Delta_{i})+\epsilon_i\eat{\epsilon\left(\Delta_{i}\mvd\Omega_{1:i-1}|\Omega_{i:m},N,\frac{\delta}{m-1}\right)}\label{eq: general upper bound on spuious tuples given high probability bound on MVDs}
\\
 & \leq(m-1)\J(\T)+\sum_{i=2}^{m}\epsilon_i\eat{\epsilon\left(\Delta_{i}\mvd\Omega_{1:i-1}|\Omega_{i:m},N,\frac{\delta}{m-1}\right)}\label{eq: general upper bound on spuious tuples given high probability bound on MVDs - using J}
\end{align}
with probability $1-\delta$, where $\epsilon_i\eqdef \epsilon\left(\Delta_{i}\mvd\Omega_{1:i-1}|\Omega_{i:m},N,\frac{\delta}{m-1}\right)$.
\end{prop}
\begin{proof}
For the MVD $\phi_i\eqdef \Delta_{i}\mvd\Omega_{1:i-1}|\Omega_{i:m}$, it
holds that 
\begin{equation}
\log[1+\rho(R,\phi)]\leq I(\Omega_{1:i-1};\Omega_{i:m}\mid\Delta_{i})+\epsilon\left(\phi_i,N,\frac{\delta}{m-1}\right)
\end{equation}
with probability larger than $1-\frac{\delta}{m-1}$. Then, \eqref{eq: general upper bound on spuious tuples given high probability bound on MVDs} 
follows from Proposition \ref{prop: Spurious tuples schema to MVDs},
and a union bound over the $m-1$ MVDs $\{\phi_{i}\}_{i=2}^{m}$ in
the support of $\schema$. The bound (\ref{eq: general upper bound on spuious tuples given high probability bound on MVDs - using J})
follows from (\ref{eq:BeeriGen}) in Theorem~\ref{thm:BoundsOnJ}.
\end{proof}
Hence, the problem of deriving an upper bound on $\log[1+\rho(R,\schema)]$,
which holds with high probability, is reduced to the problem of showing that $\log[1+\rho(R,\phi)]\leq I(A;B\mid C)+\epsilon(\phi,N,\delta)$
holds with high probability for an MVD $\phi\eqdef C\mvd A|B$, in the setting of the random relational model, assuming
that the relation size is fixed to $N$.
In other words, it now suffices to prove the probabilistic upper bound for a single MVD in the
random relation
model (Definition~\ref{def: random relation model}). 
\eat{
\begin{defn}[Random relation model]
\label{def: random relation model}Let $\Omega\eqdef\{X_{1},\ldots,X_{n}\}$
be a set of attributes with domains $\ensuremath{\D(X_{1}),\dots,\D(X_{n})}$,
and assume w.l.o.g. that $\D(X_{i})\eqdef[d_{i}]$ for $\{d_{i}\}_{i=1}^{n}\subset\mathbb{N}_{+}$.
Let $N\in\mathbb{N}_{+}$ be given such that $0<N\leq\prod_{i=1}^{n}d_{i}$.
Let $S$ be a set of $N$ tuples chosen uniformly at random from $\bigtimes_{i=1}^{n}[d_{i}]$,
without replacement. Given $S$, we let $P_{S}$ denote the empirical
distribution over $S$: 
\begin{equation}
P_{S}\left[\bigcap_{i=1}^{n}\{X_{i}=\ell_{i}\}\right]=\begin{cases}
\frac{1}{N}, & (\ell_{1},\ell_{2},\ldots,\ell_{n})\in S\\
0, & \text{otherwise}
\end{cases}\label{eq:Ps}
\end{equation}
for any $(\ell_{1},\ell_{2},\ldots,\ell_{n})\in\bigtimes_{i=1}^{n}[d_{i}]$. 
\end{defn}
}

In what follows, we focus on a single MVD, denoted $\phi\eqdef C\mvd A|B$, and where $d_{A},d_{B}$ and $d_{C}$ are the domain sizes of $A$, $B$, and $C$, respectively. \eat{It should be noted that a random
relation from the model of Definition \ref{def: random relation model}
induces the same random model on the MVD $\phi$.} Then, for
any $S\subseteq[d_{A}]\times[d_{B}]\times[d_{C}]$
\begin{equation}
P_{S}\left[A=a,B=b,C=c\right]=\begin{cases}
\frac{1}{N}, & (a,b,c)\in S\\
0, & \text{otherwise}
\end{cases}\label{eq:Ps MVD}
\end{equation}
for any $(a,b,c)\in[d_{A}]\times[d_{B}]\times[d_{C}]$, and the relation
is such that the set $S$ is chosen uniformly at random from $[d_{A}]\times[d_{B}]\times[d_{C}]$
from all possible sets of size $N$. While both the domain sizes $d_{A},d_{B}$
and $d_{C}$ and relation size $N$ are fixed in this model,
the mutual information $I(A;B\mid C)$ is a random variable due to
the random choice of the set $S$. Specifically, the random relation
instance $S\subseteq[d_{A}]\times[d_{B}]\times[d_{C}]$ where $|S|=N$,
is a random variable, and each specific realization $S=s$, defines
a triplet of random variables $A_{s}\eqdef\proj{A}(s)$, $B_{s}\eqdef\proj{B}(s)$
and $C_{s}\eqdef\proj{C}(s)$. Consequently, every such set $s$
defines various information theoretic measures, such as $H(A_{s})$,
$I(A_{s};B_{s}\mid C_{s})$, and so on. Furthermore, a random choice
of $S$ makes these information measures random quantities themselves,
for example, $H(A_{S})$ and $I(A_{S};B_{S}\mid C_{S})$ are random
variables.\eat{ due to the random choice of $S$.} In a similar fashion,
if we let $R_{S}$ denote the random relation defined by $S$, then
$\rho(R_{S},\phi)$ is again a random variable. Our main result regarding
the mutual information of an MVD is as follows:

\begin{theorem}[Confidence bound of the random mutual information of an MVD]
\label{thm: confidence of MI of MVD}Let\eat{ a required probability } $\delta\in(0,1)$. \eat{
be given.} Assume w.l.o.g. that $d_{A}\geq d_{B}$, denote $\overline{d}\eqdef\max\{d_{A},d_{C}\}$
and assume further that 
\begin{equation}
N\geq256d_{A}\overline{d}\log\left(\frac{384\overline{d}}{\delta}\right).\label{eq: assumption on N in MVD theorem}
\end{equation}
Let 
\begin{equation}
\epsilon^{*}\left(\phi,N,\delta\right)\eqdef60\sqrt{\frac{d_{A}\overline{d}\log^{3}\left(\frac{6Nd_{C}}{\delta}\right)}{N}}
\end{equation}
If $R_{S}$ is drawn from the random relation model of Definition
\ref{def: random relation model}, then 
\begin{equation}
\log\left[1+\rho(R_{S},\phi)\right]\leq I(A_{S};B_{S}\mid C_{S})+\epsilon^{*}\left(\phi,N,\delta\right)
\end{equation}
with probability larger than $1-\delta$. 
\end{theorem}

The proof of Theorem \ref{thm: confidence of MI of MVD} is fairly
complicated, and is discussed in detail in Section \ref{subsec:The-proof-of-MVD-theorem}. Theorem \ref{thm: confidence of MI of MVD} shows that,
with high probability, the upper bound $\log\left[1+\rho(R_{S},\phi)\right]\leq I(A_{S};B_{S}\mid C_{S})$
approximately holds, up to an additive factor of 
$\tilde{O}(\sqrt{d_{A}\max\{d_{A},d_{C}\}/N})$,
where the $\tilde{O}(\cdot)$ hides logarithmic terms. This result is suitable for large domain
sizes, and when the number of tuples $N$ is proportional to the
domain sizes. \eat{In principle, the number of spurious tuples $N$ can
be arbitrary in the regime $0\leq N\leq d_{A}d_{B}d_{C}$.} More accurately, the bound
holds whenever $N=\tilde{\Omega}(d_{A}d_{C})$, where $\tilde{\Omega}(\cdot)$
hides logarithmic terms (condition (\ref{eq: assumption on N in MVD theorem})),
which is a mild condition when targeting a low fraction of
spurious tuples. So, when $\delta$ is fixed to some desired reliability,
and the qualifying condition (\ref{eq: assumption on N in MVD theorem})
holds, then the deviation term in the claim of Theorem \ref{thm: confidence of MI of MVD}
is given by $\epsilon^{*}(\phi,N,\delta)=\tilde{O}(\sqrt{\max\{d_{A}^{2},d_{A}d_{C})/N})$.
Hence, when $N$ increases
as $N=\tilde{\omega}(\max\{d_{A}^{2},d_{A}d_{C}\})$, then $\epsilon^{*}(\phi,N,\delta)$
vanishes. For example, if $d_{A}=d_{B}=d_{C}\equiv d$,
then the deviation term is $\epsilon^{*}(\phi,N,\delta)=O(\sqrt{\frac{d^{2}\log^{3}(Nd)}{N}})$.
\eat{If the common domain size $d$ increases, and} and this deviation term vanishes if $N=\omega(d^{2}\log^{3}(Nd))$. As a more concrete example, if
$N=\tfrac{1}{2}d^{3}$, then the deviation term is $\epsilon^{*}(\phi,N,\delta)=O(\sqrt{\frac{\log^{3}(d)}{d}})$
which vanishes at a rather fast rate with increasing $d$. Moreover,
the dependency of the deviation term in $\delta$ is mild, and scales
as $\log^{3/2}(1/\delta)$ which is close to a sub-exponential
dependence.\footnote{For sub-exponential random variables, the dependence on $\delta$
is $\log(\frac{1}{\delta})$ \cite{boucheron2013concentration}.} 
\begin{figure}
\centering{} \includegraphics[scale=0.2]{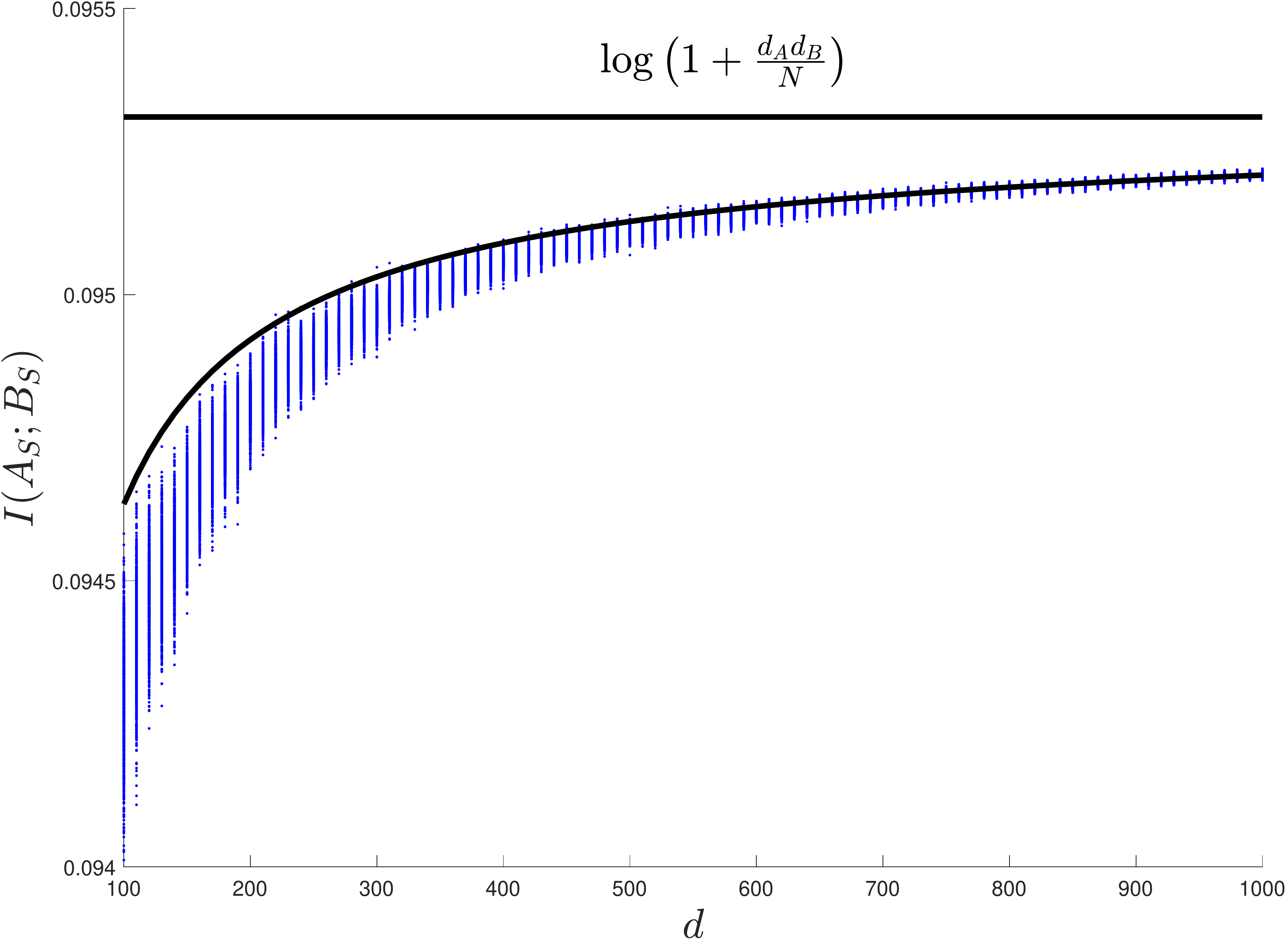}
\caption{Mutual information scattering vs. $\log(1+\rho)$ for $d_{C}=1$ and
$d_{A}=d_{B}=d$. In this experiment, we fixed the percentage of spurious tuples $\rho(R,\schema)$, and accordingly randomly generated $N\eqdef \frac{d_Ad_B}{(1+\rho(R,\schema))}$ tuples from the random relation model (Definition~\ref{def: random relation model}), and plotted the resulting mutual information. As can be seen, as the database grows, the mutual informatio approaches $\log(1+\rho)$.\label{fig:MI_vs_rho}}
\end{figure}
\eat{
\batya{to remove from here to the end of this section. Seems like a repetition of Proposition 5.2}
Combining the general Proposition \ref{prop:general upper bound on relative number of spurious tuples}
with the result of Theorem \ref{thm: confidence of MI of MVD} leads
to the following upper bound:

\begin{theorem}[Confidence bound of the random mutual information of a schema]
\label{thm: confidence of MI of schema}Let $\{X_{1},\ldots X_{n}\}$
be a set of attributes from the domain $\bigtimes_{i=1}^{n}[d_{i}]$.
Let $R_{S}$ be a relation drawn randomly from the random relation
model of Definition \ref{def: random relation model} of size $N$,
let $\T$ be a join tree, and assume that the condition (\ref{eq: assumption on N in MVD theorem})
is satisfied for each of the MVD in the support of $\T$. Then
\begin{multline}
\log\left[1+\rho(R_{S},\schema)\right]\\
\leq\sum_{i=2}^{m}I(\Omega_{1:i-1,S};\Omega_{i:m,S}\mid\Delta_{i,S})+\epsilon^{*}\left(\Delta_{i}\mvd\Omega_{1:i-1}|\Omega_{i:m},N,\frac{\delta}{m-1}\right)\label{eq: high probability bound sum of MI  and logrho}
\end{multline}
with probability larger than $1-\delta$, over the random choice of
the set $S$. 
\end{theorem}

The bound (\ref{eq: high probability bound sum of MI  and logrho})
is the sought approximated high probability bound on the relative
number of spurious tuples. Its proof is a direct and trivial application
of Proposition \ref{prop:general upper bound on relative number of spurious tuples}
and Theorem \ref{thm: confidence of MI of MVD}, and thus not explicitly
given.
}
\subsection{Proof of Theorem \ref{thm: confidence of MI of MVD}: A Confidence
Interval for the Mutual Information\label{subsec:The-proof-of-MVD-theorem}}

In this section, we discuss in detail the proof of Theorem \ref{thm: confidence of MI of MVD}.
We focus on the case in which $C$ is
a degenerate random variable ($d_{C}=1$) since the main components of the proof are already present in this simple case. 
When $d_{C}=1$, the conditional mutual information is reduced to the standard mutual information $I(A_{S};B_{S})$.
In turn, the random set $S$ which determines the random relation
is chosen uniformly at random from all subsets of $[d_{A}]\times[d_{B}]$
of a given size. To avoid confusion with the non-degenerate model,
we denote this size by $\eta$ (rather than by $N$) whenever $d_{C}=1$.
The mutual information can then be decomposed as $I(A_{S};B_{S})=H(A_{S})+H(B_{S})-H(A_{S},B_{S})$
\cite[Section 2.4]{cover2012elements}, and by the definition of the
random model, $R_S=(A_{S},B_{S})$ is distributed uniformly over the possible sets
of size $\eta$. Thus $H(A_{S},B_{S})=\log\eta$ with probability
$1$ (over the choice of $S$), and $I(A_{S};B_{S})=H(A_{S})+H(B_{S})-\log\eta$.
Due to symmetry, the analysis of $H(A_{S})$ and $H(B_{S})$
is analogous, and so we next focus on the former. The main ingredient
of the proof of Theorem~\ref{thm: confidence of MI of MVD} is a confidence
interval for the random entropy $H(A_{S})$, when $A_{S}$ is chosen from a random relation model
similar to the one of Definition \ref{def: random relation model}, albeit 
with a degenerated $C$, that is, $d_{C}=1$. At a later stage, we discuss the
generalization of this reuslt to $d_{C}>1$. The confidence bound on $H(A_{S})$
is as follows:
\begin{theorem}
\label{thm: confidence bound for random entropy}Let $A_{S}$ be drawn
according to the random relation model of Definition \ref{def: random relation model}
with $d_{C}=1$ and $N=\eta$. Assume w.l.o.g. that $d_{A}\geq d_{B}$
and that 
\begin{equation}
\eta\geq128d_{A}\log\left(\frac{128d_{A}}{\delta}\right).\label{eq: condition on eta statement}
\end{equation}
Then, for any probability $\delta\in(0,1)$ it holds that 
\begin{equation}
\log d_{A}\geq H(A_{S})\geq\log d_{A}-20\sqrt{\tfrac{d_{A}\log^{3}\left(\frac{\eta}{\delta}\right)}{\eta}}
\end{equation}
with probability $1-\delta$, over the random choice of the set $S$. 
\end{theorem}

The proof of Theorem \ref{thm: confidence bound for random entropy}
comprises most of the proof of Theorem \ref{thm: confidence of MI of MVD},
and requires a diverse set of mathematical techniques, discussed
in Section \ref{subsec:Confidence-bound-of}. For now,
taking the result of Theorem \ref{thm: confidence bound for random entropy}
as given, a high probability bound on the value of $I(A_{S};B_{S})$
for the random relation model with degenerated $C$ can be obtained
as a simple corollary to Theorem \ref{thm: confidence bound for random entropy},
as follows:
\begin{corollary}
\label{cor: confidence bound for MI}Let $\overline{\rho}=\frac{d_{A}\cdot d_{B}}{\eta}-1$.
Then, under the same assumptions of Theorem \ref{thm: confidence bound for random entropy},
\begin{equation}
I(A_{S};B_{S})\geq\log\left(1+\overline{\rho}\right)-40\sqrt{\tfrac{d_{A}\log^{3}\left(\frac{2\eta}{\delta}\right)}{\eta}},\label{eq: confidence bound for MI}
\end{equation}
with probability $1-\delta$, over the random choice of the set $S$.
\end{corollary}

Corollary \ref{cor: confidence bound for MI} will be used in
the proof of Theorem \ref{thm: confidence of MI of MVD}. Beyond that,
it also reveals the tightness of the mutual information bound in the
simpler setting of a degenerated MVD $(d_{C}=1)$. Indeed, since $A_{S}\subseteq[d_{A}]$
and $B_{S}\subseteq[d_{B}]$ for all realizations of $S$, then $\rho(R_{S},\phi)\leq\frac{d_{A}d_{B}}{\eta}-1=\overline{\rho}$.
Thus, Corollary \ref{cor: confidence bound for MI} implies
\begin{equation}
I(A_{S};B_{S})\geq\log\left[1+\rho(R_{S},\phi)\right]-40\sqrt{\tfrac{d_{A}\log^{3}\left(\frac{2\eta}{\delta}\right)}{\eta}},
\end{equation}
but actually shows the stronger bound (\ref{eq: confidence bound for MI}).\eat{
This stronger bound, in which an upper bound on $\log[1+\rho(R_{S},\phi)]$
is actually upper bounded, instead of the exact value, also essentially
holds for the non-degenerated case of Theorem \ref{thm: confidence of MI of MVD}. 
}
\paragraph{Proof outline of Theorem \ref{thm: confidence of MI of MVD} }

At this point, let us take the results of Theorem \ref{thm: confidence bound for random entropy}
and Corollary \ref{cor: confidence bound for MI} as granted. Then,
the proof of Theorem \ref{thm: confidence of MI of MVD} is essentially
a generalization of the result of Theorem \ref{thm: confidence bound for random entropy}
to the case in which $d_{C}>1$. Let us define $R_{\ell}\eqdef\sigma_{C=\ell}(R)$.
Then, in the random relation model $N_{S}(\ell)=|R_{\ell}|$ is a
random variable, and beyond the randomness in the joint distribution
of $(A_{S},B_{S})$ when  conditioned on any $C=\ell$, there is also randomness
in the number of tuples in the random relation, whenever $C=\ell$.
Hence, the mutual information conditioned on the specific value of
$C=\ell$, to wit, $I(A_{S};B_{S}\mid C_{S}=\ell)$, is drawn from
the random model in Definition \ref{def: random relation model} with
$N$ being replaced by $N_{S}(\ell)$ (the latter being a random variable
due to the random choice of $S$). The result of Corollary \ref{cor: confidence bound for MI},
regarding the mutual information of a pair of random variables $A_{S},B_{S}$,
can then be used conditionally on $C_{S}=\ell$, where $\eta$ is
being replaced by $N_{S}(\ell)$. In order for this result to hold,
the qualifying condition of Corollary \ref{cor: confidence bound for MI},
to wit $N_{S}(\ell)\geq128d_{A}\log(\frac{128d_{A}}{\delta})$, should
hold for all $\ell\in[d_{C}]$. The proof begins by showing that this
condition indeed holds for all $\ell\in[d_{C}]$ with high probability.
\AppRef{This is proved in Lemma \ref{lem: large domain size per letter in c} in Appendix \ref{sec:Proof-of-MVD-Theorem},}
{This is proved in \cite{kenig2022quantifying},}and is based on the fact that $N_{S}(\ell)$ is a hypergeometric random
variable, and on a concentration result by Serfling \cite{serfling1974probability}
for such random variables. The proof then assumes that all the following
holds: (I) For all $\ell\in[d_{C}]$, $N_{S}(\ell)$ is sufficiently
large so that the qualifying condition of Corollary \ref{cor: confidence bound for MI}
holds. (II) For each $\ell\in[d_{C}]$, the confidence bound in Corollary
\ref{cor: confidence bound for MI} holds. (III) $H(C_{S})$ is close
to $\log d_{C}$. 

\AppRef{Specifically, Lemma \ref{lem: large domain size per letter in c} assures}{Specifically, the proof in \cite{kenig2022quantifying} assures}that the
first condition holds with high probability; Corollary \ref{cor: confidence bound for MI}
assures that the second condition holds with high probability; a simple
modification of Theorem \ref{thm: confidence bound for random entropy}
shows that the third condition holds with high probability. By the
union bound, the event in which the set $s\subseteq[d_{A}]\times[d_{B}]\times[d_{C}]$
simultaneously satisfies properties (I), (II) and (III) has high probability.
The proof is completed by considering a set $s\subseteq[d_{A}]\times[d_{B}]\times[d_{C}]$,
which satisfies properties (I), (II) and (III), and relating $\log[1+\rho(R_{s},\phi)]$
to the mutual information. Concretely, an application of the \emph{log sum inequality} \AppRef{(Lemma \ref{lem: log-sum inequality} in Appendix
\ref{sec:Auxiliary-results}),}{\cite[Thm. 2.7.1]{cover2012elements}} shows that \AppRef {(see (\ref{eq: derivation of log-spurious tuples with conditional MI}))}{}
\begin{multline}
\log\left[1+\rho(R_{s},\phi)\right]\leq\log d_{C}-H(C_{s})\\ +\sum_{\ell\in[d_{C}]}\P[C_{s}=\ell]\log\left[1+\overline{\rho}_{s}(\ell)\right],
\end{multline}
which can be bounded by conditional mutual information, and an additional
additive deviation term, utilizing the aforementioned assumption that
$s$ satisfies properties (I), (II) and (III). 
\eat{
\paragraph*{An alternative random model}

Assuming $d_{C}=1$ for simplicity, Definition \ref{def: random relation model}
specifies a relation which is randomly determined by choosing a subset
of size $\eta$ from $[d_{A}]\times[d_{B}]$, uniformly at random
from all possible subsets. Equivalently, $\eta$ points are chosen
from $[d_{A}]\times[d_{B}]$, without replacement. An alternative
random model to the one of Definition \ref{def: random relation model}
is one which allows \emph{replacements.} Evidently, this model allows
for multiple tuples (a point can be chosen more than once), and thus
may be less suited to model a relation to begin with. However, one
may speculate that this model is simpler to analyze, since typically
random model which allow replacements lead to simpler analysis. Nonetheless,
simple numerical simulations show that a result analogous to Corollary
\ref{cor: confidence bound for MI} for a model with replacements
most likely do not hold, especially in the most interesting regime
in which $\overline{\rho}=\frac{d_{A}d_{B}}{N}-1$ is a constant.
The crux of the difference is that in our model $H(A_{S},B_{S})=\log\eta$
holds with probability $1$, yet the corresponding quantity in the
alternative model may not concentrate, and does not seem to be easily
characterized (at least not using the methods we have used).
}
\subsection{Confidence Bound of the Conditional Entropy \label{subsec:Confidence-bound-of}}

In this section, we describe the proof of the confidence interval in Theorem \ref{thm: confidence bound for random entropy}, which is comprised of three main steps on its own: (I) A bound
on the expected value of $H(A_{S})$, which is shown to be asymptotically
close to $\log d_{A}$ under the random relation model. (II) A concentration
result of $H(A_{S})$ to its expected value. (III) A combination of
these bounds. 
In the next two subsections we provide a formal statement of the
first two steps, and outline their proof. \AppRef{The full proof
is deferred to Appendix \ref{sec:Proof-of-MI-Theorem}, along with
the third part (which is more technical in its nature) and completes
the proof of Theorem~\ref{thm: confidence bound for random entropy}.}{The full proof is rather long, and appears in \cite{kenig2022quantifying}.} 

\subsubsection{The expected value of the entropy\label{subsec:The-expected-value}}

In this section, we state a bound on the average mutual information
$\E[H(A_{S})]$ and outline its proof. \eat{This is the first main step
of the proof of Theorem \ref{thm: confidence bound for random entropy}.}
Let us denote, for notational brevity, 
\begin{equation}
C(d)\eqdef\frac{2\log(d)}{\sqrt{d}}.\label{eq: C constant def}
\end{equation}
\begin{prop}[Bounds on the expected entropy]
\label{prop: Expected value of entropy}Assume that $d_{A}\geq d_{B}$
and that $\eta\geq60d_{A}$. If $S$ is chosen uniformly at random
from one of the possible subsets of $[d_{A}]\times[d_{B}]$ of size
$\eta$ then
\begin{equation}
0\leq\log d_{A}-\E[H(A_{S})]\leq C(d_{B}),
\end{equation}
where $C(d)$ is as defined in (\ref{eq: C constant def}). An analogous
result hold for $H(B_{S})$:
\begin{equation}
0\leq\log d_{B}-\E[H(B_{S})]\leq C(d_{A}).
\end{equation}
\end{prop}

We next present the main ideas of the proof of Proposition \ref{prop: Expected value of entropy}.
As a first step, we identify that the expected value $\E[H(A_{S})]$
is, in fact, a conditional entropy $H(A\mid S)$, to wit,\textbf{
\begin{equation}
\E\left[H(A_{S})\right]=\sum_{s}\P[S=s]\cdot H(A_{s})=H(A\mid S).\label{eq: expected MI}
\end{equation}
}The crux of the proof of Proposition \ref{prop: Expected value of entropy}
requires \emph{lower} bounding $H(A\mid S)$. Nonetheless, to illuminate
the challenge in the proof, it is insightful to first note that as
\emph{conditioning reduces entropy} \cite[Theorem 2.6.5]{cover2012elements}, and so 
\begin{equation}
H(A\mid S)\leq H(A).\label{eq: CRE for description}
\end{equation}
Using the symmetry of the distribution of the set $S$, it follows that $H(A)=\log d_{A}$ \AppRef{(see Lemma~\ref{lem:unconditionalentropies} in Appendix \ref{sec:Proof-of-MI-Theorem} for a rigorous proof)}{(see~\cite{kenig2022quantifying} for a rigorous proof)}. So, Proposition \ref{prop: Expected value of entropy}
states that $H(A\mid S)$ is close to its unconditional value
$H(A)$, up to $C(d_{B})$. In other words, we need to show
that the conditioning (over the random variable $S$) only \emph{slightly}
reduces entropy in (\ref{eq: CRE for description}). 

To further delve into the proof of this property, we closely inspect
$H(A\mid S)$. For any $(i,j)\in[d_{A}]\times[d_{B}]$, we define
the random variable $U_{S}(i,j)\eqdef\indicator\{(i,j)\in S\}\}$,
which indicates if the tuple $(i,j)$ is in the random relation $R_{S}$
(see Definition \ref{def: random relation model}). By symmetry, $\P(U_{S}(i,j)=1)=\frac{\eta}{d_{A}d_{B}}$
for all $(i,j)\in[d_{A}]\times[d_{B}]$, and hence $\{U_{S}(i,j)\}_{(i,j)\in[d_{A}]\times[d_{B}]}$
are identically distributed. The values $\{U_{S}(i,j)\}_{(i,j)\in[d_{A}]\times[d_{B}]}$
uniquely determine $R_{S}$, and so also the entropy $H(A_{S})$.
However, $\{U_{S}(i,j)\}_{(i,j)\in[d_{A}]\times[d_{B}]}$ are \emph{dependent}
random variables, and such random variables are typically more difficult
to handle than independent ones.\textbf{ }Letting $Y_{S}\equiv Y_{S}(1)\eqdef\frac{1}{d_{B}}\sum_{j\in[d_{B}]}U_{S}(1,j)$,
it can be shown that \AppRef{(see (\ref{eq: conditional entropy expression}) in Appendix \ref{sec:Proof-of-MI-Theorem})}{}
\begin{align}
H(A\mid S) & =-\frac{d_{A}d_{B}}{\eta}\E\left[Y_{S}\cdot\log(Y_{S})\right]+\log\frac{\eta}{d_{B}}.
\end{align}
Noting that $f(t)\eqdef t\log(t)$ is a convex function on $\mathbb{R}_{+}$,
one obtains from Jensen's inequality that 
\begin{equation}
-\E\left[Y_{S}\cdot\log(Y_{S})\right]\leq-\E[Y_{S}]\log\E[Y_{S}],\label{eq: Jensen inequality for E}
\end{equation}
and since $\E(Y_{S})=\frac{\eta}{d_{A}d_{B}}$, it immediately follows
that 
\begin{equation}
H(A\mid S)\leq\log d_{A}=H(A),
\end{equation}
as is already known from the conditioning reduces entropy property~\eqref{eq: CRE for description}. From the
above discussion, we deduce that in order to obtain a \emph{lower} bound on $H(A\mid S)$,
which is close to $H(A)=\log d_{A}$, it is required to show that
the Jensen-based bound in (\ref{eq: Jensen inequality for E}) is
close to an equality. Trivially, if $Y_{S}$ had been a deterministic
quantity, then any Jensen-based inequality is satisfied with equality,
and specifically (\ref{eq: Jensen inequality for E}). Continuing
this line of thought, one expects that if $Y_{S}$ is tightly concentrated
around its expected value (i.e., ``close'' to being deterministic),
then (\ref{eq: Jensen inequality for E}) approximately holds with
equality. Indeed, such relations have been extensively explored via
the \emph{functional entropy} of a non-negative random variable $X$,
defined as 
\begin{equation}
\Ent(X)\eqdef\E\left[X\log X\right]-\E[X]\log\left[\E(X)\right].\label{eq:functionalEntropy}
\end{equation}
The functional entropy\footnote{Not to be confused with the \emph{Shannon} entropy of a random variable
$H(Y)$, see \cite{boucheron2013concentration}\eat{, though the two definitions are related (if $Y$ is a likelihood
ratio of two probability measures then $\Ent(\cdot)$ is the KL divergence
between them)}.} is non-negative, and is conveniently upper bounded via \emph{logarithmic
Sobolev inequalities }(LSIs) \cite[Chapter 5]{boucheron2013concentration}\emph{.
}Specifically, these inequalities bound $\Ent(X)$ by the \emph{Efron-Stein
variance} of $X$ \cite[Chapter 5]{boucheron2013concentration}, which
in turn quantifies the concentration of $Y$ around its expected value
-- low Efron-Stein variance implies tight concentration around the
expected value, and thus low functional entropy by LSIs. Therefore,
the proof addresses the bounding of $\Ent(Y_{S})$. Nonetheless, LSIs
are typically derived for functions of independent random variables,
whereas here, as discussed, $Y_{S}\eqdef\frac{1}{d_{B}}\sum_{j\in[d_{B}]}U_{S}(1,j)$
is an average of \emph{dependent} random variables. To address this
matter, we define a new set of random variables $\{V(j)\}_{j\in[d_{B}]}$,
so that each $V(j)$ has the same marginal distribution as $U_{S}(1,j)$,
but where the $\{V(j)\}_{j\in[d_{B}]}$ are independent. In other
words, $\{V(j)\}{}_{j\in[d_{B}]}$ is a set of Bernoulli random variables for which $\P[V(j)=1]=\frac{\eta}{d_{A}d_{B}}$,
thus possibly asymmetric. We then define $\tilde{Y}\eqdef\frac{1}{d_{B}}\sum_{j\in[d_{B}]}V(j)$, and instead of directly bounding $\Ent(Y_{S})$ as is required for the proof, we bound $\Ent(\tilde{Y})$ and the difference between the two  functional entropies, to wit, we write 
\begin{equation}
\Ent(Y_{S})=\Ent(\tilde{Y})+[\Ent(Y_{S})-\Ent(\tilde{Y})],\label{eq: decomposition of the functional entropy of Ys}
\end{equation}
and then separately bound each of the terms. Denoting $\bar{\rho}\eqdef\frac{d_{A}d_{B}}{\eta}-1$
(which is an upper bound on the relative number of spurious tuples),
\AppRef{we show in Lemma \ref{lem: Bound of functional entropy} in Appendix
\ref{sec:Proof-of-MI-Theorem} that}{the proof shows that } 
\begin{equation}
\Ent(\tilde{Y})\leq\frac{2\overline{\rho}\log(1/\overline{\rho})}{1-\overline{\rho}}\cdot\frac{1}{d_{B}}.
\end{equation}
\AppRef{The proof of Lemma \ref{lem: Bound of functional entropy}}{The proof of this property}is based
on a LSI for the asymmetric Bernoulli
random variables $\{V(j)\}_{j\in[d_{B}]}$ \cite[Chapter 5]{boucheron2013concentration},
along with a careful bounding of the Efron-Stein variance of $\tilde{Y}$.
The next term in the decomposition of $\Ent(Y_{S})$ in (\ref{eq: decomposition of the functional entropy of Ys})
is absolutely bounded \AppRef{in Lemma \ref{lem: Bound on functional entropy difference }
in Appendix \ref{sec:Proof-of-MI-Theorem} as}{in the proof as} 
\begin{equation}
\left|\Ent(Y_{S})-\Ent(\tilde{Y})\right|\leq\sqrt{\frac{2\log^{2}(d_{B})}{d_{B}}}.
\end{equation}
\eat{The proof of \textbf{Lemma \ref{lem: Bound on functional entropy difference }}
first shows that the difference in functional entropies $|\Ent(Y_{S})-\Ent(\tilde{Y})|$
is upper bounded by the function $-\sqrt{t}\log(\sqrt{t})$ applied
to 
\begin{equation}
\E\left[\left(\tfrac{1}{d_{B}}\sum_{j\in[d_{B}]}(U_{S}(1,j)-V(j))\right)^{2}\right]=\E\left[(Y_{S}-\tilde{Y})^{2}\right],\label{eq:secondMoment}
\end{equation}
that is, the second moment of $Y_{S}-\tilde{Y}$. In the expectation
defining this second moment in (\ref{eq:secondMoment}), the joint
distribution of $\{U_{S}(1,j)\}_{j\in[d_{B}]}$ is set according to
the random relation model, and the joint distribution of $\{V(j)\}_{j\in[d_{B}]}$
is also set according to our construction above. However, any \emph{joint
}distribution over this pair of sets $\{U_{S}(1,j),V(j)\}_{j\in[d_{B}]}$
(that is, a \emph{coupling}), can be used in this bound. The proof
then chooses the trivial coupling in which the random variables $\{U_{S}(1,j)\}_{j\in[d_{B}]}$
are independent of $\{V(j)\}_{j\in[d_{B}]}$. For this choice, the
second moment is the simplest to evaluate, and it already leads to
an efficient bound (combined with the concentration result in the
next section).} Summing the bounds on $\Ent(\tilde{Y})$ and $|\Ent(Y_{S})-\Ent(\tilde{Y})|$
leads to a bound on $\Ent(Y_{S})$, which in turn shows that the Jensen-bound
in (\ref{eq: Jensen inequality for E}) is close to equality. This
shows that $H(A\mid S)\leq\log d_{A}$ in fact
approximately achieved, up to the defined vanishing term $C(d_{B})$.

\subsubsection{The concentration to the expected value of the entropy\label{subsec:The-concentration-to}}

We next discuss the second step of the proof of Theorem
\ref{thm: confidence bound for random entropy}. We state a concentration
bound on $H(A_{S})$ to $\E[H(A_{S})]$ and outline
its proof. For brevity, for $t\in\mathbb{R}_{+}$, we denote
\begin{equation}
h(t)\eqdef t\log(1+t).\label{eq: h function definition}
\end{equation}
\begin{prop}
\label{prop: concentration of conditional entropy}Assume that $d_{A}>d_{B}$,
that $\eta\geq60d_{A}$ and that $\eta\leq d_{A}d_{B}-d_{B}$. Then,
it holds that 
\begin{multline}
\P\left[\left|H(A_{S})-\E[H(A_{S})]\right|>t\right]\\
\leq\frac{1}{2}\cdot e^{-\frac{\eta}{12}}+\frac{1}{2}\exp\left\{ -\frac{\eta}{2d_{A}}\cdot h\left(\frac{r}{2\log(\eta/e)}\right)+4\log(\eta)\right\} ,\label{eq: concentration of conditioanl entropy lemma statement}
\end{multline}
where 
\begin{equation}
r=\max\left\{ 0,t-\frac{8d_{A}}{\eta}-C(d_{B})\right\} .\label{eq: relation between t and r in concentration}
\end{equation}
\end{prop}

Previously, in Section \ref{subsec:The-expected-value}, we defined
the random variables $\{U_{S}(i,j)\}_{i\in[d_{A}],j\in[d_{B}]}$ and
then $\{Y_{S}(i)\}_{i\in[d_{A}]}$. For the proof of Proposition \ref{prop: concentration of conditional entropy}
it will be more convenient to use their scaled version
\begin{equation}
Z_{S}(i)\eqdef d_{B}\cdot Y_{S}(i)=\sum_{j\in[d_{B}]}U_{S}(i,j).
\end{equation}
By the definition of the random relation model, for each $i\in[d_{A}]$,
$Z_{S}(i)\sim\text{Hypergeometric}(d_{A}d_{B},d_{B},\eta)$, that
is, a hypergeometric random variable, with population size $d_{A}d_{B}$,
$d_{B}$ success states in the population, and $\eta$ draws. We also
note that since $\E[Y_{S}(i)]=\frac{\eta}{d_{A}d_{B}}$ then $\E[Z_{S}(i)]=\frac{\eta}{d_{A}}$. Since $\{Z_{S}(i)\}_{i\in[d_{A}]}$ are dependent random variables,
the first step of the proof uses a union bound over all $i\in[d_{A}]$,
and thus reduces the probability required to be bounded to just a
single one of them, say $Z_{S}(1)$. Specifically, the first step
of the proof\AppRef{ (see (\ref{eq: first bound on concentration inequality})
in Appendix \ref{sec:Proof-of-MI-Theorem})}{}shows that 
\begin{multline}
\P\left[\left|H(A_{S})-\E[H(A_{S})]\right|>t\right]\\
\leq d_{A}\cdot\P\left[\left|g\left(\frac{Z_{S}}{\eta}\right)-\frac{\E[H(A_{S})]}{d_{A}}\right|>\frac{t}{d_{A}}\right],\label{eq: concentration bound first step for proof outline}
\end{multline}
where $Z_{S}\equiv Z_{S}(1)$, and $g(t)\eqdef-t\log t$. Thus, the
probability that the entropy is close to its expected value is bounded
by the probability that a function of a hypergeometric random variable
is close to its expectation. To bound the latter probability, we aim
to use known concentration results, and specifically, concentration
of Lipschitz functions of \emph{Poisson} random variables. Therefore,
the next step is to replace the hypergeometric random variable $Z_{S}$
with a Poisson random variable $W\sim\text{Poisson}$($\frac{\eta}{d_{A}}$),
which has the same mean $\E[W]=\frac{\eta}{d_{A}}=\E[Z_{S}]$. As
is well known, the binomial distribution (and more generally, the
multinomial distribution) can be ``Poissonized'' in the sense that
the probability of any event under the binomial distribution is upper
bounded by the same probability under the Poisson distribution, with
a proper factor \cite[Thm. 5.7]{mitzenmacher2017probability}.
The hypergeometric is known to behave similarly to the binomial distribution,
\eat{(e.g., there are bounds on the total-variation between them)},and
so one may expect that it can also be ``Poissonized''. 
\AppRef{Lemma \ref{lem:Poissonization}
in Appendix \ref{sec:Proof-of-MI-Theorem}, which is a preliminary
step to the proof of Proposition \ref{prop: concentration of conditional entropy},
shows this Poissonization effect, and states the proper condition and
constants.}{In \cite{kenig2022quantifying} we prove this Poissonization effect, and state the proper condition and
constants.} Its statement and results are general, and may be of independent interest. Equipped with the ``Poissonization bound'', $Z_{S}$ can be replaced
by $W$, and as a result, the bound in (\ref{eq: concentration bound first step for proof outline})
is further upper bounded with a similar bound, except that the hypergeometric random variable $Z_S$ is replaced with a Poisson random variable $W$, and a larger multiplicative pre-factor ($21d_A^3$ instead of just $d_A$)
\eat{
\begin{align}
 & \P\left[\left|H(A_{S})-\E[H(A_{S})]\right|>t\right]\nonumber \\
 & \leq21d_{A}^{3}\cdot\P\left[\left|g\left(\frac{W}{\eta}\right)-\frac{\E[H(A_{S})]}{d_{A}}\right|>\frac{t}{d_{A}}\right].\label{eq: entropy concentration probability bound with Poisson}
\end{align}
}
The next matter to address is that $g(t)=-t\log t$ is \emph{not }a
Lipschitz function since its derivative is unbounded for $t\downarrow0$
as well as $t\uparrow\infty$ (note that while $Z_{S}\leq d_{B}$
with probability $1$, $W$ is unbounded). We first address the $t\downarrow0$
case. Since $W$ is supported on integers, the minimal non-zero argument
possible for $g(\frac{W}{\eta})$ is $1/\eta$. So, if we restrict
$t\in[\frac{1}{\eta},1]$ then $g(t)$ is a Lipschitz function with
semi-norm $\frac{1}{\eta}\log\eta$. Based on this observation we
propose a function $\hat{g}_{\eta}(t)$ which well approximates $g(t)$
on one hand, and is Lipschitz on the other hand. By an application
of the triangle inequality, the term in
(\ref{eq: concentration bound first step for proof outline}) \eat{
(\ref{eq: entropy concentration probability bound with Poisson})}
is upper bounded as 
\begin{multline}
\left|g\left(\frac{W}{\eta}\right)-\frac{\E[H(A_{S})]}{d_{A}}\right|\leq\left|g\left(\frac{W}{\eta}\right)-\hat{g}_{\eta}\left(\frac{W}{\eta}\right)\right|+\\
\left|\hat{g}_{\eta}\left(\frac{W}{\eta}\right)-\E\left[\hat{g}_{\eta}\left(\frac{W}{\eta}\right)\right]\right|+\left|\E\left[\hat{g}_{\eta}\left(\frac{W}{\eta}\right)\right]-\frac{\E[H(A_{S})]}{d_{A}}\right|.\label{eq: triangle inequality for concentration proof outline}
\end{multline}
The first term in (\ref{eq: triangle inequality for concentration proof outline})
is bounded as $1/\eta$ with probability $1$ directly from the construction
of the function $\hat{g}_{\eta}(t)$. \AppRef{The last term in (\ref{eq: triangle inequality for concentration proof outline})
is bounded in Lemma \ref{lem:modified entropy of Poission}, whose
proof is rather technical, and utilizes both the bound on the expected
value of Proposition \ref{prop: Expected value of entropy} previously
stated, as well as tools such as \emph{Poisson LSI} (see Lemma \ref{lem: Poisson LSI}
in Appendix \ref{sec:Auxiliary-results}).}{The last term in (\ref{eq: triangle inequality for concentration proof outline})
is  is difficult to bound, and the proof mainly uses a \emph{Poisson LSI} \cite[Thm. 6.17]{boucheron2013concentration} along with various approximation steps.} The proof continues by
bounding the probability that the middle term in (\ref{eq: triangle inequality for concentration proof outline})
is larger than some value. The main tool for this bound in a concentration bound for Lipschitz
functions of Poisson random variables. This bound is not used directly, since $\hat{g}_{\eta}(t)$ is not Lipschitz over the entire real line. However, the argument $W/\eta$ is small enough with high probability, and thus belong to the Lipschitz continuous part of this function. Additional approximation arguments show that this suffices to obtain tight upper bound. 
\eat{As said, a concentration bound for Lipschitz
functions of Poisson random variables cannot be directly used since
the constructed $\hat{g}_{\eta}(t)$ still has an unbounded derivative
as $t\uparrow\infty$ (and thus is not Lipschitz). Nonetheless, from
concentration of Poisson random variables, the probability that $\frac{W}{\eta}$
is much larger than its expected value $\E[\frac{W}{\eta}]=\frac{1}{d_{A}}\ll1$
is small, and so with high probability the argument of $\hat{g}_{\eta}$
is much less than $1$. We thus propose a to further modify $\hat{g}_{\eta}(t)$
to a function $\tilde{g}_{\eta}(t)$ which is Lipschitz on $\mathbb{R}_{+}$
and well approximates $\hat{g}_{\eta}$ in the domain for which $W$
belongs to with very high probability (the complementary low probability
event is the source of the additional term $\frac{1}{2}\cdot e^{-\frac{\eta}{12}}$
in (\ref{eq: concentration of conditioanl entropy lemma statement})).
The concentration of $\tilde{g}_{\eta}(\frac{W}{\eta})$ to its expected
value is then assured by a concentration property of Lipschitz functions
of Poisson random variables \cite{bobkov1998modified,kontoyiannis2006measure}
(\textbf{Lemma \ref{lem: Poisson concentration} in Appendix \ref{sec:Auxiliary-results}}).} The combination of the bounds for all three terms then establishes the proof of Theorem \ref{thm: confidence bound for random entropy}. 
\eat{The proof is then completed by taking into account the various approximations
made in the above steps. This completes the outline of the second
main step of the proof of Theorem \ref{thm: confidence bound for random entropy}.} 

	\section{Conclusion}
We show that the KL-Divergence is a useful measure for capturing the loss of an AJD with respect to the number of redundant tuples generated by the acyclic join. Our proposed random database model has allowed us to establish a high probability upper-bound on the percentage of redundant tuples, which coincides with the deterministic lower bound for large databases. Overall, our findings  provide insights into the information-theoretic nature of AJD loss.

\begin{acks}
The work of B.K. was supported by the US-Israel Binational Science Foundation (BSF) Grant No. 2030983, and the work of N.W. was supported in part by the Israel Science Foundation (ISF), Grant No. 1782/22. The work was also supported by the Technion MLIS-TDSI Grant No. 86703064. The authors thanks Or Glassman for various numerical computations relatd to this research. N.W. thanks Nadav Merlis for a discussion on tail bounds for hypergeometric random variables and sampling without replacement. 
\end{acks}    
\clearpage

    \bibliographystyle{ACM-Reference-Format}
    \balance
    \bibliography{main,Spurious_tuples_bounds}

    \ifFlag
	    \onecolumn
		\appendix
\section{Proofs for Section~\ref{sec:JMeasureKL}}
\label{sec:ProofsJMeasureKL}
\begin{repproposition}{\ref{prop:PT}}
Let $P(X_1,\dots,X_n)$ be any joint probability distribution over $n$ variables, and let $(\T,\chi)$ be a join tree where $\chi(\T)=\set{X_1,\dots,X_n}$. Then $P\models \T$ (Definition~\ref{def:modelsT}) if and only if $P=P_\T$ where:
	\begin{equation}
		\label{eq:PtAppendix}
		P_\T(x_1,\dots,x_n)\eqdef \frac{\prod_{i=1}^mP[ \Omega_i](\boldsymbol{\bx}[ \Omega_i])}{\prod_{i=1}^{m-1}P[ \Delta_i](\boldsymbol{\bx}[\Delta_i])}
	\end{equation}
	where $P[ \Omega_i]$ ($P[ \Delta_i]$) denote the marginal probabilities over $\Omega_i$ ($\Delta_i$).\eat{ We say that a probability distribution $P$ \e{models} the join tree $\T$ denoted $P\models \T$ if $P=P_\T$.}
\end{repproposition}
\begin{proof}
We prove both directions by induction on $m=|\nodes(\T)|$. 
If direction: we assume that $P\models \T$ (see Definition~\ref{def:modelsT}). When $m=1$ the claim is immediate, so assume that it holds for $m\geq 1$, and we prove for $|\nodes(\T)|=m+1$.
So $(\T,\jointreeMapFunction)$ is a junction tree over $m+1$ nodes. Let $u_{m+1}$ denote a leaf in $\T$ where $\jointreeMapFunction(u_{m+1})=\Omega_{m+1}$, and let $u_p=\parent(u_{m+1})$ in $\T$ where $\jointreeMapFunction(u_p)=\Omega_p$. We let $\Delta_m=\Omega_p\cap \Omega_{m+1}$.
By the assumption that $P\models \T$, then by Definition~\ref{def:modelsT}, we have that $I(\Omega_{(1:m)};\Omega_{m+1}|\Delta_m)$.
Now, create a new junction tree $\T'$ where nodes $\Omega_{m+1}$ and $\Omega_p$ are combined to a single node $u_x$ where $\jointreeMapFunction(u_x)=\Omega_{m+1}\cup \Omega_p$. Since $\T$ is a junction tree (Definition~\ref{def:joinTree}), then so is $\T'$. The set of edges of $\T'$ is a subset of those of $\T$. Therefore, since $P\models \T$ then it must hold that $P\models \T'$. Since $|\nodes(\T')|=m$, then by the induction hypothesis we have that:
\begin{equation}
\label{eq:PtPropProof}
    P_{\T'}(x_1,\dots,x_n)=\left(\frac{\prod_{i\in [m]\setminus \set{p}}P[\Omega_i](\bx[\Omega_i])}{\prod_{i\in [m-1]}P[\Delta_i](\bx[\Delta_i])}\right)\times P[\jointreeMapFunction(u_x)](\bx[\jointreeMapFunction(u_x)])
\end{equation}
Since $I(\Omega_{(1:m)};\Omega_{m+1}|\Delta_m)=0$, and since $\Omega_p\subseteq \Omega_{(1:m)}$ then $I(\Omega_{p};\Omega_{m+1}|\Delta_m)=0$. Therefore, $P[\jointreeMapFunction(u_x)]=\frac{P[\Omega_{m+1}]P[\Omega_p)]}{P[\Delta_m]}$. Substituting this back in~\eqref{eq:PtPropProof} proves the claim for the junction tree $\T$ with $m+1$ nodes.

For the other direction, we assume that~\eqref{eq:PtAppendix} holds, and prove by induction on $m=|\nodes(\T)|$ that $P\models\T$. That is, we prove that for every edge $i\in \edges(\T)$ where $i\in [m-1]$, it holds that $I(\Omega_{1,(i-1)};\Omega_i|\Delta_{i:m})=0$.
The claim clearly holds for $m=1$ (since $\edges(\T)=\emptyset$). We assume the claim holds for junction trees with at most $m\geq 1$ nodes, and prove for a junction tree $\T$ with $m+1$ nodes.
Let $u_{m+1}$ be a leaf node in $\T$ with parent $u_p$. Let $\T'$ be the junction tree that results from $\T$ by removing the node $u_{m+1}$. By the induction hypothesis we have that $I(\Omega_{1:(i-1)};\Omega_{i:m}|\Delta_i)$ for every $i\in [1,m-1]$. Now, by the assumption of the claim, we have that $P_\T=P_{\T'}\cdot\frac{P(\Omega_{m+1})}{P(\Delta_m)}$. Since $\jointreeMapFunction(\T')=\Omega_{(1:m)}$, we immediately get that $I(\Omega_{1:m}; \Omega_{m+1}|\Delta_m)=0$ as required.
\end{proof}

\begin{replemma}{\ref{lem:marginals}}
\lemMarginals
\end{replemma}
\begin{proof}
We prove the claim by by induction on $m=|\nodes(\T)|$. By definition of $P_\T$, the claim is immediate for $m=1$. So, we assume the claim holds for $m\geq 1$, and prove for $|\nodes(\T)|=m+1$.
Let $u_{m+1}$ be a leaf node in $\T$ where $\parent(u_{m+1})=u_p$, and $\Omega_{m+1}\cap \Omega_p=\Delta_{m}$. Consider the tree $T'$ where $\nodes(\T')\eqdef \nodes(\T)\setminus \set{u_{m+1}}$. Let $B_{m+1}\eqdef \Omega_{m+1}\setminus \Delta_m$. By the junction tree property (see Definition~\ref{def:joinTree}), $B_{m+1}\cap \Omega_{(1:m)}=\emptyset$. Accordingly, define $P_{\T'}(\Omega\setminus B_{m+1})=\frac{\prod_{i=1}^mP[\Omega_i]}{\prod_{i=1}^{m-1}P[\Delta_i]}$. Since $\T'$ has exactly $m$ nodes, and since $u_p\in \nodes(\T')$, then by the induction hypothesis, it holds that $P[\Omega_p]=P_{\T'}[\Omega_p]=\sum_{X\notin \Omega_p}P_{\T'}(\Omega)$.
Now,
\begin{align}
    \sum_{X\notin \Omega_{m+1}}P_\T(\Omega)&=\sum_{X\notin \Omega_{m+1}}\frac{\prod_{i=1}^mP[\Omega_i]}{\prod_{i=1}^{m-1}P[\Delta_i]}\cdot \frac{P[\Omega_{m+1}]}{P[\Delta_m]}\\
    &=\frac{P[\Omega_{m+1}]}{P[\Delta_m]}\cdot \sum_{X\notin \Delta_m}\frac{\prod_{i=1}^mP[\Omega_i]}{\prod_{i=1}^{m-1}P[\Delta_i]}\\
    &=\frac{P[\Omega_{m+1}]}{P[\Delta_m]}\cdot \sum_{X\notin \Delta_m}P_{\T'}(\Omega\setminus B_{m+1})\\
    &\trre[=,a]\frac{P[\Omega_{m+1}]}{P[\Delta_m]}\cdot \sum_{X\notin \Delta_m}\left(\sum_{X\notin \Omega_p}P_{\T'}(\Omega\setminus B_{m+1})\right)\\    &\trre[=,b]\frac{P[\Omega_{m+1}]}{P[\Delta_m]}\cdot \sum_{X\notin \Delta_m}P[\Omega_p] \\
    &\trre[=,c] \frac{P[\Omega_{m+1}]}{P[\Delta_m]}\cdot P[\Delta_m]=P[\Omega_{m+1}],
\end{align}
where $(a)$ follows since $\Delta_m\subseteq \Omega_p$, $(b)$ follows from the induction hypothesis, and $(c)$ follows again from $\Delta_m\subseteq \Omega_p$. 
\end{proof}

\begin{replemma}{\ref{lem:Pt}}. 
	\lemmaPt
\end{replemma}
\begin{proof}
	From Lemma~\ref{lem:marginals} we have that, for every $i\in \set{1,\dots,m}$,
	$P_\T[\Omega_i]=P[\Omega_i]$ where,  $\Omega_i=\jointreeMapFunction(u_i)$. Since $\Delta_i\subset \Omega_i$, then $P_\T[\Delta_i]=P[\Delta_i]$.
	\begin{align}
		\min_{Q \models \T}D_{KL}(P||Q)&=\min_{Q \models \T}\expectation_{P}\left[\log \frac{P(X_1,\dots,X_n)}{Q(X_1,\dots,X_n)}\right]\\	
		&=\min_{Q \models \T}\expectation_{P}\left[\log \frac{P(X_1,\dots,X_n)}{P_\T(X_1,\dots,X_n)}\cdot \frac{P_\T(X_1,\dots,X_n)}{Q(X_1,\dots,X_n)}\right]\\
		\eat{
			&=\min_{Q \models \T}\expectation_{P(\bX)}\left[\log \frac{P(X_1,\dots,X_n)}{P_\T(X_1,\dots,X_n)}\right]\\
			&~~~~~~~~~+\min_{Q \models \T}\expectation_{P}\left[\log \frac{P_\T(X_1,\dots,X_n)}{Q(X_1,\dots,X_n)}\right]\\}
		&=D_{KL}(P||P_\T)+\min_{Q \models \T}\expectation_{P}\left[\log \frac{P_\T(X_1,\dots,X_n)}{Q(X_1,\dots,X_n)}\right] 
	\end{align}
	Since the chosen distribution $Q(\bX)$ has no consequence on the first term $D_{KL}(P|| P_\T)$, we take a closer look at the second term, to wit, $\min_{Q \models \T}\expectation_{P}\left[\log \frac{P_\T(X_1,\dots,X_n)}{Q(X_1,\dots,X_n)}\right]$. Since $Q \models \T$, then
	$Q(X_1,\dots,X_n)=\frac{\prod_{i=1}^mQ[\Omega_i](\bX[\Omega_i])}{\prod_{i=1}^{m-1}Q[\Delta_i](\bX[\Delta_i])}$ (see~\eqref{eq:Pt}).
	Hence, in what follows, we refer to $Q$ as $Q_{\T}$. In the remainder of the proof we show that:
	\begin{align}
		\expectation_{P}\left[\log \frac{P_\T(X_1,\dots,X_n)}{Q^\T(X_1,\dots,X_n)}\right]&=\expectation_{P_{\T}}\left[\log \frac{P_\T(X_1,\dots,X_n)}{Q_\T(X_1,\dots,X_n)}\right] \\
		&=D_{KL}(P_\T||Q_\T),
	\end{align}
    where the last equality follows from \eqref{eq:KLD}. 
	Since $D_{KL}(P_\T||Q_\T)\geq 0$, with equality if and only if $P_\T=Q_\T$, then choosing $Q_\T$ to be $P_\T$ minimizes $D_{KL}(P||Q)$, thus proving the claim. The remainder of the proof follows from the fact that $P[\Omega_i]=P_\T[\Omega_i]$ for every $i\in [1,m]$, and $P[\Delta_j]=P_\T[\Delta_j]$ for every $j\in [m-1]$.
	\begin{align}
		&\expectation_{P}\left[\log \frac{P_\T(X_1,\dots,X_n)}{Q^\T(X_1,\dots,X_n)}\right] \nonumber\\
		&= \sum_{\bx\in \D(\bX)}P(\bx)\log\frac{\prod_{i=1}^mP^\T[\Omega_i](\bx[\Omega_i])}{\prod_{i=1}^mQ^\T[\Omega_i](\bx[\Omega_i])}\cdot\frac{\prod_{i=1}^{m-1}Q^\T[\Delta_i](\bx[\Delta_i]}{\prod_{i=1}^{m-1}P^\T[ \Delta_i](\bx[\Delta_i])}\\
		&=\sum_{\bx\in \D(\bX)}P(\bx)\log\frac{\prod_{i=1}^mP^\T[\Omega_i](\bx[\Omega_i])}{\prod_{i=1}^mQ^\T[\Omega_i](\bx[\Omega_i])}+\sum_{\bx\in \D(X)}P(\bx)\log\frac{\prod_{i=1}^{m-1}Q^\T[\Delta_i](\bx[\Delta_i])}{\prod_{i=1}^{m-1}P^\T[\Delta_i](\bx[\Delta_i])}\\
		&=\sum_{\bx\in \D(\bX)}P(\bx) \sum_{i=1}^m\log \frac{P^\T[\Omega_i](\bx[ \Omega_i])}{Q^\T[\Omega_i](\bx[\Omega_i])}+\sum_{\bx\in \D(\bX)}P(\bx) \sum_{i=1}^{m-1}\log \frac{Q^\T[\Delta_i](\bx[\Delta_i])}{P^\T[\Delta_i](\bx[\Delta_i])}\\
		&=\sum_{i=1}^m\sum_{\bx\in \D(\Omega_i)}P[\Omega_i](\bx)\log \frac{P^\T[\Omega_i](\bx)}{Q^\T[\Omega_i](\bx)}+ \sum_{i=1}^{m-1}\sum_{\bx\in \D(\Delta_i)}P[ \Delta_i](\bx)\log \frac{Q^\T[\Delta_i](\bx)}{P^\T[\Delta_i](\bx}\\
		&=\sum_{i=1}^m\sum_{\bx\in \D(\Omega_i)}P^\T[\Omega_i](\bx)\log \frac{P^\T[\Omega_i](\bx)}{Q^\T[\Omega_i](\bx)}+ \sum_{i=1}^{m-1}\sum_{\bx\in \D(\Delta_i)}P^\T[\Delta_i](\bx)\log \frac{Q^\T[\Delta_i](\bx)}{P^\T[\Delta_i](\bx}\\
		\eat{		
			&= \sum_{i=1}^m\sum_{\bx_{\mid \Omega_i}}P^\T_{\Omega_i}(\bx_{\mid \Omega_i})\log \frac{P_{\Omega_i}^\T(\bx_{\mid \Omega_i})}{Q_{\Omega_i}^\T(\bx_{\mid \Omega_i})}+ \sum_{i=2}^m\sum_{\bx_{\mid \Delta_i}}P^\T_{\Delta_i}(\bx_{\mid \Delta_i})\log \frac{Q_{\Delta_i}^\T(\bx_{\mid \Delta_i})}{P_{\Delta_i}^\T(\bx_{\mid \Delta_i})}\\}
		&=\sum_{i=1}^m \E_{P_\T}\left[\log \frac{P_\T[\Omega_i]}{Q_\T[\Omega_i]}\right]+\sum_{i=1}^{m-1}\E_{P_\T}\left[\log \frac{Q_\T[\Delta_i]}{P_\T[\Delta_i]}\right]\\
		&=\E_{P_\T}\left[\sum_{i=1}^m \log \frac{P_\T[\Omega_i]}{Q_\T[\Omega_i]}-\sum_{i=1}^{m-1}\log \frac{P_\T[\Delta_i]}{Q_\T[\Delta_i]}\right]\\
		&=\E_{P_\T}\left[\log \frac{\prod_{i=1}^m P_\T[\Omega_i]}{ \prod_{i=1}^m Q_\T[\Omega_i]}-\log \frac{\prod_{i=1}^{m-1} P_\T[\Delta_i]}{\prod_{i=1}^{m-1}Q_\T[\Delta_i]}\right]\\
		&= \E_{P_\T}\left[\log \frac{P^\T(X_1,\dots,X_n)}{Q^\T(X_1,\dots,X_n)}\right]\\
		&=D_{KL}(P^{\T},Q^{\T}),
	\end{align}
	and since the right term is nonnegative and equals zero if and only if we
	choose $Q^\T=P^\T$, the result follows.
\end{proof}
\subsection{Proof of Theorem~\ref{prop:JMeasureKL}}
\begin{reptheorem}{\ref{prop:JMeasureKL}}
	\JKLThm
\end{reptheorem}
\begin{proof}
	From Lemma~\ref{lem:Pt}, we have that $\min_{Q\models\T}D_{KL}(P||Q)=D_{KL}(P||P_\T)$, where $P_\T$ is defined in~\eqref{eq:Pt}. 
	Therefore, we prove that $\J(\T)=D_{KL}(P||P_\T)$.
	\begin{align}
		&D_{KL}(P||P_\T)=E_{P(\bX)}\left[\log \frac{P(\bX)}{P_\T(\bX)} \right] \\
		&=\sum_{\bx\in \D(\bX)}P(\bx)\log \frac{P(\bx)\prod_{i=1}^{m-1}P[\Delta_i](\bx[\Delta_i])}{\prod_{i=1}^mP[\Omega_i](\bx[\Omega_i])} \\
		&=\sum_{\bx\in \D(\bX)}P(\bx)\log P(\bx)+\sum_{\bx}P(\bx)\log \frac{\prod_{i=1}^{m-1}P[\Delta_i](\bx[\Delta_i])}{\prod_{i=1}^mP[\Omega_i](\bx[\Omega_i])}\\
		&=-H(\bX)+\sum_{\bx}P(\bx)\log \prod_{i=1}^{m-1}P[\Delta_i](\bx[\Delta_i])\nonumber \\
		&\qquad\qquad\qquad\qquad-\sum_{\bx}P(\bx) \log\prod_{i=1}^mP[\Omega_i](\bx[\Omega_i])\\
		&=-H(\bX)+\sum_{\bx}P(\bx)\sum_{i=1}^{m-1}\log P[\Delta_i](\bx[\Delta_i])\nonumber\\
		&\qquad\qquad\qquad\qquad-\sum_{\bx}P(\bx) \sum_{i=1}^m\log P[\Omega_i](\bx[\Omega_i])\\
		&=-H(\bX)+\sum_{i=1}^{m-1}\sum_{\by\in \D(\Delta_i)}P[\Delta_i](\by)\log P[\Delta_i](\by)\nonumber\\
		&\qquad\qquad\qquad\qquad-\sum_{i=1}^m\sum_{\by\in \D(\Omega_i)}P[\Omega_i](\by)\log P[\Omega_i](\by)\\
		\eat{
		&=\sum_{\bx}P(\bx)\log P(\bx)+\sum_{i=1}^{m-1} \sum_{\bx[\Delta_i]}P[\Delta_i](\bx[\Delta_i])\log P[\Delta_i](\bx[ \Delta_i])-\sum_{i=1}^m\sum_{\bx[\Omega_i]}P[\Omega_i](\bx_{[\Omega_i]}) \log P_{\Omega_i}(\bx_{\mid \Omega_i})\\}
		&=-H(\bX)-\Sigma_{i=1}^{m-1}H(\Delta_i)+\Sigma_{i=1}^mH(\Omega_i)&\\
		&=\J(\T).&
	\end{align}
 \end{proof}
		\section{Proof of Theorem \ref{thm: confidence bound for random entropy}
and Corollary \ref{cor: confidence bound for MI} \label{sec:Proof-of-MI-Theorem}}

In this section, we prove Theorem \ref{thm: confidence bound for random entropy}
and afterwards Corollary \ref{cor: confidence bound for MI}. We begin
with the proof of Theorem \ref{thm: confidence bound for random entropy},
which is comprised of three main steps: (I) The proof of Proposition
\ref{prop: Expected value of entropy} (analysis of $\E[H(A_{S})]$).
(II) The proof of  Proposition \ref{prop: concentration of conditional entropy}
(concentration of $H(A_{S})$ to its expected value). (III) A combination
of both propositions to establish a confidence interval.

\subsection{Proof of Proposition \ref{prop: Expected value of entropy}}

Let $A$ be the random variable that agrees with the distribution
of $A_{s}$, conditioned on $S=s$. Bayes rule implies that for any
$i\in[d_{A}]$ 
\begin{equation}
\P[A=i]=\sum_{s}\P[S=s]\cdot\P[A_{s}=i].
\end{equation}
Before delving into the proof of Proposition \ref{prop: Expected value of entropy},
we note that it is fairly intuitive that $A$ is distributed uniformly
over $[d_{A}]$, and so its entropy equals to $\log d_{A}$. This
is rigorously established in the next lemma. An analogous statement
holds for $B$. 
\begin{lem}
\label{lem:unconditionalentropies}It holds that $H(A)=\log d_{A}$
and $H(B)=\log d_{B}$.
\end{lem}

\begin{proof}
We prove only for $A$. For any $i\in[d_{A}]$, by elementary arguments
\begin{align}
\P(A=i) & =\sum_{j\in[d_{B}]}\P(A=i,B=j)\\
 & =\sum_{j\in[d_{B}]}\sum_{s\colon|s|=\eta}\P(A=i,B=j,S=s)\\
 & =\sum_{s\colon|s|=\eta}\sum_{j\in[d_{B}]}\P(S=s)\cdot\P(A=i,B=j\mid S=s)\\
 & =\sum_{s\colon|s|=\eta}\sum_{j\in[d_{B}]}\P(S=s)\cdot\P(A_{s}=i,B_{s}=j)\\
 & =\sum_{s\colon|s|=\eta}\P(S=s)\sum_{j\in[d_{B}]}\frac{1}{\eta}\cdot\indicator((i,j)\in s)\\
 & =\sum_{j\in[d_{B}]}\frac{1}{\eta}\sum_{s\colon|s|=\eta}\P(S=s)\cdot\indicator((i,j)\in s)\\
 & =\frac{d_{B}}{\eta}\cdot\sum_{s\colon|s|=\eta}\P(S=s)\cdot\indicator((i,1)\in s)\\
 & =\frac{d_{B}}{\eta}\cdot\P\left[(i,1)\in S\right]\\
 & =\frac{d_{B}}{\eta}\cdot\frac{\eta}{d_{A}d_{B}}=\frac{1}{d_{A}}.
\end{align}
\end{proof}
We now turn to the more challenging task of bounding the average entropy
$\E[H(A_{S})]$, and the proof of Proposition \ref{prop: Expected value of entropy},
which shows that $\E[H(A_{S})]=H(A\mid S)$ is close its unconditional
value $H(A)=\log d_{A}$. 
\begin{proof}[Proof of Proposition \ref{prop: Expected value of entropy}]
Recall that we denote $\bar{\rho}\eqdef\frac{d_{A}d_{B}}{\eta}-1.$
Let $U_{S}(i,j)\eqdef\indicator\{(i,j)\in S\}\}$ be $\{0,1\}$ random
variables, and further let 
\begin{equation}
Y_{S}(i)\eqdef\frac{1}{d_{B}}\sum_{j\in[d_{B}]}U_{S}(i,j),
\end{equation}
for which $\sum_{i\in[d_{A}]}Y_{S}(i)=\frac{\eta}{d_{B}}$ holds with
probability $1$. It should be noted that $\{U_{S}(i,j)\}_{i\in[d_{A}],j\in[d_{B}]}$
are identically distributed, but \emph{not }independent (they are
\emph{exchangeable}). With this notation 
\begin{equation}
H(A_{S})=\sum_{i\in[d_{A}]}-\frac{d_{B}Y_{S}(i)}{\eta}\log\left(\frac{d_{B}Y_{S}(i)}{\eta}\right)=\frac{d_{B}}{\eta}\sum_{i\in[d_{A}]}-Y_{S}(i)\log(Y_{S}(i))+\log\frac{\eta}{d_{B}},
\end{equation}
and so 
\begin{align}
\E[H(A_{S})] & =\frac{d_{B}}{\eta}\sum_{i\in[d_{A}]}-\E\left[Y_{S}(i)\cdot\log(Y_{S}(i))\right]+\log\frac{\eta}{d_{B}}\\
 & \trre[=,a]-\frac{d_{A}d_{B}}{\eta}\E\left[Y_{S}(1)\cdot\log(Y_{S}(1))\right]+\log\frac{\eta}{d_{B}}\\
 & \trre[=,b]-\frac{d_{A}d_{B}}{\eta}\E\left[Y_{S}\cdot\log(Y_{S})\right]+\log\frac{\eta}{d_{B}}\\
 & \trre[=,c]-\frac{d_{A}d_{B}}{\eta}\cdot\left(\E[Y_{S}]\cdot\log\left(\E[Y_{S}]\right)+\Ent(Y_{S})\right)+\log\frac{\eta}{d_{B}}\\
 & \trre[=,d]-\frac{1}{1+\overline{\rho}}\cdot\Ent(Y_{S})+\log d_{A},\label{eq: conditional entropy expression}
\end{align}
where $(a)$ follows from the symmetry of the distribution of $S$,
in $(b)$ we denote $Y_{S}\eqdef Y_{S}(1)$ for brevity, in $(c)$
we define the functional entropy \cite[Chapter 5]{boucheron2013concentration}
of a non-negative random variable $X$ by 
\begin{equation}
\Ent(X)\eqdef\E\left[X\log X\right]-\E[X]\log\left[\E(X)\right],
\end{equation}
and $(d)$ follows since $\E[Y_{S}]=\frac{\eta}{d_{A}d_{B}}=\frac{1}{1+\overline{\rho}}$. 

We continue the proof by bounding $\Ent(Y_{S})$, with $Y_{S}=\frac{1}{d_{B}}\sum_{j\in[d_{B}]}U_{S}(1,j)$,
and omit the first index (which is constant $1$) for notational brevity.
As mentioned before, $\{U_{S}(j)\}_{j\in[d_{B}]}$ are not independent
random variables, but as we shall see, this dependence is rather weak.
Each of them (marginally) follows the identical probability distribution
\begin{equation}
\P\left[U_{S}(j)=1\right]=\P\left[\indicator\{(1,j)\in S\}\}\right]=\frac{\eta}{d_{A}d_{B}}=\frac{1}{1+\overline{\rho}}.
\end{equation}
Thus, we define another sequence of random variables $\{V(j)\}_{j\in[d_{B}]}$
that are i.i.d., and where $V(j)$ is distributed as $U_{S}(j)$,
that is
\begin{equation}
\P\left[V(j)=1\right]=\P\left[U_{S}(j)=1\right].
\end{equation}
We then denote, analogously to $Y_{S}$, the random variable $\tilde{Y}\eqdef\frac{1}{d_{B}}\sum_{j\in[d_{B}]}V(j)$.
We expect that the distribution of $\tilde{Y}$ is close to that of
$Y_{S}$, and thus upper bound $\Ent(Y_{S})$ as follows:
\begin{align}
\Ent(Y_{S}) & \leq\Ent(\tilde{Y})+\left|\Ent(Y_{S})-\Ent(\tilde{Y})\right|\label{eq: bound on functional entropy prelim decomposition}\\
 & \leq\frac{2\overline{\rho}\log(1/\overline{\rho})}{1-\overline{\rho}}\cdot\frac{1}{d_{B}}+\sqrt{\frac{2\log^{2}\left((1+\overline{\rho})d_{B}\right)}{(1+\overline{\rho})d_{B}}},\label{eq: bound on functional entropy}
\end{align}
where the first term in (\ref{eq: bound on functional entropy prelim decomposition})
is bounded in the following Lemma \ref{lem: Bound of functional entropy},
and the second term in (\ref{eq: bound on functional entropy prelim decomposition})
is bounded afterwards in Lemma \ref{lem: Bound on functional entropy difference }.
Equipped with this bound on $\Ent(Y_{S})$, we return to (\ref{eq: conditional entropy expression})
to obtain 
\begin{align}
\E[H(A_{S})] & \trre[\geq,a]\log d_{A}-\frac{2\overline{\rho}\log(1/\overline{\rho})}{(1-\overline{\rho}^{2})d_{B}}-\sqrt{\frac{2\log^{2}(d_{B})}{(1+\overline{\rho})^{2}d_{B}}}\\
 & \trre[\geq,b]\log d_{A}-\frac{1}{d_{B}}-\sqrt{\frac{2\log^{2}(d_{B})}{(1+\overline{\rho})^{2}d_{B}}}\\
 & \trre[\geq,c]\log d_{A}-2\sqrt{\frac{\log^{2}(d_{B})}{d_{B}}},
\end{align}
where $(a)$ follows from (\ref{eq: conditional entropy expression})
and (\ref{eq: bound on functional entropy}), $(b)$ follows since
$\frac{2\overline{\rho}\log(1/\overline{\rho})}{(1-\overline{\rho}^{2})}\leq1$
for all $\overline{\rho}\ge0$, and $(c)$ is a slight weakening of
the bound. The result of the proposition then follows with the definition
of $C(d)$ stated in (\ref{eq: C constant def}). To complete the
proof, it remains to establish (\ref{eq: bound on functional entropy}).
Next, this is stated and then proved in Lemmas \ref{lem: Bound of functional entropy}
and \ref{lem: Bound on functional entropy difference }.
\end{proof}
\begin{lem}
\label{lem: Bound of functional entropy}Under the setting of the
proof of Proposition \ref{prop: Expected value of entropy}
\begin{equation}
\Ent(\tilde{Y})\leq\frac{2\overline{\rho}\log(1/\overline{\rho})}{1-\overline{\rho}}\cdot\frac{1}{d_{B}}.\label{eq: bound on functional entropy of bernoulli}
\end{equation}
\end{lem}

\begin{proof}
Recall that $\{V(j)\}_{j\in[d_{B}]}$ are $\{0,1\}$ i.i.d. Bernoulli
random variables. We denote by $R(j)\eqdef2V(j)-1$ the corresponding
$\{-1,1\}$ variables. Further denote the function
\begin{equation}
f\left(R(1),\ldots,R(d_{B})\right)\eqdef\sqrt{\frac{1}{2d_{B}}\sum_{j\in[d_{B}]}(R(j)+1)},
\end{equation}
so that using the above notation, $f^{2}(R(1),\ldots,R(d_{B}))=\sum_{j\in[d_{B}]}V(j)=\tilde{Y}$.
Thus, our goal is to bound $\Ent(f^{2})$, and to this end, we will
utilize an LSI for asymmetric Bernoulli random variables \cite[Chapter 5]{boucheron2013concentration}
restated in Lemma \ref{lem: LSI for asymmetric Bernoulli} in Appendix
\ref{sec:Auxiliary-results}. Let $\mathcal{E}(g)$ be the Efron-Stein
variance of $g$ defined in the statement of Lemma \ref{lem: LSI for asymmetric Bernoulli}
in (\ref{eq: Efron-Stein variance}) (note that it also depends on
$p$, the assumed distribution of $\{R(j)\}_{j\in[d_{B}]}$). Then,
the LSI in Lemma \ref{lem: LSI for asymmetric Bernoulli} states that
\begin{equation}
\Ent(g^{2})\leq\frac{p(1-p)}{1-2p}\log\left(\frac{1-p}{p}\right)\cdot\mathcal{E}(g).\label{eq: LSI on entropy}
\end{equation}
First, we evaluate the pre-factor in (\ref{eq: LSI on entropy}).
Per our definitions, it holds that 
\begin{equation}
p=\P[R(1)=1]=\P[U_{S}(1)=1]=\frac{\eta}{d_{A}d_{B}}=\frac{1}{1+\overline{\rho}},
\end{equation}
and so, the pre-factor in the functional entropy bound (\ref{eq: LSI on entropy})
is 
\begin{align}
\frac{p(1-p)}{1-2p}\log\frac{1-p}{p} & =\frac{\overline{\rho}}{1-\overline{\rho}^{2}}\log(1/\overline{\rho}).\label{eq: pre-factor evaluation in LSI}
\end{align}
Second, we bound the Efron-Stein variance as 
\begin{align}
\mathcal{E}(f) & =\E\left[\sum_{j\in[d_{B}]}\left(\sqrt{\frac{1}{2d_{B}}\sum_{i\in[d_{B}]}(R(i)+1)}-\sqrt{\frac{1}{2d_{B}}\sum_{i\in[d_{B}]}((-1)^{\indicator\{i=j\}}R(i)+1)}\right)^{2}\right]\\
 & =\E\left[\sum_{j\in[d_{B}]}\left(\sqrt{\frac{1}{2d_{B}}\sum_{i\in[d_{B}]}(R(i)+1)}-\sqrt{\frac{1}{2d_{B}}\sum_{i\in[d_{B}]}(R(i)+1)-\frac{R(j)}{d_{B}}}\right)^{2}\right]\\
 & =\E\left[\sum_{j\in[d_{B}]}\left(\sqrt{\frac{1}{d_{B}}\sum_{i\in[d_{B}]}V(i)}-\sqrt{\frac{1}{d_{B}}\sum_{i\in[d_{B}]}V(i)-\frac{2V(j)-1}{d_{B}}}\right)^{2}\right]
\end{align}
using $V(j)\eqdef\frac{R(j)+1}{2}$ (i.e., returning to $\{0,1\}$
random variables). Note that $\E[V(j)]=\frac{1}{1+\overline{\rho}}$
for all $j\in[d_{B}]$. Consider the event 
\begin{equation}
\mathcal{A}\eqdef\left\{ \frac{1}{d_{B}}\sum_{i\in[d_{B}]}V(i)<\frac{1}{2}\E[V(i)]\right\} .
\end{equation}
Utilizing the relative Chernoff's bound ((\ref{eq: relative chernoff})
in Lemma \ref{lem: relative Chernoff for binomial} in Appendix \ref{sec:Auxiliary-results})
for $\mathcal{A}$ (with $\xi=\frac{1}{2}$) results
\begin{equation}
\P[\mathcal{A}]\leq e^{-\frac{d_{B}}{12(1+\overline{\rho})}}.
\end{equation}
On $\mathcal{A}^{c}$ (the complementary event of $\mathcal{A}$)
it holds that $\frac{1}{d_{B}}\sum_{i\in[d_{B}]}V(i)\geq\frac{1}{2(1+\overline{\rho})}$.
Under the assumption of the Proposition \ref{prop: Expected value of entropy}
$\eta\geq60d_{A}$ it definitely holds that $d_{B}\geq4(1+\overline{\rho})$,
and so $|\frac{2V(j)-1}{d_{B}}|\leq\E[V(i)]=\frac{1}{(1+\overline{\rho})}$.
We continue with the bound on $\mathcal{E}(f)$ as follows:
\begin{align}
\mathcal{E}(f) & =\P(\mathcal{A})\cdot\E\left[\sum_{j\in[d_{B}]}\left(\sqrt{\frac{1}{d_{B}}\sum_{i\in[d_{B}]}V(i)}-\sqrt{\frac{1}{d_{B}}\sum_{i\in[d_{B}]}V(i)-\frac{2V(j)-1}{d_{B}}}\right)^{2}\,\middle\vert\,\mathcal{A}\right]\nonumber \\
 & \hphantom{=}+\E\left[\sum_{j\in[d_{B}]}\left(\sqrt{\frac{1}{d_{B}}\sum_{i\in[d_{B}]}V(i)}-\sqrt{\frac{1}{d_{B}}\sum_{i\in[d_{B}]}V(i)-\frac{2V(j)-1}{d_{B}}}\right)^{2}\cdot\indicator\left\{ \mathcal{A}^{c}\right\} \right].\label{eq: ES as two terms}
\end{align}
We separately bound each of the two terms of (\ref{eq: ES as two terms}).
For the first term in (\ref{eq: ES as two terms}), it holds with
probability $1$ that 
\begin{equation}
\sum_{j\in[d_{B}]}\left(\sqrt{\frac{1}{d_{B}}\sum_{i\in[d_{B}]}V(i)}-\sqrt{\frac{1}{d_{B}}\sum_{i\in[d_{B}]}V(i)-\frac{2V(j)-1}{d_{B}}}\right)^{2}\leq1.
\end{equation}
To see this, note that this trivially holds if all $V(i)=0$. Otherwise,
we use the fact that the concavity of $\sqrt{t}$ implies that $(\sqrt{t}-\sqrt{t-\tau})^{2}$
for $t\in[\frac{1}{d_{B}},1]$ and $|\tau|\leq\frac{1}{d_{B}}$ is
maximized at $t=\frac{1}{d_{B}}$. Hence, 
\begin{equation}
\P(\mathcal{A})\cdot\E\left[\sum_{j\in[d_{B}]}\left(\sqrt{\frac{1}{d_{B}}\sum_{i\in[d_{B}]}V(i)}-\sqrt{\frac{1}{d_{B}}\sum_{i\in[d_{B}]}V(i)-\frac{2V(j)-1}{d_{B}}}\right)^{2}\,\middle\vert\,\mathcal{A}\right]\leq e^{-\frac{d_{B}}{12(1+\overline{\rho})}}.\label{eq: ES bound around zero case}
\end{equation}
This term decays exponentially fast to zero with $d_{B}$. For the
second term of (\ref{eq: ES as two terms}), it holds that $\frac{1}{d_{B}}\sum_{i\in[d_{B}]}V(i)\geq\frac{1}{2(1+\overline{\rho})}$
and so 
\begin{equation}
\frac{1}{d_{B}}\sum_{i\in[d_{B}]}V(i)-\frac{2V(j)-1}{d_{B}}\geq\frac{1}{2(1+\overline{\rho})}-\frac{1}{2d_{B}}\geq\frac{1}{4(1+\overline{\rho})},
\end{equation}
where the last inequality holds under the assumption that $d_{B}\geq4(1+\overline{\rho})$.
On the interval $[\frac{1}{4}\frac{1}{1+\overline{\rho}},1]$ the
function $t\to\sqrt{t}$ has maximal derivative 
\begin{equation}
\max_{t\in[\frac{1}{4(1+\overline{\rho})},1]}\frac{\d\sqrt{t}}{\d t}=\max_{t\in[\frac{1}{4(1+\overline{\rho})},1]}\frac{1}{2\sqrt{t}}=\frac{1}{2\sqrt{\frac{1}{4(1+\overline{\rho})}}}=\sqrt{1+\overline{\rho}}.
\end{equation}
In other words, the square-root function is $(\sqrt{1+\overline{\rho}})$-Lipschitz
on that interval. Thus, for any $t_{0},t_{1}\in[\frac{1}{2}\frac{1}{1+\overline{\rho}},1]$
it holds that 
\begin{equation}
\left|\sqrt{t_{0}}-\sqrt{t_{1}}\right|\leq\sqrt{1+\overline{\rho}}\cdot|t_{0}-t_{1}|.
\end{equation}
Applying this to the second term of (\ref{eq: ES as two terms}) results
\begin{align}
 & \E\left[\sum_{j\in[d_{B}]}\left(\sqrt{\frac{1}{d_{B}}\sum_{i\in[d_{B}]}V(i)}-\sqrt{\frac{1}{d_{B}}\sum_{i\in[d_{B}]}V(i)-\frac{2V(j)-1}{d_{B}}}\right)^{2}\cdot\indicator\left\{ \mathcal{A}^{c}\right\} \right]\nonumber \\
 & \leq\E\left[\sum_{j\in[d_{B}]}(1+\overline{\rho})\cdot\frac{1}{d_{B}}\right]\\
 & \leq\frac{1+\overline{\rho}}{d_{B}}.\label{eq: ES bound Lipschitz case}
\end{align}
Substituting the bounds (\ref{eq: ES bound around zero case}) and
(\ref{eq: ES bound Lipschitz case}) into (\ref{eq: ES as two terms})
results
\begin{equation}
\mathcal{E}(f)\leq e^{-\frac{d_{B}}{12(1+\overline{\rho})}}+\frac{1+\overline{\rho}}{d_{B}}\leq\frac{2(1+\overline{\rho})}{d_{B}},\label{eq: ES bound}
\end{equation}
where the last inequality holds since by the assumption $\eta\geq60d_{A}$
it holds that $d_{B}\geq60(1+\overline{\rho})$, and since it can
be numerically verified that $e^{-t}\leq\frac{1}{12t}$ for all $t\geq5$.
Thus, from the LSI (\ref{eq: LSI on entropy}), the computation of
the pre-factor in (\ref{eq: pre-factor evaluation in LSI}) and from
(\ref{eq: ES bound}), the entropy of $\tilde{Y}$ is bounded as claimed
in the statement of the lemma (\ref{eq: bound on functional entropy of bernoulli}).
\end{proof}
\begin{lem}
\label{lem: Bound on functional entropy difference }Under the setting
of the proof of Proposition \ref{prop: Expected value of entropy}
\begin{equation}
\left|\Ent(Y_{S})-\Ent(\tilde{Y})\right|\leq\sqrt{\frac{2\log^{2}(d_{B})}{d_{B}}}.\label{eq: bound on functional entropy difference statement}
\end{equation}
\end{lem}

\begin{proof}
Recall that $\{V(j)\}_{j\in[d_{B}]}$ are an i.i.d. version of $\{U_{S}(j)\}_{j\in[d_{B}]}$
in $\{0,1\}$ with equal marginals. Letting $g(t)\eqdef-t\cdot\log t$
it holds that 
\begin{equation}
\Ent(Y_{S})-\Ent(\tilde{Y})=\E\left[g\left(\frac{1}{d_{B}}\sum_{j\in[d_{B}]}U_{S}(j)\right)-g\left(\frac{1}{d_{B}}\sum_{j\in[d_{B}]}V(j)\right)\right].
\end{equation}
Note that when computing the expectation of the difference, we may
assume any joint distribution on $\{U_{S}(j)\}_{j\in[d_{B}]}$ and
$\{V(j)\}_{j\in[d_{B}]}$, which agrees with the marginals. In what
follows we choose them as \emph{independent}. We further bound 
\begin{align}
\Ent(Y_{S})-\Ent(\tilde{Y}) & \trre[\leq,a]2\E\left[g\left(\left|\frac{1}{d_{B}}\sum_{j\in[d_{B}]}\left(U_{S}(j)-V(j)\right)\right|\right)\right] \\
 & \trre[\leq,b]2g\left(\E\left[\left|\frac{1}{d_{B}}\sum_{j\in[d_{B}]}\left(U_{S}(j)-V(j)\right)\right|\right]\right),\label{eq: Ent difference from multinomial to iid in the proof via continuity of self information}
\end{align}
where $(a)$ follows from Lemma \ref{lem: continuity of self-information function},
which states that $|g(t)-g(s)|\leq2g(|s-t|)$ for $s,t\in[0,1]$,
and $(b)$ follows from Jensen's inequality and the concavity of $g(t)$.
We next upper bound the argument of $g(\cdot)$ (though note that
$g(t)$ is only monotonically increasing on $t\in[0,e^{-1}]$), as
follows. We begin with Jensen's inequality that implies
\begin{equation}
\E\left[\left|\frac{1}{d_{B}}\sum_{j\in[d_{B}]}\left(U_{S}(j)-V(j)\right)\right|\right]\leq\sqrt{\E\left[\left(\frac{1}{d_{B}}\sum_{j\in[d_{B}]}\left(U_{S}(j)-V(j)\right)\right)^{2}\right]}.\label{eq: from absoulte to square for mutlinomial and iid bits}
\end{equation}
We next evaluate the inner expectation (inside the square root), while,
as said, assuming that $\{U_{S}(j)\}_{j\in[d_{B}]}$ are independent
of $\{V(j)\}_{j\in[d_{B}]}$. For any pair of independent random variables
$X,\tilde{X}$ such that $\E[X]=\E[\tilde{X}]$ and $\E[X^{2}],\E[\tilde{X}^{2}]<\infty$
it holds that 
\begin{equation}
\E[(X-\tilde{X})^{2}]=\E[X^{2}]-2\E[X]\E[\tilde{X}]+\E[\tilde{X}^{2}]=\V[X]+\V[\tilde{X}].
\end{equation}
We next use this result for $X\equiv\frac{1}{d_{B}}\sum_{j\in[d_{B}]}U_{S}(j)$
and $\tilde{X}\equiv\frac{1}{d_{B}}\sum_{j\in[d_{B}]}V(j)$ in order
to bound the second moment on the left-hand side of (\ref{eq: from absoulte to square for mutlinomial and iid bits}).
First, since $\{V(j)\}_{j\in[d_{B}]}$ are i.i.d., and $V(1)\sim\text{Bernoulli}(\frac{\eta}{d_{A}d_{B}})$
it holds that 
\begin{equation}
\V\left[\frac{1}{d_{B}}\sum_{j\in[d_{B}]}V(j)\right]=\frac{1}{d_{B}^{2}}\sum_{j\in[d_{B}]}\V[V(j)]=\frac{1}{d_{B}}\cdot\V[V(1)]\leq\frac{1}{d_{B}}\cdot\frac{\eta}{d_{A}d_{B}}\leq\frac{1}{d_{B}}.
\end{equation}
Second, we evaluate 
\begin{align}
\V\left[\frac{1}{d_{B}}\sum_{j\in[d_{B}]}U_{S}(j)\right] & =\E\left[\left(\frac{1}{d_{B}}\sum_{j\in[d_{B}]}U_{S}(j)\right)^{2}\right]-\E^{2}\left[\frac{1}{d_{B}}\sum_{j\in[d_{B}]}U_{S}(j)\right]\\
 & =\E\left[\left(\frac{1}{d_{B}}\sum_{j\in[d_{B}]}U_{S}(j)\right)^{2}\right]-\left(\frac{\eta}{d_{A}d_{B}}\right)^{2}\\
 & =\E\left[\frac{1}{d_{B}^{2}}\sum_{j\in[d_{B}]}\sum_{k\in[d_{B}]}U_{S}(j)U_{S}(k)\right]-\left(\frac{\eta}{d_{A}d_{B}}\right)^{2}\\
 & \trre[=,a]\frac{d_{B}^{2}-d_{B}}{d_{B}^{2}}\E[U_{S}(1)U_{S}(2)]+\frac{1}{d_{B}}\E[U_{S}(1)]-\left(\frac{\eta}{d_{A}d_{B}}\right)^{2}\\
 & \leq\E[U_{S}(1)U_{S}(2)]+\frac{1}{d_{B}}-\left(\frac{\eta}{d_{A}d_{B}}\right)^{2}\\
 & \trre[\leq,b]\frac{1}{d_{B}},
\end{align}
where $(a)$ follows since $\{U_{S}(j)\}_{j\in[d_{B}]}$ are exchangeable
random variables (the joint distribution is invariant to permutations),
and since $U_{S}^{2}(1)=U_{S}(1)\in\{0,1\}$, and $(b)$ follows from
\begin{align}
\E[U_{S}(1)U_{S}(2)] & =\P[U_{S}(1)=1,U_{S}(2)=1]\\
 & =\P[U_{S}(1)=1]\cdot\P[U_{S}(1)=1\mid U_{S}(k)=1]\\
 & =\frac{\eta}{d_{A}d_{B}}\cdot\frac{\eta-1}{d_{A}d_{B}-1}\\
 & \leq\left(\frac{\eta}{d_{A}d_{B}}\right)^{2},
\end{align}
where the last inequality here holds since $\eta\leq d_{A}d_{B}$. 

Therefore 
\begin{equation}
\E\left[\left(\frac{1}{d_{B}}\sum_{j\in[d_{B}]}U_{S}(j)-V(j)\right)^{2}\right]=\V\left[\frac{1}{d_{B}}\sum_{j\in[d_{B}]}U_{S}(j)\right]+\V\left[\frac{1}{d_{B}}\sum_{j\in[d_{B}]}V(j)\right]\leq\frac{2}{d_{B}}.
\end{equation}
In turn, (\ref{eq: from absoulte to square for mutlinomial and iid bits})
results 
\begin{equation}
\E\left[\left|\frac{1}{d_{B}}\sum_{j\in[d_{B}]}\left(U_{S}(j)-V(j)\right)\right|\right]\leq\sqrt{\frac{2}{d_{B}}.}
\end{equation}
Under the assumption $\eta\geq60d_{A}$ it holds that $d_{B}\geq60(1+\overline{\rho})\geq15$
and then the upper bound $\sqrt{\frac{2}{d_{B}}}\leq e^{-1}$ holds.
As $g(t)$ is monotonically increasing in $t\in[0,e^{-1}]$, this
implies that
\begin{equation}
g\left(\E\left[\left|\frac{1}{d_{B}}\sum_{j\in[d_{B}]}U_{S}(j)-\frac{1}{d_{B}}\sum_{j\in[d_{B}]}V(j)\right|\right]\right)\leq g\left(\sqrt{\frac{2}{d_{B}}}\right).
\end{equation}
Combining this with (\ref{eq: Ent difference from multinomial to iid in the proof via continuity of self information})
results
\begin{equation}
\Ent(Y_{S})-\Ent(\tilde{Y})\leq2g\left(\sqrt{\frac{2}{d_{B}}}\right)=\sqrt{\frac{2}{d_{B}}}\log\left(\frac{d_{B}}{2}\right)\leq\sqrt{\frac{2\log^{2}(d_{B})}{d_{B}}},
\end{equation}
as claimed. 
\end{proof}

\subsection{Proof of Proposition \ref{prop: concentration of conditional entropy}}

As a preliminary step to the proof of Proposition \ref{prop: concentration of conditional entropy},
we state and prove a Poissonization bound on the hypergeometric random
variable $Z_{S}$. This bound can be of independent interest. 
\begin{lem}
\label{lem:Poissonization}Assume that $d_{A}\geq d_{B}$ and that
$\eta\in[d_{A},d_{A}d_{B}-d_{B}]$. Let $Z_{S}\sim\text{Hypergeometric}(d_{A}d_{B},d_{B},\eta)$
and let $W\sim\text{Poisson}(\frac{\eta}{d_{A}})$. Then, for any
$b\in[d_{B}]$
\begin{equation}
\P[Z_{S}=b]\leq21\cdot d_{A}^{2}\cdot\P[W=b].\label{eq: poissonization of the distribution}
\end{equation}
\end{lem}

\begin{proof}
We note that under the assumption $\eta\leq d_{A}d_{B}-d_{B}$ it
holds that $Z_{S}$ is supported on $\{0,1,2,\ldots,d_{B}\}$. The
following chain of inequalities then holds
\begin{align}
\P[Z_{S}=b] & \trre[=,a]{d_{B} \choose b}\cdot\prod_{i=0}^{b-1}\frac{\eta-i}{d_{A}d_{B}-i}\cdot\prod_{i=0}^{d_{B}-b-1}\left(1-\frac{\eta-b}{d_{A}d_{B}-b-i}\right)\\
 & \trre[=,b]\frac{1}{b!}\cdot\prod_{i=0}^{b-1}\frac{(\eta-i)(d_{B}-i)}{d_{A}d_{B}-i}\cdot\prod_{i=0}^{d_{B}-b-1}\left(1-\frac{\eta-b}{d_{A}d_{B}-b-i}\right)\\
 & =\frac{1}{b!}\cdot\left(\frac{\eta}{d_{A}}\right)^{b}\cdot\prod_{i=0}^{b-1}\frac{(1-\frac{i}{\eta})(d_{B}-i)}{d_{B}-\frac{i}{d_{A}}}\cdot\prod_{i=0}^{d_{B}-b-1}\left(1-\frac{\eta-b}{d_{A}d_{B}-b-i}\right)\\
 & =\frac{1}{b!}\cdot\left(\frac{\eta}{d_{A}}\right)^{b}\cdot\prod_{i=0}^{b-1}\frac{(1-\frac{i}{\eta})(1-\frac{i}{d_{B}})}{\left(1-\frac{i}{d_{A}d_{B}}\right)}\prod_{i=0}^{d_{B}-b-1}\left(1-\frac{\eta-b}{d_{A}d_{B}-b-i}\right)\\
 & \trre[\leq,c]4\cdot\frac{1}{b!}\cdot\left(\frac{\eta}{d_{A}}\right)^{b}\cdot\prod_{i=0}^{b-1}\left(1-\frac{i}{\eta}\right)\left(1-\frac{i}{d_{B}}\right)\prod_{i=0}^{d_{B}-b-1}\left(1-\frac{\eta-b}{d_{A}d_{B}-b-i}\right)\\
 & \leq4\cdot\frac{1}{b!}\cdot\left(\frac{\eta}{d_{A}}\right)^{b}\cdot\prod_{i=0}^{b-1}\left(1-\frac{i}{\eta}\right)\times\prod_{i=0}^{b-1}\left(1-\frac{i}{d_{B}}\right)\times\left(1-\frac{\eta-b}{d_{A}d_{B}-b}\right)^{d_{B}-b-1},\label{eq: Poissonization first bound}
\end{align}
where:
\begin{itemize}
\item In $(a)$, the first term ${d_{B} \choose b}$ is the number of possibilities
to choose the $b$ non-zero $\{U_{S}(1,j)\}_{j\in[d_{B}]}$, the second
term is the probability that this specific set of $j$'s is chosen
to be $1$, and the third term is the probability that the complementary
set is not chosen;
\item In $(b)$ we use 
\begin{equation}
{d_{B} \choose b}=\frac{d_{B}!}{b!(d_{B}-b)!}=\frac{1}{b!}\cdot d_{B}(d_{B}-1)\cdots d_{B}(d_{B}-b+1);
\end{equation}
\item In $(c)$ we use 
\begin{equation}
\prod_{i=0}^{b-1}\left(1-\frac{i}{d_{A}d_{B}}\right)\geq\left(1-\frac{d_{B}}{d_{A}d_{B}}\right)^{d_{B}}=\left[\left(1-\frac{1}{d_{A}}\right)^{d_{A}}\right]^{d_{B}/d_{A}}\geq\left[\frac{1}{4}\right]^{d_{B}/d_{A}}\geq\frac{1}{4}
\end{equation}
where here, the first inequality utilizes $b\leq d_{B}$, the second
inequality follows from the fact that $(1-\frac{1}{d})^{d}$ is monotonic
increasing on $[1,\infty)$, and so $(1-\frac{1}{d_{A}})^{d_{A}}\geq\frac{1}{4}$
(obtained for $d_{A}=2$), and the last inequality utilizes the assumption
$d_{B}\leq d_{A}$. 
\end{itemize}
We next bound the two product terms in (\ref{eq: Poissonization first bound}).
For the first product term, it holds that 
\begin{align}
\log\prod_{i=0}^{b-1}\left(1-\frac{i}{\eta}\right) & =\sum_{i=0}^{b-1}\log\left(1-\frac{i}{\eta}\right)\\
 & =\sum_{i=1}^{b-1}\log\left(1-\frac{i}{\eta}\right)\\
 & \trre[\leq,a]\int_{0}^{b-1}\log\left(1-\frac{x}{\eta}\right)\d x\\
 & =\left.\left(\eta-x\right)\cdot\log\left(1-\frac{x}{\eta}\right)-\left(\eta-x\right)\right|_{b-1}^{0}\\
 & =-\eta-\left[\left(\eta-b+1\right)\cdot\log\left(1-\frac{b-1}{\eta}\right)-\left(\eta-b+1\right)\right]\\
 & =-\left[\left(\eta-b+1\right)\cdot\log\left(1-\frac{b-1}{\eta}\right)+b-1\right],\label{eq: integral bound on first product}
\end{align}
where $(a)$ follows from the monotonicity of $t\to(1-\frac{t}{\eta})$
and Riemann integration. Similarly, it holds for the second product
term in (\ref{eq: Poissonization first bound}) that
\begin{equation}
\log\prod_{i=0}^{b-1}\left(1-\frac{i}{d_{B}}\right)\leq-\left[\left(d_{B}-b+1\right)\cdot\log\left(1-\frac{b-1}{d_{B}}\right)+b-1\right].\label{eq: integral bound on second product}
\end{equation}
Inserting the estimates (\ref{eq: integral bound on first product})
and (\ref{eq: integral bound on second product}) back to (\ref{eq: Poissonization first bound})
results 
\begin{align}
 & \P[Z_{S}=b]\nonumber \\
 & \leq4\cdot\frac{1}{b!}\cdot\left(\frac{\eta}{d_{A}}\right)^{b}\exp\left[-\left[\left(\eta-b+1\right)\cdot\log\left(1-\frac{b-1}{\eta}\right)+b-1\right]\right]\nonumber \\
 & \hphantom{==}\times\exp\left[-\left[\left(d_{B}-b+1\right)\cdot\log\left(1-\frac{b-1}{d_{B}}\right)+b-1\right]\right]\nonumber \\
 & \hphantom{==}\times\exp\left[\left(d_{B}-b-1\right)\log\left(1-\frac{\eta-b}{d_{A}d_{B}-b}\right)\right].\label{eq: Poissonization second bound}
\end{align}
Since the Poisson p.m.f. of $W\sim\text{Poisson}(\frac{\eta}{d_{A}})$
is given by 
\begin{equation}
\P[W=b]=\frac{\left(\frac{\eta}{d_{A}}\right)^{b}}{b!}\cdot e^{-d_{A}/\eta}
\end{equation}
it follows from (\ref{eq: Poissonization second bound}) that 
\begin{equation}
\frac{\P[Z_{S}=b]}{\P[W=b]}\leq4\cdot\exp[Q],
\end{equation}
where 
\begin{align}
Q & \eqdef-2b-\left(\eta-b+1\right)\cdot\log\left(\frac{\eta-b+1}{\eta}\right)-\left(d_{B}-b+1\right)\cdot\log\left(\frac{d_{B}-b+1}{d_{B}}\right)\nonumber \\
 & +\left(d_{B}-b-1\right)\log\left(1-\frac{\eta-b}{d_{A}d_{B}-b}\right)+\frac{\eta}{d_{A}}.\label{eq: definitoin of true probability to Poisson ratio exponent}
\end{align}
In the rest of the proof of the lemma, we prove that $Q\leq3+2\log d_{A}$
for any possible $0\leq b\leq d_{B}\leq d_{A}\leq\eta$. We prove
this in a few steps. We first simplify the fourth additive term of
$Q$, to wit, 
\begin{equation}
T\eqdef\left(d_{B}-b-1\right)\log\left(1-\frac{\eta-b}{d_{A}d_{B}-b}\right),
\end{equation}
by showing that $T$ can be tightly upper bounded by 
\begin{equation}
\tilde{T}\eqdef\left(d_{B}-b-1\right)\log\left(1-\frac{\eta}{d_{A}d_{B}}\right).
\end{equation}
To this end, we bound the difference 
\begin{align}
\tilde{T}-T & =\left(d_{B}-b-1\right)\log\left(1-\frac{\eta}{d_{A}d_{B}}\right)-\left(d_{B}-b-1\right)\log\left(1-\frac{\eta-b}{d_{A}d_{B}-b}\right)\\
 & =\left(d_{B}-b-1\right)\log\left(\frac{d_{A}d_{B}-\eta}{d_{A}d_{B}}\cdot\frac{d_{A}d_{B}-b}{d_{A}d_{B}-\eta}\right)\\
 & =\left(d_{B}-b-1\right)\log\left(1-\frac{b}{d_{A}d_{B}}\right)\\
 & \geq\left(d_{B}-b-1\right)\log\left(1-\frac{1}{d_{A}}\right),\label{eq: bound on difference in the fourth Poissonization term}
\end{align}
where the inequality follows since $b\leq d_{B}$. If $b<d_{B}$ then
we continue the lower bound on $\tilde{T}-T$ in (\ref{eq: bound on difference in the fourth Poissonization term})
as 
\begin{align}
\tilde{T}-T & \geq\left(d_{B}-b-1\right)\log\left(1-\frac{1}{d_{A}}\right)\\
 & \trre[\geq,a]-2\left(d_{B}-b-1\right)\frac{1}{d_{A}}\\
 & \trre[\geq,a]-2,
\end{align}
where $(a)$ follows since $\log(1-x)\geq-2x$ holds for $x\in[0,1/2]$,
and $(b)$ follows since $d_{B}\leq d_{A}$. Otherwise, if $b=d_{B}$,
then 
\begin{equation}
\tilde{T}-T\geq-\log\left(1-\frac{1}{d_{A}}\right)=\log\left(\frac{d_{A}}{d_{A}-1}\right)\geq0.
\end{equation}
We thus deduce that $T\leq\tilde{T}+2$ for any $0\leq b\leq d_{B}$.
Next, we further show that $\tilde{T}$ can be accurately upper bounded
by 
\begin{equation}
\overline{T}\eqdef\left(d_{B}-b+1\right)\log\left(1-\frac{\eta}{d_{A}d_{B}}\right).
\end{equation}
Indeed, denoting $q=d_{A}d_{B}-\eta$ where $q\geq d_{B}$ holds by
our assumption, it holds that 
\begin{align}
\overline{T}-\tilde{T} & =\left[\left(d_{B}-b+1\right)-\left(d_{B}-b-1\right)\right]\log\left(1-\frac{\eta}{d_{A}d_{B}}\right)\\
 & =2\cdot\log\left(1-\frac{\eta}{d_{A}d_{B}}\right)\\
 & =2\cdot\log\left(1-\frac{d_{A}d_{B}-q}{d_{A}d_{B}}\right)\\
 & =2\cdot\log\left(\frac{q}{d_{A}d_{B}}\right)\\
 & \geq-2\log d_{A}.
\end{align}
We thus have upper bounded $T$, the fourth term in (\ref{eq: definitoin of true probability to Poisson ratio exponent}),
by 
\begin{equation}
T\leq\tilde{T}+2\leq\overline{T}+2+2\log d_{A}.
\end{equation}
Inserting this upper bound back to (\ref{eq: definitoin of true probability to Poisson ratio exponent}),
we may bound $Q\leq Q_{1}+2+2\log d_{A}$ where 
\begin{align}
Q_{1} & \eqdef-2b-\left(\eta-b+1\right)\cdot\log\left(\frac{\eta-b+1}{\eta}\right)-\left(d_{B}-b+1\right)\cdot\log\left(\frac{d_{B}-b+1}{d_{B}}\right)\nonumber \\
 & \hphantom{==}+\left(d_{B}-b+1\right)\log\left(1-\frac{\eta}{d_{A}d_{B}}\right)+\frac{\eta}{d_{A}}\\
 & =-2b-\left(\eta-b+1\right)\cdot\log\left(\frac{\eta-b+1}{\eta}\right)\nonumber \\
 & \hphantom{==}-\left(d_{B}-b+1\right)\cdot\log\left((d_{B}-b+1)\cdot\frac{d_{A}}{d_{A}d_{B}-\eta}\right)+\frac{\eta}{d_{A}}.\label{eq: definition of true probability to Poisson ratio exponent second bound}
\end{align}
Next, we maximize $Q_{1}$ over $d_{A}$. Focusing only on the terms
which depend on $d_{A}$, the term required to be maximized is 
\begin{equation}
-(d_{B}-b+1)\cdot\log\left(\frac{d_{A}}{d_{A}d_{B}-\eta}\right)+\frac{\eta}{d_{A}}=(d_{B}-b+1)\cdot\log\left(d_{B}-\frac{\eta}{d_{A}}\right)+\frac{\eta}{d_{A}}.
\end{equation}
Under the constraint $\eta>d_{A}$ it can be easily verified that
the maxima occurs when $\frac{\eta}{d_{A}}=b-1$. Substituting this
into (\ref{eq: definition of true probability to Poisson ratio exponent second bound}),
we obtain that $Q_{1}\leq Q_{2}$ where 
\begin{equation}
Q_{2}\eqdef-b-1-\left(\eta-b+1\right)\cdot\log\left(\frac{\eta-b+1}{\eta}\right).\label{eq: definition of true probability to Poisson ratio exponent third bound}
\end{equation}
We further upper bound $Q_{2}$ by finding $b\in[0,d_{B}]$ which
maximizes its value. Taking derivative, we get 
\begin{equation}
\frac{\d Q_{2}}{\d b}=-1+\log\left(\frac{\eta-b+1}{\eta}\right)+1\leq0,
\end{equation}
and so the constrained maxima is $b=0$. Substituting $b=0$ in (\ref{eq: definition of true probability to Poisson ratio exponent third bound})
results $Q_{2}\leq Q_{3}$ where 
\begin{equation}
Q_{3}=1-\left(\eta+1\right)\cdot\log\left(\frac{\eta+1}{\eta}\right)\leq1.
\end{equation}
Summarizing all the above bounds, we obtain 
\begin{equation}
\frac{\P[Z_{S}=b]}{\P[W=b]}\leq4\cdot e^{3+2\log d_{A}}\leq21\cdot e^{2\log d_{A}}=21\cdot d_{A}^{2},
\end{equation}
 as claimed.
\end{proof}

We may now prove Proposition \ref{prop: concentration of conditional entropy}. 
\begin{proof}[Proof of Proposition \ref{prop: concentration of conditional entropy}]
Let $W\sim\text{Poisson}(\frac{\eta}{d_{A}})$, so that $W$ has
the same mean as $Z_{S}(i)$, and further let $g(t)\eqdef-t\log t$
(defining $g(0)\eqdef0$, which continuously extends $g(t)$ for $t\in[0,1]$).
We begin with the following chain of inequalities:
\begin{align}
 & \P\left[\left|H(A_{S})-\E[H(A_{S})]\right|>t\right]\nonumber \\
 & =\P\left[\left|\sum_{i\in[d_{A}]}g\left(\frac{Z_{S}(i)}{\eta}\right)-\E[H(A_{S})]\right|>t\right]\\
 & \leq\P\left[\frac{1}{d_{A}}\sum_{i\in[d_{A}]}\left|g\left(\frac{Z_{S}(i)}{\eta}\right)-\frac{\E[H(A_{S})]}{d_{A}}\right|>\frac{t}{d_{A}}\right]\\
 & \trre[\leq,a]d_{A}\cdot\P\left[\left|g\left(\frac{Z_{S}(i)}{\eta}\right)-\frac{\E[H(A_{S})]}{d_{A}}\right|>\frac{t}{d_{A}}\right]\\
 & \trre[\leq,b]21\cdot d_{A}^{3}\cdot\P\left[\left|g\left(\frac{W}{\eta}\right)-\frac{\E[H(A_{S})]}{d_{A}}\right|>\frac{t}{d_{A}}\right],\label{eq: first bound on concentration inequality}
\end{align}
where $(a)$ follows from the union bound, $(b)$ follows from the
Poisson based bound of Lemma \ref{lem:Poissonization} in (\ref{eq: poissonization of the distribution}). 

To further bound the probability in (\ref{eq: first bound on concentration inequality})
we aim to use a concentration inequality for Lipschitz functions of
Poisson random variables. However, strictly speaking $g(t)$ is not
a Lipschitz function, since the derivative $\frac{\d}{\d t}g(t)=-\log(et)$
is unbounded as $t\downarrow0$. However, if the argument $t$ is
lower bounded then $g(t)$ can be modified to a Lipschitz function,
with a small error. To this end, let a parameter $\zeta\geq e$ be
given, and define a modified version of $g(t)$ by
\begin{equation}
\hat{g}_{\zeta}(t)\eqdef\begin{cases}
t\log(\zeta/e)+1/\zeta, & 0\leq t\leq1/\zeta\\
-t\log t, & t\geq1/\zeta
\end{cases}.
\end{equation}
Then, on $[0,1]$, $\hat{g}_{\zeta}(t)$ is a continuous function
of $t$, its first derivative is continuous, it is $\log(\frac{\zeta}{e})$-Lipschitz,
and the difference between $g(t)$ and $\hat{g}_{\zeta}(t)$ is upper
bounded as
\begin{equation}
\max_{t\in[0,1]}\left|\hat{g}_{\zeta}(t)-g(t)\right|=\max_{t\in[0,1/\zeta]}\left|\hat{g}_{\zeta}(t)-g(t)\right|=\max_{t\in[0,1/\zeta]}\hat{g}_{\zeta}(t)-g(t)=\frac{1}{\zeta},\label{eq: bound on approximated tlogt}
\end{equation}
since on $t\in[0,1/\zeta]$, the function $t\to\hat{g}_{\zeta}(t)-g(t)$
is nonnegative, and monotonic decreasing (thus obtains its maximal
value at $t=0$). The function $\hat{g}_{\zeta}(t)$ is appropriate
for approximating $g(t)$ for $t\geq1/\zeta$. Specifically, we will
next use this approximation with $\zeta=\eta$ since the argument
$g(\frac{W}{\eta})$ appearing in (\ref{eq: first bound on concentration inequality})
is either zero or at least $1/\eta$. Concretely, the term defining
the event in probability (\ref{eq: first bound on concentration inequality})
is upper bounded as
\begin{align}
 & \left|g\left(\frac{W}{\eta}\right)-\frac{\E[H(A_{S})]}{d_{A}}\right|\nonumber \\
 & \trre[\leq,a]\left|g\left(\frac{W}{\eta}\right)-\hat{g}_{\eta}\left(\frac{W}{\eta}\right)\right|+\left|\hat{g}_{\eta}\left(\frac{W}{\eta}\right)-\E\left[\hat{g}_{\eta}\left(\frac{W}{\eta}\right)\right]\right|+\underbrace{\left|\E\left[\hat{g}_{\eta}\left(\frac{W}{\eta}\right)\right]-\frac{\E[H(A_{S})]}{d_{A}}\right|}_{\eqdef T}\\
 & \trre[\leq,b]\frac{1}{\eta}+\left|\hat{g}_{\eta}\left(\frac{W}{\eta}\right)-\E\left[\hat{g}_{\eta}\left(\frac{W}{\eta}\right)\right]\right|+\frac{5}{\eta}+\frac{C(d_{B})}{d_{A}},\label{eq: triangleq inequality approximation of t.log(t) function}
\end{align}
where $(a)$ follows from the triangle inequality, $(b)$ follows
from (\ref{eq: bound on approximated tlogt}), and the bound on $T$
follows from Lemma \ref{lem:modified entropy of Poission} that will
be proved separately after completing the rest of the proof. 

Letting 
\begin{equation}
\tilde{t}\eqdef\frac{t}{d_{A}}-\frac{5}{\eta}-\frac{C(d_{B})}{d_{A}},
\end{equation}
(which is positive under the assumption of the lemma), and utilizing
(\ref{eq: triangleq inequality approximation of t.log(t) function})
in (\ref{eq: first bound on concentration inequality}), the probability
of interest is upper bounded as 
\begin{equation}
\P\left[\left|H(A_{S})-\E[H(A_{S})]\right|>t\right]\leq21\cdot d_{A}^{3}\cdot\P\left[\left|\hat{g}_{\eta}\left(\frac{W}{\eta}\right)-\E\left[\hat{g}_{\eta}\left(\frac{W}{\eta}\right)\right]\right|>\tilde{t}\right].
\end{equation}
Next, we will bound this probability by utilizing concentration of
Lipschitz functions of Poisson random variables. However, since after
replacing $Z_{S}(i)$, which is bounded by $\eta$ with probability
$1$, with the unbounded $W$, the function $\hat{g}_{\eta}(t)$ is
not Lipschitz on the unbounded support of $\frac{W}{\eta}$. Indeed,
$\hat{g}_{\eta}(t)=-t\log t$ for $t>1$ and its derivative $-\log(te)$
is unbounded as $t\uparrow\infty$. However, since $\E[W]=\frac{\eta}{d_{A}}$,
then $\frac{W}{\eta}$ is close with high probability to its expected
value $\frac{1}{d_{A}}$, which is less than $\frac{1}{2}$. On the
region close to $\frac{1}{d_{A}}$, the function $\hat{g}_{\eta}(t)$
is indeed Lipschitz, which allows the utilization of the Poisson concentration
result. Formally, from Chernoff's bound for a Poisson random variables
(\ref{eq: Poisson Chernoff general-1}) (Lemma \ref{lem: Chernoff for Poisson rvs}
in Appendix \ref{sec:Auxiliary-results}) it holds that 
\begin{equation}
\P\left[X\geq\alpha\E[X]\right]\leq e^{-\alpha\lambda}\label{eq: Poisson Chernoff general}
\end{equation}
for any $\alpha>3e\approx8.15$. Thus, defining the event ${\cal E}\eqdef\{\frac{W}{\eta}\geq e^{-1}\}$,
and using $\E[W]=\frac{\eta}{d_{A}}$, it holds that 
\begin{align}
\P[\mathcal{E}]=\P\left[\frac{W}{\eta}\geq e^{-1}\right] & \le\P\left[\frac{W}{\eta}\geq\frac{1}{3}\right]\\
 & =\P\left[W\geq\frac{d_{A}}{3}\E[W]\right]\\
 & \leq e^{-\frac{\eta}{3}},\label{eq: Poission Chernoff for our bound}
\end{align}
assuming $\frac{d_{A}}{3}\geq3e$ (which is satisfied by the assumption
of the proposition $d_{A}\geq d_{B}\geq60(1+\overline{\rho})$, which,
in turn, holds by the assumption $\eta\geq60d_{A}$). 

Now, consider a further modification of $\hat{g}_{\eta}$, given by
\begin{equation}
\tilde{g}_{\eta}(t)\eqdef\begin{cases}
\hat{g}_{\eta}(t), & 0\leq t\leq e^{-1}\\
\hat{g}_{\eta}(e^{-1}), & t>e^{-1}
\end{cases}.
\end{equation}
This is a continuous function, with a continuous first derivative,
bounded by $\frac{1}{\eta}\log(\eta/e)$ (essentially, $\tilde{g}_{\eta}(t)$
tracks $\hat{g}_{\eta}(t)$ exactly for $t\in[0,e^{-1}]$, and as
$t$ increases above $e^{-1}$, $\tilde{g}_{\eta}(t)$ remains constant
at the maximal value of $\hat{g}_{\eta}(e^{-1})=e^{-1}$ obtained
at $t=e^{-1}$). We first bound the difference in the expectation
between $\hat{g}_{\eta}$ and its modification $\tilde{g}_{\eta}$,
that is,
\begin{align}
 & \left|\E\left[\hat{g}_{\eta}\left(\frac{W}{\eta}\right)-\tilde{g}_{\eta}\left(\frac{W}{\eta}\right)\right]\right|\nonumber \\
 & \leq\E\left[\left|\hat{g}_{\eta}\left(\frac{W}{\eta}\right)-\tilde{g}_{\eta}\left(\frac{W}{\eta}\right)\right|\right]\\
 & =\E\left[\left|\hat{g}_{\eta}\left(\frac{W}{\eta}\right)-\tilde{g}_{\eta}\left(\frac{W}{\eta}\right)\right|\cdot\indicator\left\{ \frac{W}{\eta}>e^{-1}\right\} \right]\\
 & =\E\left[\left|-\frac{W}{\eta}\log\frac{W}{\eta}-e^{-1}\right|\cdot\indicator\left\{ \frac{W}{\eta}>e^{-1}\right\} \right]\\
 & \leq\E\left[\left(\left|\frac{W}{\eta}\log\frac{W}{\eta}\right|+e^{-1}\right)\cdot\indicator\left\{ \frac{W}{\eta}>e^{-1}\right\} \right]\\
 & \leq\E\left[\left|\frac{W}{\eta}\log\frac{W}{\eta}\right|\cdot\indicator\left\{ \frac{W}{\eta}>e^{-1}\right\} \right]+e^{-1}\cdot\P\left[\frac{W}{\eta}\geq e^{-1}\right]\\
 & \trre[\leq,a]\E\left[\left(\left(\frac{W}{\eta}\right)^{2}+1\right)\cdot\indicator\left\{ \frac{W}{\eta}>e^{-1}\right\} \right]+e^{-1}\cdot\P\left[\frac{W}{\eta}\geq e^{-1}\right]\\
 & =\E\left[\frac{W^{2}}{\eta^{2}}\right]+(1+e^{-1})\P\left[\frac{W}{\eta}\geq e^{-1}\right]\\
 & \trre[\leq,b]\frac{\frac{\eta}{d_{A}}+\frac{\eta^{2}}{d_{A}^{2}}}{\eta^{2}}+(1+e^{-1})\P\left[\frac{W}{\eta}\geq e^{-1}\right]\\
 & \trre[\leq,c]\frac{1}{d_{A}\eta}+\frac{1}{d_{A}^{2}}+2e^{-\eta/3}\\
 & \trre[\leq,d]\frac{1}{d_{A}\eta}+\frac{1}{d_{A}^{2}}+\frac{1}{\eta}\\
 & \trre[\leq,e]\frac{3}{\eta},\label{eq: bound on difference in expectation for double modification}
\end{align}
where $(a)$ follows since $|t\log t|\leq t^{2}+1$ for all $t\geq e^{-1}$,
$(b)$ follows since for $X\sim\text{Poisson}(\lambda)$ $\E[X^{2}]=\lambda+\lambda^{2}$,
and $W$ has parameter $\E[W]=\frac{\eta}{d_{A}}$, $(c)$ follows
from (\ref{eq: Poission Chernoff for our bound}), $(d)$ follows
from $e^{-t}\leq\frac{1}{6t}$ for all $t\geq16$ when plugging $t\equiv\eta/3$
and under the assumption $\eta\geq d_{A}\geq60$, and $(e)$ follows
again from $d_{A}^{2}>\eta>d_{A}$. Defining 
\begin{equation}
\tilde{t}-\frac{3}{\eta}=\frac{t}{d_{A}}-\frac{8}{\eta}-\frac{C(d_{B})}{d_{A}}\eqdef\hat{t},
\end{equation}
we thus obtain 

\begin{align}
 & \P\left[\left|\hat{g}_{\eta}\left(\frac{W}{\eta}\right)-\E\left[\hat{g}_{\eta}\left(\frac{W}{\eta}\right)\right]\right|>\tilde{t}\right]\nonumber \\
 & \trre[\leq,a]\P\left[\left|\hat{g}_{\eta}\left(\frac{W}{\eta}\right)-\E\left[\tilde{g}_{\eta}\left(\frac{W}{\eta}\right)\right]\right|>\hat{t}\right]\\
 & =\P[{\cal E}]\cdot\P\left[\left|\hat{g}_{\eta}\left(\frac{W}{\eta}\right)-\E\left[\tilde{g}_{\eta}\left(\frac{W}{\eta}\right)\right]\right|>\hat{t}\,\middle\vert\,{\cal E}\right]\nonumber \\
 & \hphantom{=}+\P[{\cal E}^{c}]\cdot\P\left[\left|\hat{g}_{\eta}\left(\frac{W}{\eta}\right)-\E\left[\tilde{g}_{\eta}\left(\frac{W}{\eta}\right)\right]\right|>\hat{t}\,\middle\vert\,{\cal E}^{c}\right]\\
 & \leq\P[{\cal E}]+\P\left[\left|\hat{g}_{\eta}\left(\frac{W}{\eta}\right)-\E\left[\tilde{g}_{\eta}\left(\frac{W}{\eta}\right)\right]\right|>\hat{t}\,\middle\vert\,{\cal E}^{c}\right]\\
 & =\P[{\cal E}]+\frac{\P\left[\left|\tilde{g}_{\eta}\left(\frac{W}{\eta}\right)-\E\left[\tilde{g}_{\eta}\left(\frac{W}{\eta}\right)\right]\right|>\hat{t},{\cal E}^{c}\right]}{\P[{\cal E}^{c}]}\\
 & \trre[\leq,b]e^{-\frac{\eta}{3}}+\frac{\P\left[\left|\tilde{g}_{\eta}\left(\frac{W}{\eta}\right)-\E\left[\tilde{g}_{\eta}\left(\frac{W}{\eta}\right)\right]\right|>\hat{t},{\cal E}^{c}\right]}{1-e^{-\eta/3}}\\
 & \leq e^{-\frac{\eta}{3}}+\frac{\P\left[\left|\tilde{g}_{\eta}\left(\frac{W}{\eta}\right)-\E\left[\tilde{g}_{\eta}\left(\frac{W}{\eta}\right)\right]\right|>\hat{t}\right]}{1-e^{-\eta/3}}\\
 & \trre[\leq,c]e^{-\frac{\eta}{3}}+2\P\left[\left|\tilde{g}_{\eta}\left(\frac{W}{\eta}\right)-\E\left[\tilde{g}_{\eta}\left(\frac{W}{\eta}\right)\right]\right|>\hat{t}\right],\label{eq: concentration of first modified entropy in terms of second modificiation}
\end{align}
where $(a)$ follows from (\ref{eq: bound on difference in expectation for double modification})
and the definitions of $\tilde{t}$ and $\hat{t}$, $(b)$ follows
from (\ref{eq: Poission Chernoff for our bound}), and $(c)$ follows
from the assumption $\eta\geq d_{A}\geq4$. 

Now, $\frac{\eta}{\log(\eta/e)}\tilde{g}_{\eta}(t)$ is $1$-Lipschitz
over $t\in[0,\infty)$. From the Poisson concentration of Lipschitz
functions in Lemma \ref{lem: Poisson concentration} (Appendix \ref{sec:Auxiliary-results}),
for any $\hat{t}>0$ 
\begin{align}
 & \P\left[\left|\tilde{g}_{\eta}\left(\frac{W}{\eta}\right)-\E\left[\tilde{g}_{\eta}\left(\frac{W}{\eta}\right)\right]\right|>\hat{t}\right]\nonumber \\
 & =\P\left[\left|\frac{\eta}{\log(\eta/e)}\hat{g}_{\eta}\left(\frac{W}{\eta}\right)-\E\left[\frac{\eta}{\log(\eta/e)}\hat{g}_{\eta}\left(\frac{W}{\eta}\right)\right]>\frac{\eta\hat{t}}{\log(\eta/e)}\right|\right]\\
 & \leq\exp\left\{ -\frac{\eta\hat{t}}{4\log(\eta/e)}\log\left(1+\frac{\eta\hat{t}}{2\log(\eta/e)\cdot\eta/d_{A}}\right)\right\} \\
 & \leq\exp\left\{ -\frac{\eta\hat{t}}{4\log(\eta/e)}\log\left(1+\frac{d_{A}\hat{t}}{2\log(\eta/e)}\right)\right\} .
\end{align}
Inserting this back into (\ref{eq: concentration of first modified entropy in terms of second modificiation}),
and then back into (\ref{eq: first bound on concentration inequality})
results 
\begin{equation}
\P\left[\left|H(A_{S})-\E[H(A_{S})]\right|>t\right]
\leq21\cdot d_{A}^{3}\cdot\left[e^{-\frac{\eta}{3}}+2\exp\left\{ -\frac{\eta r}{4d_{A}\log(\eta/e)}\log\left(1+\frac{r}{2\log(\eta/e)}\right)\right\} \right]\label{eq: bound on concentraion of random entropy before simplification}
\end{equation}
where 
\begin{equation}
r=d_{A}\hat{t}=t-\frac{8d_{A}}{\eta}-C(d_{B}).\label{eq: r versus t}
\end{equation}
 We next simplify this bound by loosening it. For the first term in
(\ref{eq: bound on concentraion of random entropy before simplification}),
the assumption of the proposition $\eta\geq d_{A}$ implies that 
\begin{equation}
d_{A}^{3}\leq\eta^{3}\leq64\cdot\left(\frac{\eta}{4}\right)^{3}\leq\frac{64}{2688}\cdot e^{\eta/3}=\frac{1}{42}\cdot e^{\eta/3},
\end{equation}
since $e^{t}\geq2688\cdot t^{3}$ for all $t\geq900$ (as can be numerically
verified), and since we set $t\equiv\frac{\eta}{4}\geq\frac{60d_{A}}{4}\geq\frac{3600}{4}$.
For the second term in (\ref{eq: bound on concentraion of random entropy before simplification}),
we use $\eta\geq60d_{A}\geq3600$ to loosely bound $42d_{A}^{3}\leq\frac{1}{2}d_{A}^{4}\leq\frac{1}{2}\eta^{4}$.
Using this bound and the upper bound (\ref{eq: r versus t}) results
\begin{align}
 & \P\left[\left|H(A_{S})-\E[H(A_{S})]\right|>t\right]\nonumber \\
 & \leq\frac{1}{2}e^{-\frac{\eta}{12}}+\frac{1}{2}\exp\left\{ -\frac{\eta r}{4d_{A}\log(\eta/e)}\log\left(1+\frac{r}{2\log(\eta/e)}\right)+4\log(\eta)\right\} ,
\end{align}
which is the statement of the proposition in (\ref{eq: concentration of conditioanl entropy lemma statement}),
using the notation $h(t)\eqdef t\log(1+t)$. To complete the proof
it remains to prove Lemma \ref{lem:modified entropy of Poission},
which follows next.
\end{proof}
\begin{lem}
\label{lem:modified entropy of Poission}Under the setting and notation
in the proof of Proposition \ref{prop: concentration of conditional entropy},
\begin{equation}
\left|\E\left[\hat{g}_{\eta}\left(\frac{W}{\eta}\right)\right]-\frac{\E[H(A_{S})]}{d_{A}}\right|\leq\frac{5}{\eta}+\frac{C(d_{B})}{d_{A}}.
\end{equation}
\end{lem}
\begin{proof}
Proposition \ref{prop: Expected value of entropy} implies that
\begin{equation}
\left|\frac{\E[H(A_{S})]}{d_{A}}-\frac{\log d_{A}}{d_{A}}\right|\leq\frac{C(d_{B})}{d_{A}}.\label{eq: result of lemma on conditional entropy in the proof}
\end{equation}
We next show that $\E[\hat{g}_{\eta}(\frac{W}{\eta})]$ is $(\frac{5}{\eta})$-close
to $\frac{\log d_{A}}{d_{A}}$. First, we derive an upper bound
\begin{align}
\E\left[\hat{g}_{\eta}\left(\frac{W}{\eta}\right)\right] & \trre[\leq,a]\hat{g}_{\eta}\left(\E\left[\frac{W}{\eta}\right]\right)\\
 & =\hat{g}_{\eta}\left(\frac{1}{d_{A}}\right)\\
 & \trre[\leq,b]g\left(\frac{1}{d_{A}}\right)+\frac{1}{\eta}\\
 & =\frac{\log d_{A}}{d_{A}}+\frac{1}{\eta},\label{eq: Poisson approximation to entropy upper bound}
\end{align}
where $(a)$ follows from Jensen's inequality since $\hat{g}_{\eta}$
is concave, $(b)$ follows from the construction of $\hat{g}_{\eta}$
in (\ref{eq: bound on approximated tlogt}). Thus, $\E[\hat{g}_{\eta}(\frac{W}{\eta})]$
is $(\frac{1}{\eta})$-close to $\frac{\log d_{A}}{d_{A}}$ from below.
Second, we derive a lower bound as follows 
\begin{align}
\E\left[\hat{g}_{\eta}\left(\frac{W}{\eta}\right)\right] & \geq\E\left[g\left(\frac{W}{\eta}\right)\right]-\frac{1}{\eta}\\
 & =\E\left[-\frac{W}{\eta}\log\left(\frac{W}{\eta}\right)\right]-\frac{1}{\eta}\\
 & =\frac{1}{\eta}\E\left[-W\log(W)\right]+\frac{1}{\eta}\E\left[W\log(\eta)\right]-\frac{1}{\eta}\\
 & =-\frac{1}{\eta}\E\left[W\log(W)\right]+\frac{\log\eta}{d_{A}}-\frac{1}{\eta}\\
 & =-\frac{1}{\eta}\E\left[W\log(W)\right]+\frac{1}{\eta}\E\left[W\right]\log\left(\E[W]\right)-\frac{1}{\eta}\E\left[W\right]\log\left(\E[W]\right)+\frac{\log\eta}{d_{A}}-\frac{1}{\eta}\\
 & =-\frac{1}{\eta}\Ent(W)-\frac{1}{\eta}\E\left[W\right]\log\left(\E[W]\right)+\frac{\log\eta}{d_{A}}-\frac{1}{\eta}\\
 & =-\frac{1}{\eta}\Ent(W)-\frac{1}{d_{A}}\log\left(\frac{\eta}{d_{A}}\right)+\frac{\log\eta}{d_{A}}-\frac{1}{\eta}\\
 & =-\frac{1}{\eta}\Ent(W)+\frac{\log d_{A}}{d_{A}}-\frac{1}{\eta},\label{eq: Poisson approximation to entropy}
\end{align}
where the first inequality follows from (\ref{eq: bound on approximated tlogt}).
We next complete the proof of the lemma by showing that $\Ent(W)\leq4$,
thus showing that $\E[\hat{g}_{\eta}(\frac{W}{\eta})]$ is $(\frac{5}{\eta})$-close
to $\frac{\log d_{A}}{d_{A}}$ from above. Combining (\ref{eq: result of lemma on conditional entropy in the proof})
with (\ref{eq: Poisson approximation to entropy upper bound}) and
(\ref{eq: Poisson approximation to entropy}) results
\begin{equation}
\left|\E\left[\hat{g}_{\eta}\left(\frac{W}{\eta}\right)\right]-\frac{\E[H(A_{S})]}{d_{A}}\right|\leq\frac{5}{\eta}+\frac{C(d_{B})}{d_{A}},
\end{equation}
as was required to be proved. 

As said, we complete the proof by proving that $\Ent(W)\leq4$, as
follows. For a\emph{ }positive integer $w\in\mathbb{N}_{+}$, and
$\zeta>2$, let 
\begin{equation}
f_{\zeta}(w)\eqdef\begin{cases}
1/\zeta, & w=0\\
w, & w\geq1
\end{cases},
\end{equation}
and note that for any $w\in\mathbb{N}$
\begin{equation}
\max_{w\in\mathbb{N}}\left|w\log w-f_{\zeta}(w)\log f_{\zeta}(w)\right|=\frac{\log\zeta}{\zeta}\label{eq: approximation of Poisson function}
\end{equation}
(note that here $0\log0$ is defined as $0$, by continuity $t\log t\to0$
as $t\downarrow0$). Further note that 
\begin{equation}
\E\left[f_{\zeta}(W)\right]=\E[W]+\frac{1}{\zeta}e^{-\eta/d_{A}}=\frac{\eta}{d_{A}}+\frac{1}{\zeta}e^{-\eta/d_{A}}.
\end{equation}
Hence, 
\begin{align}
 & \left|\Ent(W)-\Ent(f_{\zeta}(W))\right|\nonumber \\
 & \leq\left|\E\left[W\log\left(W\right)\right]-\E\left[f_{\zeta}(W)\log\left(f(W)\right)\right]\right| \nonumber\\
 & \hphantom{=}+\left|\E\left[W\right]\log\left(\E[W]\right)-\E\left[f_{\zeta}(W)\right]\log\left(\E[f(W)]\right)\right|\\
 & \trre[\leq,a]\frac{\log\zeta}{\zeta}+\frac{1}{\zeta}e^{-\eta/d_{A}}\\
 & \leq\frac{\log\zeta}{\zeta},\label{eq: bound on functional entropy differene}
\end{align}
where $(a)$ follows from (\ref{eq: approximation of Poisson function})
and from the fact that $t\log t$ is $1$-Lipschitz on $t\in[1,\infty)$,
while assuming $\E[f_{\zeta}(W)]\geq\E[W]=\frac{\eta}{d_{A}}>1$ (by
assumption). We next bound $\Ent(f_{\zeta}(W))$ using the Poisson
LSI in Lemma \ref{lem: Poisson LSI} in Appendix \ref{sec:Auxiliary-results},
as follows:
\begin{align}
\Ent(f_{\zeta}(W)) & \leq\E\left[f_{\zeta}(W)\right]\E\left[\frac{1}{f_{\zeta}(W)}\right]\\
 & =\left(\frac{\eta}{d_{A}}+\frac{1}{\zeta}e^{-\eta/d_{A}}\right)\cdot\E\left[\frac{1}{f_{\zeta}(W)}\right]\\
 & \trre[\leq,a]\left(\frac{\eta}{d_{A}}+\frac{1}{\zeta}e^{-\eta/d_{A}}\right)\cdot\E\left[\frac{1}{\frac{1}{\zeta}+\frac{1}{2}W}\right]\\
 & =\left(\frac{\eta}{d_{A}}+\frac{1}{\zeta}e^{-\eta/d_{A}}\right)\zeta\cdot\E\left[\frac{1}{1+\frac{\zeta}{2}W}\right]\\
 & =\left(\frac{\zeta\eta}{d_{A}}+e^{-\eta/d_{A}}\right)\cdot\E\left[\frac{1}{1+\frac{\zeta}{2}W}\right]\\
 & \trre[\leq,b]\left(\frac{\zeta\eta}{d_{A}}+e^{-\eta/d_{A}}\right)\E\left[\frac{1}{1+W}\right]\\
 & \trre[=,c]\left(\frac{\zeta\eta}{d_{A}}+e^{-\eta/d_{A}}\right)\frac{1}{\eta/d_{A}}\left(1-e^{-\eta/d_{A}}\right)\\
 & \leq\left(\frac{\zeta\eta}{d_{A}}+e^{-\eta/d_{A}}\right)\frac{1}{\eta/d_{A}}\\
 & \leq\zeta+1,\label{eq: Bound on approximated entropy}
\end{align}
where $(a)$ follows from $1/\zeta+\frac{1}{2}w\leq f_{\zeta}(w)$
for all $w\in\mathbb{N}$ (under the assumption $\zeta>2$), $(b)$
follows again from $\zeta>2$, $(c)$ follows from follows from a
direct (series) computation, for $W\sim\text{Poisson}(\lambda)$ 
\begin{align}
\E\left[\frac{1}{1+W}\right] & =\sum_{i=0}^{\infty}\frac{\lambda^{i}e^{-\lambda}}{i!}\frac{1}{1+i}\\
 & =\frac{1}{\lambda}\sum_{i=0}^{\infty}\frac{\lambda^{i+1}e^{-\lambda}}{(i+1)!}\\
 & =\frac{1}{\lambda}\sum_{i=0}^{\infty}\frac{\lambda^{i+1}e^{-\lambda}}{(i+1)!}\\
 & =\frac{1}{\lambda}\sum_{i=1}^{\infty}\frac{\lambda^{i}e^{-\lambda}}{i!}\\
 & =\frac{1}{\lambda}\left(1-e^{-\lambda}\right).
\end{align}
Thus, from (\ref{eq: bound on functional entropy differene}) and
(\ref{eq: Bound on approximated entropy})
\begin{equation}
\Ent(W)\leq\min_{\zeta>2}\left[\zeta+1+\frac{\log\zeta}{\zeta}\right]\leq4,
\end{equation}
as was required to be proved in order to complete the proof. 
\end{proof}

\subsection{Proof of Theorem \ref{thm: confidence bound for random entropy}}

Based on Propositions \ref{prop: Expected value of entropy} and \ref{prop: concentration of conditional entropy},
we may next prove Theorem \ref{thm: confidence bound for random entropy}. 
\begin{proof}[Proof of Theorem \ref{thm: confidence bound for random entropy}]
The bound $\log d_{A}\geq H(A_{S})$ follows since $A_{S}$ is supported
on a domain of size at most $d_{A}$. We thus next focus on the lower
bound. Let first us consider the case that $\eta\leq d_{A}d_{B}-d_{B}$
for which the concentration results of Proposition \ref{prop: concentration of conditional entropy}
is valid. We will afterwards separately handle the complementary case
$\eta>d_{A}d_{B}-d_{B}$. The conditions of Theorem \ref{thm: confidence bound for random entropy}
$d_{A}\geq d_{B}$ and $\eta\geq128d_{A}\log\left(\frac{128d_{A}}{\delta}\right)$
imply that $\eta\ge60d_{A}$ also holds, and so the qualifying conditions
of Propositions \ref{prop: Expected value of entropy} and \ref{prop: concentration of conditional entropy},
and thus their results, are valid. 

Recall that concentration inequality of Proposition \ref{prop: concentration of conditional entropy}
which is comprised of two terms. We require each term to be less than
$\delta/2$. The first term is $\frac{1}{2}\cdot e^{-\frac{\eta}{12}}$
and is bounded by $\delta/2$ by the assumption $\eta\geq12\log(\frac{1}{\delta})$.
For the second term we evaluate the necessary values for $r$ to achieve
an upper bound of $\delta/2$. To this end, let $q\eqdef\frac{r}{2\log(\eta/e)}$
so that the second term is 
\begin{equation}
\frac{1}{2}\exp\left\{ -\frac{\eta}{2d_{A}}\cdot h(q)+4\log(\eta)\right\} .
\end{equation}
For this term to be less than $\delta/2$ it suffices that 
\begin{equation}
h(q)\geq\frac{2d_{A}}{\eta}\left[4\log(\eta)+\log\frac{1}{\delta}\right].
\end{equation}
Assuming that $q\in[0,\frac{1}{2}]$, it holds that $\log(1+q)\geq q-q^{2}>\frac{q}{2}$
and then $h(q)\geq\frac{q^{2}}{2}.$ Thus, assuming $q\leq\frac{1}{2}$,
a sufficient condition is 
\begin{equation}
q=\sqrt{\frac{4d_{A}}{\eta}\left[4\log(\eta)+\log\frac{1}{\delta}\right]}
\end{equation}
as long as this value is less than $1/2$, which is satisfied if 
\begin{equation}
\frac{4d_{A}}{\eta}\left[4\log\left(\frac{\eta}{\delta}\right)\right]\leq\frac{1}{8},
\end{equation}
or, equivalently 
\begin{equation}
\frac{\eta/\delta}{\log(\eta/\delta)}\geq\frac{128d_{A}}{\delta}.
\end{equation}
By Lemma \ref{lem: x over log x} in Appendix \ref{sec:Auxiliary-results},
this condition holds if 
\begin{equation}
\eta\geq128d_{A}\log\left(\frac{128d_{A}}{\delta}\right).
\end{equation}
Note that the condition required for the first term, that is $\eta\geq12\log(\frac{1}{\delta})$
holds under this last condition, and thus unnecessary.

Then, the required $r$ for the second term to be less than $\delta/2$
is 
\begin{align}
r & =2\log(\eta/e)q\\
 & =2\log(\eta/e)\sqrt{\frac{4d_{A}}{\eta}\left[4\log(\eta)+\log\frac{1}{\delta}\right]}\\
 & \leq\sqrt{\frac{64d_{A}\log^{3}\left(\frac{\eta}{\delta}\right)}{\eta}}.
\end{align}
We thus deduce from the above and Propositions \ref{prop: Expected value of entropy}
and \ref{prop: concentration of conditional entropy}, that with probability
larger than $1-\delta$
\begin{align}
\left|H(A_{S})-\log d_{A}\right| & \leq\left|H(A_{S})-\E[H(A_{S})]\right|+\left|\E[H(A_{S})]-\log d_{A}\right|\\
 & \leq\sqrt{\frac{64d_{A}\log^{3}\left(\frac{\eta}{\delta}\right)}{\eta}}+\frac{8d_{A}}{\eta}+2C(d_{B})\\
 & \leq16\sqrt{\frac{d_{A}\log^{3}\left(\frac{\eta}{\delta}\right)}{\eta}}+2C(d_{B})\\
 & =16\sqrt{\frac{d_{A}\log^{3}\left(\frac{\eta}{\delta}\right)}{\eta}}+\frac{4\log(d_{B})}{\sqrt{d_{B}}}\\
 & \trre[\leq,a]20\sqrt{\frac{d_{A}\log^{3}\left(\frac{\eta}{\delta}\right)}{\eta}},\label{eq: high probability bound entropy proof}
\end{align}
where the last inequality holds since $\log(\eta)\geq\log(d_{B})$
and $\frac{d_{A}}{\eta}\geq\frac{1}{d_{B}}$. 

We next consider the case $\eta>d_{A}d_{B}-d_{B}$, and provide a
bound which holds with probability $1$. Note that in this case, it
holds that 
\begin{equation}
\tau\eqdef d_{B}-d_{A}d_{B}+\eta\geq0.
\end{equation}
Trivially, $H(A_{s})\leq\log d_{A}$, so we next focus on a lower
bound. Recalling the definition of $Z_{S}(i)$, the entropy is given
by 
\begin{equation}
H(A_{S})=-\sum_{i\in[d_{A}]}\frac{Z_{S}(i)}{\eta}\log\left(\frac{Z_{S}(i)}{\eta}\right).
\end{equation}
The entropy is a concave function of $\{Z_{S}(i)\}_{i\in[d_{A}]}$,
and so its minimal value under the constraints $\tau\leq Z_{S}(i)\leq d_{B}$
and $\sum_{i\in[d_{A}]}Z_{S}(i)=\eta$ is attained at the boundary
of the constraint set. Concretely, it must hold that $Z_{S}(i)\in\{\tau,d_{B}\}$
for all $i\in[d_{A}]$. However, since $\eta\geq d_{A}d_{B}-d_{B}$,
it can hold that $Z_{S}(i)=\tau$ only for a single $i\in[d_{A}]$.
Thus, the assignment with minimal entropy is given by $(Z_{S}(1),\ldots,Z_{S}(d_{A}))=(d_{B},d_{B},\ldots,d_{B},\tau)$
and the resulting minimal entropy implies that 
\begin{align}
H(A_{S}) & \geq-(d_{A}-1)\frac{d_{B}}{\eta}\log\left(\frac{d_{B}}{\eta}\right)-\frac{\tau}{\eta}\log\left(\frac{\tau}{\eta}\right)\\
 & \trre[\geq,a]-(d_{A}-1)\frac{d_{B}}{\eta}\log\left(\frac{d_{B}}{\eta}\right)\\
 & =\frac{(d_{A}d_{B}-d_{B})}{\eta}\log\left(\frac{\eta}{d_{B}}\right)\\
 & \trre[\geq,b]\frac{(d_{A}d_{B}-d_{B})}{d_{A}d_{B}}\log\left(\frac{d_{A}d_{B}-d_{B}}{d_{B}}\right)\\
 & =\left(1-\frac{1}{d_{A}}\right)\log\left(d_{A}-1\right)\\
 & \geq\log(d_{A})+\log\left(1-\frac{1}{d_{A}}\right)-\frac{\log(d_{A})}{d_{A}}\\
 & \trre[\geq,c]\log(d_{A})-2\frac{\log(d_{A})}{d_{A}},
\end{align}
where $(a)$ holds since $\log(\frac{\tau}{\eta})\leq\log(\frac{d_{B}}{\eta})\leq0$,
$(b)$ holds since $d_{A}d_{B}-d_{B}\leq\eta\leq d_{A}d_{B}$ is assumed,
$(c)$ holds since for $q\in[-\frac{1}{2},0]$, it holds that $\log(1+q)\geq q-q^{2}>\frac{q}{2}$.
The resulting bound, which holds with probability $1$, is only tighter
than the high probability derived in (\ref{eq: high probability bound entropy proof}).
Combining the results for the case $\eta\leq d_{A}d_{B}-d_{B}$ (a
high probability bound), and for $\eta>d_{A}d_{B}-d_{B}$ (a probability
$1$ bound) then establishes the claim of the theorem.
\end{proof}

\subsection{Proof of Corollary \ref{cor: confidence bound for MI}}

The proof of Corollary \ref{cor: confidence bound for MI} follows
rather directly from the confidence bound for entropy of Theorem \ref{thm: confidence bound for random entropy},
and the standard decomposition of mutual information to a sum of marginal
entropies minus the joint entropy. The formal proof is below.
\begin{proof}[Proof of Corollary \ref{cor: confidence bound for MI}]
By Theorem \ref{thm: confidence bound for random entropy} and the
union bound, both 
\begin{equation}
H(A_{S})\geq\log d_{A}-20\sqrt{\frac{d_{A}\log^{3}\left(\frac{\eta}{\delta}\right)}{\eta}}\label{eq: use of entropy concentration for MI - A}
\end{equation}
and 
\begin{equation}
H(B_{S})\geq\log d_{B}-20\sqrt{\frac{d_{B}\log^{3}\left(\frac{\eta}{\delta}\right)}{\eta}}\label{eq: use of entropy concentration for MI - B}
\end{equation}
hold with probability larger than $1-2\delta$. We then have that
\begin{align}
I(A_{S};B_{S}) & \trre[=,a]H(A_{S})+H(B_{S})-H(A_{S},B_{S})\\
 & =H(A_{S})+H(B_{S})-\log(\eta)\\
 & \trre[\geq,b]\log d_{A}-20\sqrt{\frac{d_{A}\log^{3}\left(\frac{\eta}{\delta}\right)}{\eta}}+\log d_{B}-20\sqrt{\frac{d_{B}\log^{3}\left(\frac{\eta}{\delta}\right)}{\eta}}-\log\eta\\
 & \trre[\geq,c]\log d_{A}-20\sqrt{\frac{d_{A}\log^{3}\left(\frac{\eta}{\delta}\right)}{\eta}}+\log d_{B}-20\sqrt{\frac{d_{B}\log^{3}\left(\frac{\eta}{\delta}\right)}{\eta}}-\log\eta\\
 & =\log\left(1+\overline{\rho}\right)-40\sqrt{\frac{d_{A}\log^{3}\left(\frac{\eta}{\delta}\right)}{\eta}},
\end{align}
where $(a)$ follows from the standard decomposition of mutual information
to a sum of entropies \cite[Chapter 2]{cover2012elements}, $(b)$
follows from (\ref{eq: use of entropy concentration for MI - A})
and (\ref{eq: use of entropy concentration for MI - B}), $(c)$ follows
from $\sqrt{d_{A}}+\sqrt{d_{B}}\leq\sqrt{2(d_{A}+d_{B})}\leq\sqrt{4d_{A}}$.
The proof is completed by changing notation from $2\delta\to\delta$. 
\end{proof}

\section{Proof of Theorem \ref{thm: confidence of MI of MVD} \label{sec:Proof-of-MVD-Theorem}}

We begin with the following lemma.
\begin{lem}
\label{lem: large domain size per letter in c}Let $R_{S}$ be a random
relation drawn according to the random relation model in Definition
\ref{def: random relation model}, with 
\begin{equation}
N\geq256d_{A}\overline{d}\log\left(\frac{128\overline{d}}{\delta}\right).
\end{equation}
where $\overline{d}=\max\{d_{A},d_{C}\}$. Let $R_{S,\ell}\eqdef\sigma_{C=\ell}(R_{S})$,
and let $N_{S}(\ell)=|R_{S,\ell}|$. Then
\begin{equation}
\min_{\ell\in[d_{C}]}N_{S}(\ell)\geq128d_{A}\log\left(\frac{128d_{A}}{\delta}\right)
\end{equation}
 with probability larger than $1-\delta$. 
\end{lem}

\begin{proof}
Let us first consider a specific $\ell\in[d_{C}]$. Then, $N_{S}(\ell)\sim\text{Hypergeometric}(d_{A}d_{B}d_{C},d_{A}d_{B},N)$
where $d_{A}d_{B}d_{C}$ is the population size, $d_{A}d_{B}$ is
the number of success states in the population, to wit, the tuples
for which $C=\ell$, and $N$ is the number of draws. It evidently
holds that $\E[N_{S}(\ell)]=\frac{N}{d_{C}}$. We next show that $N_{S}(\ell)$
is larger than $1/2$ of its expected value with high probability.
Concretely,
\begin{align}
\P\left[N_{S}(\ell)\leq\frac{1}{2}\E[N_{S}(\ell)]\right] & =\P\left[N_{S}(\ell)\leq\frac{N}{2d_{C}}\right]\\
 & =\P\left[N-N_{S}(\ell)\geq N\left(1-\frac{1}{2d_{C}}\right)\right]\\
 & =\P\left[N-N_{S}(\ell)\geq\E\left[N-N_{S}(\ell)\right]+\frac{N}{2d_{C}}\right]\\
 & \leq\exp\left[-\frac{N}{2d_{C}^{2}}\right],
\end{align}
where the inequality follows from the following reasoning: Due to
the symmetry of the hypergeometric distribution, it holds that 
\begin{equation}
N-N_{S}(\ell)\sim\text{Hypergeometric}(d_{A}d_{B}d_{C},d_{A}d_{B}d_{C}-d_{A}d_{B},N),
\end{equation}
and so Serfling's inequality (Lemma \ref{lem: Serfling's inequality}
in Appendix \ref{sec:Auxiliary-results}) used with $\epsilon=\frac{N}{2d_{C}}$
directly results the stated bound. If we choose 
\begin{equation}
N\geq2d_{C}^{2}\cdot\log\frac{d_{C}}{\delta}
\end{equation}
then under the assumptions of the lemma
\begin{equation}
N_{S}(\ell)\geq\frac{1}{2}\E[N_{S}(\ell)]=\frac{N}{2d_{C}}\geq128d_{A}\log\left(\frac{128d_{A}}{\delta}\right)
\end{equation}
holds with probability $1-\frac{\delta}{d_{C}}$. Taking a union bound
over all $\ell\in[d_{C}]$ assures that this holds uniformly over
$[d_{C}]$. 
\end{proof}
We may now prove Theorem \ref{thm: confidence of MI of MVD}.
\begin{proof}[Proof of Theorem \ref{thm: confidence of MI of MVD}]
 We use the bound on the (regular) mutual information for the $d_{C}=1$
case (Corollary \ref{cor: confidence bound for MI}) for each $\ell\in[d_{C}]$
separately. Lemma \ref{lem: large domain size per letter in c} assures
that with high probability, the number of sample points satisfies
\begin{equation}
\min_{\ell\in[d_{C}]}N_{S}(\ell)\geq128d_{A}\log\left(\frac{128d_{A}}{\delta}\right)\label{eq: first event - large number of points for all l}
\end{equation}
with probability $1-\delta$. For each $\ell\in[d_{C}]$. Then, applying
the result of Corollary \ref{cor: confidence bound for MI} with $\delta$
replaced by $\frac{\delta}{d_{C}}$ and then taking a union bound
over $\ell\in[d_{C}]$ assures that with probability $1-\delta$ it
holds that 
\begin{equation}
I(A_{S};B_{S}\mid C_{S}=\ell)\geq\log\left(1+\overline{\rho}_{S}(\ell)\right)-40\sqrt{\frac{d_{A}\log^{3}\left(\frac{2N_{S}(\ell)d_{C}}{\delta}\right)}{N_{S}(\ell)}}\,\,\,\,\,\,\,\,\text{for all }\ell\in[d_{C}]\label{eq: second event - good MI for all l}
\end{equation}
where
\begin{equation}
\overline{\rho}_{S}(\ell)=\frac{d_{A}\cdot d_{B}}{N_{S}(\ell)}-1.
\end{equation}
Moreover, considering $F\equiv(A,B)$ as a single joint random variable
with domain $[d_{f}]=[d_{A}d_{B}]$, the random draw of the set $S$
of size $N$ from $[d_{A}]\times[d_{B}]\times[d_{C}]$ can be considered
as a draw from $[d_{C}]\times[d_{f}]$. Using Theorem \ref{thm: confidence bound for random entropy}
then assures that 
\begin{equation}
0\leq\log d_{C}-H(C_{S})\leq20\sqrt{\frac{\max\{d_{C},d_{A}d_{B}\}\log^{3}\left(\frac{N}{\delta}\right)}{N}}\label{eq: third event - good entropy for c}
\end{equation}
holds with probability larger than $1-\delta$, as long as 
\begin{equation}
N\geq128\cdot\max\{d_{C},d_{A}d_{B}\}\log\left(\frac{128\cdot\max\{d_{C},d_{A}d_{B}\}}{\delta}\right).
\end{equation}
It can be verified that this indeed condition holds under the assumption
on $N$ in (\ref{eq: assumption on N in MVD theorem}) made in the
statement of Theorem \ref{thm: confidence of MI of MVD}. 

We next assume that the events in (\ref{eq: first event - large number of points for all l}),
(\ref{eq: second event - good MI for all l}) and (\ref{eq: third event - good entropy for c})
simultaneously hold. By the union bound, this occurs with probability
larger than $1-3\delta$. Note that $\P[C_{s}=\ell]=\frac{N_{s}(\ell)}{N}$
holds, and let $s$ be a given set which belongs to the high probability
set. Then, the relative number of spurious tuples is upper bounded
as 
\begin{align}
 & \log\left[1+\rho(s,\phi)\right]\nonumber \\
 & =\log\left[\frac{\sum_{\ell\in[d_{C}]}|\proj{A}(R_{\ell})|\cdot|\proj{B}(R_{\ell})|}{\sum_{\ell\in[d_{C}]}\proj{A,B}(R_{\ell})}\right]\\
 & \leq\log\left[\frac{d_{C}d_{A}d_{B}}{\sum_{\ell\in[d_{C}]}N_{s}(\ell)}\right]\\
 & =\log d_{C}+\log\left[\frac{1}{\sum_{\ell\in[d_{C}]}\frac{1}{1+\overline{\rho}_{s}(\ell)}}\right]\\
 & =\log d_{C}+\left(\sum_{\ell\in[d_{C}]}\P[C_{s}=\ell]\right)\cdot\log\left[\frac{\sum_{\ell\in[d_{C}]}\P[C_{s}=\ell]}{\sum_{\ell\in[d_{C}]}\frac{1}{1+\overline{\rho}_{s}(\ell)}}\right]\\
 & \trre[\leq,a]\log d_{C}+\sum_{\ell\in[d_{C}]}\P[C_{s}=\ell]\log\left[\P[C_{s}=\ell]\cdot(1+\overline{\rho}_{s}(\ell))\right]\\
 & =\log d_{C}-H(C_{s})+\sum_{\ell\in[d_{C}]}\P[C_{s}=\ell]\log\left[1+\overline{\rho}_{s}(\ell)\right]\\
 & \trre[\leq,b]20\sqrt{\frac{\max\{d_{C},d_{A}d_{B}\}\cdot\log^{3}\left(\frac{N}{\delta}\right)}{N}}\nonumber \\
 & \hphantom{=====}+\sum_{\ell\in[d_{C}]}\P[C_{s}=\ell]\left(I(A_{S};B_{S}\mid C_{S}=\ell)+40\sqrt{\frac{d_{A}\log^{3}\left(\frac{2N_{S}(\ell)d_{C}}{\delta}\right)}{N_{S}(\ell)}}\right)\\
 & \trre[=,c]20\sqrt{\frac{\max\{d_{C},d_{A}d_{B}\}\cdot\log^{3}\left(\frac{N}{\delta}\right)}{N}}+I(A_{s};B_{s}\mid C_{s})+\sum_{\ell\in[d_{C}]}\P[C_{s}=\ell]40\sqrt{\frac{d_{A}\log^{3}\left(\frac{2N_{S}(\ell)d_{C}}{\delta}\right)}{N_{S}(\ell)}}\\
 & \trre[\leq,d]20\sqrt{\frac{\max\{d_{C},d_{A}d_{B}\}\cdot\log^{3}\left(\frac{N}{\delta}\right)}{N}}+I(A_{s};B_{s}\mid C_{s})+40\sqrt{\frac{d_{A}d_{C}\log^{3}\left(\frac{2Nd_{C}}{\delta}\right)}{N}}\\
 & \trre[\leq,e]I(A_{s};B_{s}\mid C_{s})+60\sqrt{\frac{d_{A}\overline{d}\log^{3}\left(\frac{2Nd_{C}}{\delta}\right)}{N}},\label{eq: derivation of log-spurious tuples with conditional MI}
\end{align}
where $(a)$ follows from the log sum inequality (Lemma \ref{lem: log-sum inequality}
in Appendix \ref{sec:Auxiliary-results}), $(b)$ follows from the
assumption that the events in (\ref{eq: second event - good MI for all l})
and (\ref{eq: third event - good entropy for c}) hold, $(c)$ from
the standard definition of conditional mutual information as an weighted
average of mutual information terms
\begin{equation}
I(A_{s};B_{s}\mid C_{s})\eqdef\sum_{\ell\in[d_{C}]}\P[C_{s}=\ell]I(A_{s};B_{s}\mid C_{s}=\ell),
\end{equation}
and $(d)$ follows from the bound 
\begin{align}
 & \sum_{\ell\in[d_{C}]}\P[C_{s}=\ell]40\sqrt{\frac{d_{A}\log^{3}\left(\frac{2N_{S}(\ell)d_{C}}{\delta}\right)}{N_{S}(\ell)}}\nonumber \\
 & \leq40\sqrt{d_{A}\log^{3}\left(\frac{2Nd_{C}}{\delta}\right)}\cdot\sum_{\ell\in[d_{C}]}\P[C_{s}=\ell]\sqrt{\frac{1}{N_{S}(\ell)}}\\
 & =\frac{40\sqrt{d_{A}\log^{3}\left(\frac{2Nd_{C}}{\delta}\right)}}{N}\cdot\sum_{\ell\in[d_{C}]}\sqrt{N_{s}(\ell)}\\
 & \leq40\sqrt{\frac{d_{A}d_{C}\log^{3}\left(\frac{2Nd_{C}}{\delta}\right)}{N}},
\end{align}
where here the first inequality utilizes $N_{S}(\ell)\leq N$ for
the logarithmic term, and the second inequality uses the fact that
$\sum_{\ell\in[d_{C}]}\sqrt{N_{s}(\ell)}$ under the constraint $\sum_{\ell\in[d_{C}]}N_{s}(\ell)=N$
is maximized for $N_{s}(\ell)=\frac{N}{d_{C}}$ for all $\ell$. Finally,
in passage $(e)$ of (\ref{eq: derivation of log-spurious tuples with conditional MI})
we slightly loosen the bound using $\overline{d}=\max\{d_{A},d_{C}\}\geq\max\{d_{A},d_{B}\}$.
The proof is completed by replacing $3\delta\to\delta$.
\end{proof}

\section{Auxiliary results \label{sec:Auxiliary-results}}
\begin{lem}[{An LSI for asymmetric Bernoulli random variables \cite[Chapter 5]{boucheron2013concentration}}]
\label{lem: LSI for asymmetric Bernoulli} For any function $g\colon\{-1,1\}^{d}\to\mathbb{R}$
and i.i.d. random variables $\P[R(j)=1]=p=1-\P[R(j)=-1]$ let 
\begin{equation}
\mathcal{E}(g)\eqdef p(1-p)\E\left[\sum_{j\in[d]}\left(g\left(R(1),\ldots,R(j),\ldots,R(d)\right)-g\left(R(1),\ldots,-R(j),\ldots,R(d)\right)\right)^{2}\right]\label{eq: Efron-Stein variance}
\end{equation}
be the Efron-Stein variance of $g$. Then, the LSI \cite[Thms. 5.1 and 5.2]{boucheron2013concentration}
states that 
\begin{equation}
\Ent(g^{2})\leq\frac{1}{1-2p}\log\left(\frac{1-p}{p}\right)\cdot\mathcal{E}(g).
\end{equation}
\end{lem}

\begin{lem}[Relative Chernoff's bound for a binomial random variable]
\label{lem: relative Chernoff for binomial}For $B_{i}\sim\text{Bernoulli}(p)$
i.i.d., $i\in[n]$, it holds that for any $\xi\in[0,1]$,
\begin{equation}
\P\left[\left|\frac{1}{n}\sum_{i=1}^{n}B_{i}-p\right|\geq\xi p\right]\leq2e^{-\frac{\xi^{2}pn}{3}}.\label{eq: relative chernoff}
\end{equation}
\label{lem: continuity of self-information function}Let $g(t)\eqdef-t\log(t)$.
Then for any $s,t\in[0,1]$ it holds that 
\begin{equation}
|g(t)-g(s)|\leq2g(|s-t|).
\end{equation}
\end{lem}

\begin{proof}
The function $g(t)$ is concave, positive, and satisfies $g(0)=g(1)=0$.
Assume w.l.o.g. that $0\leq t<s\leq1$. The proof for the case $0\leq s-t\leq\frac{1}{2}$
was explained in \cite[Thm. 17.3.3]{cover2012elements}: The chord
of the function $g(t)$ from $t$ to $s$ has maximum absolute slope
either at the extremes -- either at $t=0$ or $t=1-(s-t)$. Then,
\begin{equation}
|g(t)-g(s)|\leq\max\left\{ g(s-t),g(1-(s-t))\right\} =g(s-t)
\end{equation}
where the last inequality is by the assumption $0\leq s-t\leq\frac{1}{2}$.
Otherwise, if $\frac{1}{2}\leq s-t\leq1$ then it must hold that $0\leq t\leq\frac{1}{2}\leq s\leq1$.
Trivially, 
\begin{equation}
\left|g(s)-g(t)\right|\leq\left|g(s)\right|+\left|g(t)\right|\leq1\leq2(s-t)\leq2(s-t)\log\frac{1}{(s-t)}=2g(s-t).
\end{equation}
\end{proof}
\begin{lem}[{Chernoff's bound for a Poisson random variables \cite[Thm. 5.4]{mitzenmacher2017probability}}]
\label{lem: Chernoff for Poisson rvs}For $X\sim\text{Poisson}(\lambda)$
it holds that
\begin{equation}
\P\left[X\geq\alpha\E[X]\right]\leq e^{-\lambda}\left(\frac{e}{\alpha}\right)^{\alpha\lambda}\leq\left(\frac{e}{\alpha}\right)^{\alpha\lambda}=e^{-\alpha\lambda\log(\alpha/e)}\leq e^{-\alpha\lambda}\label{eq: Poisson Chernoff general-1}
\end{equation}
for all $\alpha>3e\approx8.15$. 
\end{lem}

For the next two lemmas, we consider a function $f\colon\mathbb{N}\to(0,\infty)$,
and denote its derivative by
\begin{equation}
Df(x)\eqdef f(x+1)-f(x).
\end{equation}
\begin{lem}[Poisson concentration of Lipschitz functions \cite{bobkov1998modified,kontoyiannis2006measure}]
\label{lem: Poisson concentration}Let $W\sim\text{Poisson}(\lambda)$,
and assume that $f$ is $1$-Lipschitz, that is $|Df(w)|\leq1$ for
all $w\in\mathbb{N}_{+}$. Then, for any $t>0$
\begin{equation}
\P\left[f(W)-\E[f(W)]>t\right]\leq\exp\left[-\frac{t}{4}\log\left(1+\frac{t}{2\lambda}\right)\right].\label{eq: Poisson concentration}
\end{equation}
\end{lem}

\begin{lem}[{Poisson LSI \cite[Thm. 6.17]{boucheron2013concentration}}]
\label{lem: Poisson LSI}Let $W\sim\text{Poisson}(\lambda)$. Then,
\begin{equation}
\Ent[f(W)]\leq\lambda\cdot\E\left[\frac{\left|Df(W)\right|^{2}}{f(W)}\right].
\end{equation}
\end{lem}

\begin{lem}
\label{lem: x over log x}If $x\geq y\log(y)$ then $\frac{x}{\log x}\geq y$. 
\end{lem}

\begin{proof}
As $y\geq e$ then both $\log(y)=1>0$ and $\log(\log(y))\geq0$.
Choosing $x=y\log(y)$ it holds that 
\begin{equation}
\frac{x}{\log x}=\frac{y\log(y)}{\log(y\log(y))}=y\cdot\left[1+\frac{\log(y)}{\log(y)+\log(\log(y))}\right]\geq y.
\end{equation}
\end{proof}

\paragraph*{Hypergeometric distribution}

We denote a random variable distributed according to a hypergeometric
distribution as $Y\sim\text{Hypergeometric}(L,M,\ell)$ where $L$
is the population size, $M$ is the number of success states in the
population, and $\ell$ is the number of draws from the population.
The mean of the distribution is $\E[Y]=\ell\cdot\frac{M}{L}$.
\begin{lem}[Serfling's inequality (simplified) \cite{serfling1974probability},
see also \cite{greene2017exponential}]
\label{lem: Serfling's inequality} Let $Y\sim\text{Hypergeometric}(L,K,\ell)$.
Then, for any $\epsilon>0$ and $1\leq\ell\leq L$
\begin{equation}
\P\left[Y-\E[Y]\geq\epsilon\right]\leq\exp\left[\frac{-2\epsilon^{2}}{\ell(1-\frac{\ell-1}{L})}\right]\leq\exp\left[\frac{-2\epsilon^{2}}{\ell}\right].\label{eq: Serflings inequality}
\end{equation}
\end{lem}

\begin{lem}[{The log sum inequality \cite[Thm. 2.7.1]{cover2012elements}}]
\label{lem: log-sum inequality}For nonnegative numbers $\{a_{i}\}_{i\in[n]}$
and $\{b_{i}\}_{i\in[n]}$ it holds that 
\begin{equation}
\sum_{i=1}^{n}a_{i}\log\frac{\sum_{i=1}^{n}a_{i}}{\sum_{i=1}^{n}b_{i}}\leq\sum_{i=1}^{n}a_{i}\log\left(\frac{a_{i}}{b_{i}}\right).\label{eq: log-sum inequality}
\end{equation}
\end{lem}

    \else
    \fi

\end{document}